\documentclass[12pt,reqno]{article}

\usepackage[authoryear]{natbib}
\usepackage{mathrsfs}
\usepackage{ragged2e}

\usepackage{natbib}
\usepackage[Symbol]{upgreek}
\usepackage[T1]{fontenc}
\usepackage[latin9]{inputenc}
\usepackage[letterpaper]{geometry}
\usepackage{color}
\usepackage{amsmath}
\usepackage{amsthm}
\usepackage{amssymb}
\usepackage{setspace}
\usepackage[authoryear]{natbib}
\usepackage[unicode=true,pdfusetitle,
 bookmarks=true,bookmarksnumbered=false,bookmarksopen=false,
 breaklinks=false,pdfborder={0 0 0},pdfborderstyle={},backref=false,colorlinks=true]
 {hyperref}
\hypersetup{
 linkcolor=magenta,citecolor=blue}
 
\usepackage{graphicx}
\graphicspath{ {./graphs/} }
\usepackage{babel}

\makeatother
\usepackage{array}
\usepackage{booktabs}
\usepackage{multirow}

\numberwithin{equation}{section}

\theoremstyle{remark}
\newtheorem{rem}{\protect\remarkname}[section]
\theoremstyle{definition}
\newtheorem{defn}{\protect\definitionname}
\theoremstyle{plain}
\newtheorem{thm}{\protect\theoremname}
\theoremstyle{plain}
\newtheorem{prop}{\protect\propositionname}
\theoremstyle{definition}
\newtheorem{example}{\protect\examplename}
\theoremstyle{definition}

\newtheorem{lem}{\protect\lemmaname}[section]

\newtheorem{assumption}{Assumption}

\providecommand{\conditionname}{Condition}
\providecommand{\definitionname}{Definition}
\providecommand{\examplename}{Example}
\providecommand{\lemmaname}{Lemma}
\providecommand{\propositionname}{Proposition}
\providecommand{\remarkname}{Remark}
\providecommand{\theoremname}{Theorem}

\newcommand{\continuation}{??}

 \newcommand{\sieved}{d_p}
 \newcommand{\sievedstar}{d_p^*}

\begin{document}

\sloppy

\newgeometry{verbose,tmargin=1in,bmargin=1in,lmargin=1.25in,rmargin=1.25in,footskip=1cm}

\title{Leave No One Undermined:  Policy Targeting with Regret Aversion{\thanks{We would like to thank Don Andrews, Isaiah Andrews, Tim Armstrong, Yong Cai, Kevin Chen,  Xiaohong Chen, Harold Chiang, Tim Christensen, Bruce Hansen, Marc Henry, Lihua Lei, Yuan Liao, Doug Miller, Francesca Molinari, Jos\'e Luis Montiel Olea, Mikkel Plagborg-M{\o}ller, Jack Porter, Sophie Sun, Chris Taber, Max Tabord-Meehan,   Aleksey Tetenov, Davide Viviano, Ed Vytlacil, Kohei Yata, and the participants at the Bravo/JEA/SNSF Workshop on
``Using Data to Make Decisions'' (Brown), Madison, Yale, Greater NY Metropolitan Area Econometrics Colloquium (Rochester), SEA 2025 (Tampa), and SETA 2025 (Macau) for helpful comments.
Chen Qiu gratefully acknowledges financial support from NSF grant (number SES-2315600).
}}}
\author{Toru Kitagawa\thanks{ Department of Economics, Brown University. Email: toru\_kitagawa@brown.edu} \and Sokbae Lee\thanks{Department of Economics, Columbia University. Email: sl3841@columbia.edu} \and Chen Qiu\thanks{Department of Economics, Cornell University. Email: cq62@cornell.edu}}
\date{April 2026}
\maketitle

\begin{abstract}
While the importance of personalized policymaking is widely recognized, fully personalized implementation remains rare in practice, often due to legal, fairness or cost concerns. We study the problem of policy targeting for a regret-averse planner when training data gives a rich set of observables while the assignment rules can only depend on its subset. Our regret-averse criterion reflects a planner's concern about regret inequality across the population. This, in general, leads to a fractional optimal rule due to treatment effect heterogeneity beyond the average treatment effects conditional on the subset of  observables. We propose a debiased empirical risk minimization approach to learn the optimal rule from data and establish favorable, new upper and lower bounds for the excess risk, indicating a convergence rate of $1/n$ and asymptotic efficiency in certain cases. We apply our approach to the National JTPA Study and the International Stroke Trial.
\end{abstract}

\newpage
\section{Introduction}

Personalization and policy targeting gained wide recognition in social and medical sciences. Yet, practice of complete personalization is rare. For example, as noted by \cite{manski2022patient}, President's Council of Advisors on Science and Technology (\citeyear{PCAST}) defines  ``personalized medicine'' as:

\emph{``...the tailoring of medical treatment to the specific characteristics of each patient. In an operational sense, however, personalized medicine does not literally mean the creation of drugs or medical devices that are unique to a patient...''}

Indeed, even though a rich set of observables $X$ are present in the data, policy makers (PM) often only form allocation rules based on a few covariates $W$, a subset of and not as rich as $X$. If treatment responses are significantly different in $X$ even after conditioning on $W$, how should the PM design its optimal treatment policy?

For instance, consider the mentoring program studied by \cite*{resnjanskij2024can} that aims to improve the labor market prospects for disadvantaged  adolescents in Germany. \cite{resnjanskij2024can} collected a rich set of covariate information on the participants of their randomized control trial. Their study highlights drastically different treatment responses depending on the social economic status (SES, classified based on answers to questions like how many books are at home, whether the adolescent is a first-generation migrant or has a single parent, etc.) of the adolescents: Low SES adolescents respond positively to the mentoring program while  higher SES adolescents respond negatively (see left panel of Figure \ref{fig:mentor}).

\begin{figure}[h!]
    \centering
    \includegraphics[width=0.46\linewidth]{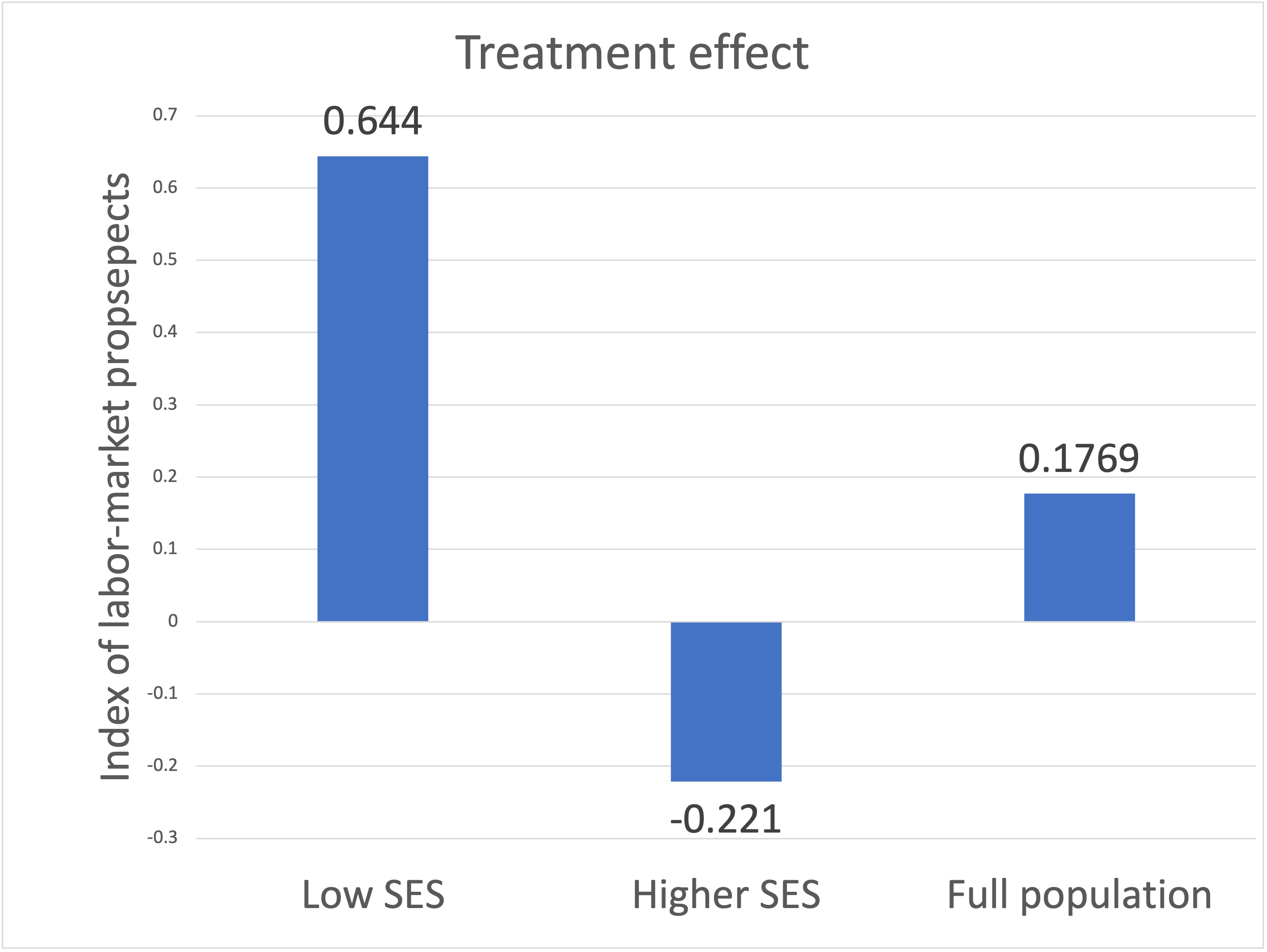}
    \includegraphics[width=0.46\linewidth]{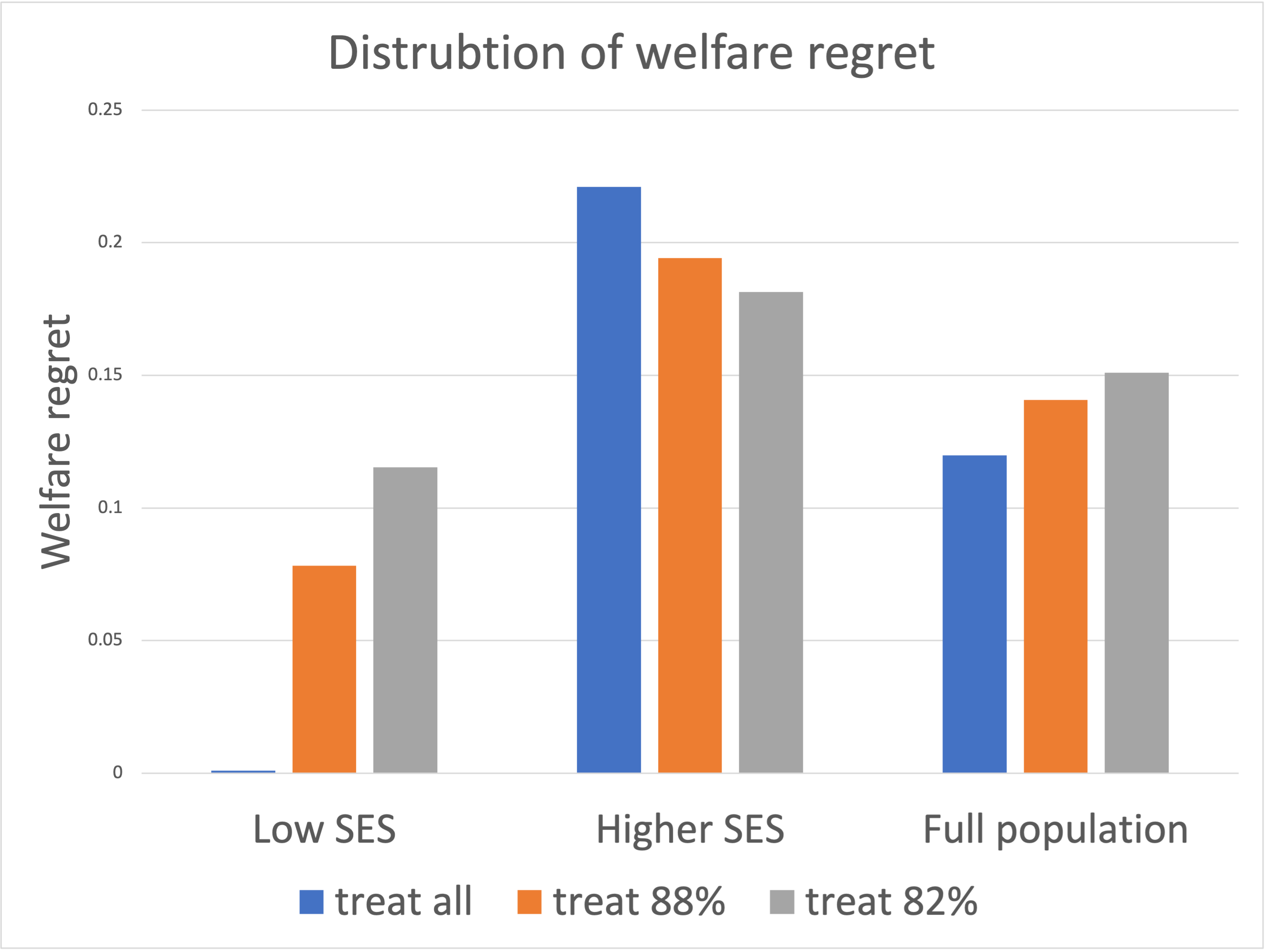}
    \caption{Left: Treatment effects taken from \citet[][Column 4, Table 2]{resnjanskij2024can}; Right: Our calculation of welfare regrets for three hypothetical decision rules (note for the rule that treats all, regret of the low SES group is zero).}
    \label{fig:mentor}
\end{figure}

Motivated by their heterogeneity analyses,  imagine a PM who decides what fraction of the population should be treated. However, suppose the PM cannot condition their decision on SES, because including it may be perceived as  illegal or unfair, or measuring it precisely at the decision making stage is too costly. This corresponds to a scenario in which $X$ refers to SES and there is no $W$. Since the full population average  treatment effect is positive, the common approach that aims to maximize the average welfare (e.g., \citealt{manski2004statistical,kitagawa2018should,AW20,MT17}) would inform the PM to treat all adolescents, even though higher SES adolescents lose from the program. Is the common approach still reasonable?  And if not, what alternative decision criterion can reflect the heterogeneous treatment responses along SES?

In this paper, we answer these questions by studying the problem of policy targeting for a regret averse PM. The PM aims to find an optimal allocation rule $\delta$ that maps subset characteristics $W$ to  $[0,1]$ indicating treatment fractions, with a loss function $L(\delta):=\mathbb{E}\left[Reg^{\alpha}(X,\delta)\right]$,
where $Reg(x,\delta)$ refers to the welfare regret of group $X=x$ when applied with $\delta$, i.e., the efficiency loss compared to its best achievable welfare,     $\alpha\geq1$ and $\mathbb{E}[\cdot]$ denotes expectation with respect to the marginal distribution of $X$. The case of $\alpha>1$ is what we call regret aversion, while $\alpha=1$ corresponds to the  common approach, i.e., minimizing mean welfare regret  (equivalent to  maximizing mean welfare).

We link the regret-averse loss to the PM's aversion to unequal distributions of regret among groups defined by covariates that are not allowed to use as input to the treatment rule.\footnote{\cite*{liang2021algorithm} provide microfoundation for preference types that yield a preference for fairness, which in turn can be interpreted as inequality aversion. \cite*{auerbach2024testing,liu2024inference} conduct statistical inference based on the theoretical characterization of the fairness-accuracy frontier in \cite{liang2021algorithm}.}  For example, viewing the index of labor market prospects as a proxy for welfare, the right panel of Figure \ref{fig:mentor} plots distributions of welfare regrets  for three hypothetical rules (treat all, treat 88\%, and treat 82\%).  Treating all benefits and induces no regret for low SES adolescents, but hurts and generates a high regret for higher SES adolescents.\footnote{As explained by \cite{resnjanskij2024can}, mentoring may actaully crowd out more useful inputs offered by higher SES families, such as parental attachment or participation in other useful activities.} The other two fractional rules enjoy a more equal distribution of regret between the two groups, although they also have higher mean regrets. If $\alpha = 1$, PM displays no aversion to inequality of regrets and evaluates rules solely based on the mean, indicating treating all is indeed optimal. However, if $\alpha>1$, the policymaker dislikes and penalizes rules with higher inequality. We formally quantify such aversion  via an Atkinson inequality index applied to regret (calculated in Table \ref{tab:mentor} according to \eqref{eq:inequality}).\footnote{\cite{atkinson1970measurement}'s measure of inequality is at an individual level rather than group level. We
may interpret $X$ in our setup as individuals in his framework.} The higher the value of  $\alpha$, the more the PM dislikes inequality. In fact, treating  88\% (82\%) of the population is optimal if $\alpha=2$ (3).  As one of the fundamental goals of sustainable development policy is to  ``reduce the inequalities (\ldots) that (\ldots) undermine the potential of individuals''\footnote{United Nations Sustainable Development Group, \url{https://unsdg.un.org/2030-agenda/universal-values/leave-no-one-behind}.}, and a wide literature in  program evaluation \citep*[e.g.,][]{angrist2012benefits,angrist2022marginal,gray2023long} documents significant subgroup treatment heterogeneity along gender, race and other costly or sensitive attributes, our approach develops a new perspective in the current policy learning literature. 

\begin{table}[htbp]
\centering
\caption{Regret and inequality of various allocation rules in \cite{resnjanskij2024can}}
\begin{tabular}{lcccccc}
\toprule
& \multicolumn{3}{c}{Regret} & \multicolumn{3}{c}{Atkinson inequality index} \\
\cmidrule(lr){2-4} \cmidrule(lr){5-7}
Allocation rule & Low SES & Higher SES & Mean & $\alpha=1$ & $\alpha=2$ & $\alpha=3$ \\
\midrule
Treat all  & 0 & 0.22 & 0.12 & 0 & 0.36 & 0.51 \\
Treat 88\% & 0.08 & 0.19 & 0.14 & 0 & 0.08 & 0.14 \\
Treat 82\% & 0.12 & 0.18 & 0.15 & 0 & 0.02 & 0.05 \\
\bottomrule
\end{tabular}\label{tab:mentor}
\end{table}

In general, both $W$ and $X$ can be continuous and multidimensional. We show the key insights persist:  If $\alpha=1$, the optimal rule is to treat everyone or no one in the same $W$ group,  depending on the  sign of the corresponding CATE($W$), even if significantly  heterogeneous treatment responses remains  beyond $W$. Whenever $\alpha>1$,  the optimal rule is fractional, if the treatment effect heterogeneity is severe enough to alter the sign of  CATE($X$) within the same $W$.\footnote{We stress that the mechanism for fractional rules is different from that in \cite{kitagawa2022treatment}, in which fractional rules arise due to sampling uncertainty. Fractional rules also show up with  partially identified welfare with \citep{Manski2000,manski2005social,manski2007identification, Manski2007} or without \citep{ stoye2012minimax,yata2021,manski2022identification,montielolea2023decision, kitagawa2023treatment} true knowledge of the identified set, with nonlinear welfare \citep{manski2007admissible,manski20092009}, when the decision maker targets a functional of the outcome distribution that is not quasi-convex \citep*{kock2022functional,kock2023treatment}, or when agents respond with strategic behavior \citep{munro2020learning}.}  This insight carries over analogously even with a capacity constraint.

The preceding example  ignores the statistical uncertainty in the estimation of heterogeneous treatment effects. Focusing on $\alpha=2$ for tractability, we  propose an empirical squared regret minimization approach with debiasing and cross fitting to properly account for the estimation uncertainty from training data. Our procedure accommodates a wide range of black-box machine learning methods for estimating the heterogeneous treatment effects. We develop new asymptotic upper and lower bounds for the excess risk of our proposal. 
In the case of a correctly specified linear sieve policy class, our procedure achieves a fast convergence rate of  $O(1/n)$  and is in fact asymptotically efficient. 
In terms of computation, our squared regret approach is attractive due to the weighted least squares structure of the objective function. It will be straightforward to accommodate policy classes with convex constraints,  compared to the mean regret approach (which often relies on maximum score type optimization).   


We further illustrate the value of our approach using two real datasets. For the National JTPA Study that measured the benefit and cost of employment and training programs, 
we consider a PM  who designs treatment policies based on  pre-program years of education and earnings only, even though a wider range of characteristics like gender, race and marital status are available in the data. For the International Stroke Trial data that assessed
the effect of aspirin treatment for patients with acute ischemic stroke, we considered a hypothetical scenario  in which a doctor determines whether a patient should be treated with aspirin based on their age only, even though the training data contains many covariates of the patients, including their  demographic and medical history. 
In both datasets, the estimated optimal policy fractions with our squared regret approach reveal considerable treatment effect heterogeneity  for some population subgroups, which can lead to significant regret inequality  should a singleton policy be applied instead. 
These exercises suggest that our squared regret approach reveals  additional important information that may not be assessed with the mean regret paradigm alone. 


The treatment choice literature has become an area of active research since the pioneering works of \cite{Manski2000, manski2002treatment,manski2004statistical} and \cite{Dehejia2005}. When the policy maker cannot differentiate individuals based on observable characteristics, finite and asymptotic results are developed by  \cite{schlag2006eleven}, \cite{stoye2009minimax},  \cite{HiranoPorter2009}, \cite{tetenov2012statistical}, \cite{masten2023minimax}, and \cite{chen2024note}.  \cite{manski2014quantile} and \cite*{guggenberger2024minimax}  in point-identified situations, and by  \cite{Manski2000,manski2005social,manski2007identification, Manski2007, manski20092009},  \cite{stoye2012minimax}, \cite{christensen2020}, \cite{ishihara2021}, \cite{yata2021}, \cite{manski2022identification}, \cite{ishihara2023bandwidth} and \cite{montielolea2023decision} in partially-identified settings. When the policy maker is able to condition on individual characteristics, studies on personalization and policy targeting include \cite{manski2004statistical}, \cite{BhattacharyaDupas2012}, \cite{kitagawa2018should, KT21}, \cite{MT17}, \cite{AW20}, \cite{sun2021empirical}, \cite{adjaho2022externally}, \cite{han2023optimal}, \cite{cui2023individualized},
\cite{terschuur2024locally}, \cite{viviano2024policy} and \cite{viviano2024fair}, among others, for analyses in different settings. 

Our squared regret criterion coincides with a quadratic surrogate criterion for a cost-sensitive classification problem to predict the sign of CATE($X$) with cost being squared CATE($X$). See \cite{Zhang_2004}, \cite*{Bartlett_et_al_2006}, and references therein. Note, however, that our approach fundamentally differs from classification with the quadratic surrogate since our approach takes the squared regret as the ultimate objective function to minimize rather than a surrogate for the binary classification loss. As a result, our analysis can allow constrained $W$-individualized rules without raising the inconsistency issue of the constrained classification studied in \cite*{KST21}.

The rest of the paper is organized as follows: Section \ref{sec:setup} sets the stage, motivates our regret-averse loss function via inequality aversion and  studies the population optimal allocation rule. Section \ref{sec:proposal} presents our main proposal. Section \ref{sec:stat} develops the statistical performance guarantee. Section \ref{sec:capacity} discusses the case with capacity constraint. Empirical applications are in Section \ref{sec:emp.app}. Additional proofs, lemmas and technical results  are reserved in the Appendix and Online Supplement.





\section{Setup}\label{sec:setup}


Consider a policymaker who has access to a random sample of size $n$:
$\ensuremath{Z^{n}:=\left\{ Z_i\right\} {}_{i=1}^{n}\in\mathcal{Z}^{n}}$,
where $Z_{i}:=\{X_{i},D_{i},Y_{i}\}$, $X_{i}\in\mathcal{X}\subseteq\mathbb{R}^{d_X}$ is
the observed pre-treatment characteristics (covariates) of unit $i$,
e.g., their gender, race, pre-treatment education level, etc., $D_{i}\in\left\{ 0,1\right\} $
is the binary treatment indicator ($D_{i}=1$ means unit $i$ is under
treatment and $D_{i}=0$ means under control), and $Y_{i}\in\mathbb{R}$
is the observed outcome of interest of unit $i$, generated as
\begin{equation}\label{eq:observed.outcome}
Y_{i}=D_{i}Y_{i}(1)+(1-D_{i})Y_{i}(0),  
\end{equation}
in which $Y_{i}(1),Y_i(0)\in\mathcal{Y}\subseteq\mathbb{R}$ are the
potential outcomes under treatment and control, respectively. Denote
by $P\in\mathcal{P}$ the joint distribution of $\left\{ X_{i},D_{i},Y_{i}(1),Y_{i}(0)\right\} $.
Then, the random sample $Z^{n}$ follows a joint distribution written as $P^{n}\in\mathcal{P}^n$, determined jointly by $P$, $n$ and \eqref{eq:observed.outcome}. In this paper, we maintain the following unconfoundeness and overlap assumptions:
\begin{assumption}\label{asm:unconfounded}
For each $i=1,...,n$, we have:
\begin{itemize}
\item[(i)] $Y_i(1),Y_i(0)\perp D_i\mid X_i$, i.e., $Y_i(1)$ and $Y_i(0)$ are independent
of $D_i$ conditional on $X_i$;

\item[(ii)] $\pi(x):=Pr\{D_i=1|X_i=x\}$ is bounded away from zero and one, i.e.,
$0<\underline{\pi}<\pi(x)<\bar{\pi}<1$ for some $\underline{\pi}$ and
$\bar{\pi}$, for all $x\in\mathcal{X}$.
\end{itemize}
\end{assumption}

The policymaker wishes to allocate a binary treatment $D\in\left\{ 0,1\right\} $
to a future population that shares the same marginal distribution
of $\left\{ X_{i},Y_{i}(1),Y_{i}(0)\right\} $ induced by $P$. For
each subpopulation group $X=x$ in which the covariate takes a specific
value $x\in\mathcal{X}$, write 
\begin{align*}
\gamma_{1}(x)  :=\mathbb{E}[Y(1)\mid X=x],\quad
\gamma_{0}(x)  :=\mathbb{E}[Y(0)\mid X=x],
\end{align*}
where $\mathbb{E}[\cdotp|\cdotp]$ denotes the conditional expectation
under $P$. Write
$\tau(x):=\gamma_{1}(x)-\gamma_{0}(x)
$
as the conditional average treatment effect (CATE) for subgroup $X=x$. Under Assumption \ref{asm:unconfounded}, the CATE $\tau(x)$ is point identified for
all $x\in\mathcal{X}$. 
We focus on the situation when --- at the decision-making
stage --- the set of covariates that the policymaker can actually
condition on, denoted by $W\in\mathcal{W}\subseteq\mathbb{R}^{d_W}$, is only a subset of and
not as rich as $X$, i.e., we may partition $X=\{W,X_1\}$ for some $X_1\in \mathbb{R}^{d_{X_1}}$. 
\begin{example}[Legal or fairness concern] A randomized control trial may collect sensitive characteristics (e.g., gender and race), while  the policymaker cannot differentiate treatment decisions  based on them due to legal or fairness concerns. For example, many countries have anti-discrimination laws  that prohibit  treating an individual differently because of their membership to a protected class. Calls for the removal of race in many clinical diagnoses are also growing,  see, e.g., \cite*{briggs2022healing,manski2022patient,manski2023using} and the debates therein. 
\end{example}


\begin{example}[Costly or manipulated variables at decision-making stage] In some scenarios, certain covariates are known to be important and are diligently recorded at the data-collecting stage. However, these variables could be costly to collect in practice and as a result, the decision maker does not observe these variables at the actual decision-making stage. For example, for patients in severe life-threatening
conditions such as sepsis, a physician must make a timely bedside intervention before lab measurements regarding key conditions of the patients can be returned \citep*{tan2022rise}. Moreover, some covariates may also be manipulated easily (e.g., they are not reported precisely) at the decision making stage, which makes them unsuitable to be included for treatment allocation. 
\end{example}


\begin{example}[Single-index rules and subgroup analyses] Even in the absence of legal, fairness, or cost concerns, policy makers may prefer simple and interpretable rules. For example, the decision maker may determine treatment eligibility based on a single scalar variable $W:=\varphi(X)$, a function of  the whole observed covariate $X$. 
See, e.g., \cite{kitagawa2018should,crippa2024regret} and references therein. Policy makers may also have particular pre-defined subgroups of interest for policy making, e.g., subgroups based on income or education brackets that are much coarser than observed income and education levels. These subgroups of interest may  be determined ex-ante in the pre-analysis plan prior to the data-collecting stage, or  determined ex-post by certain machine learning algorithms with data collected from earlier studies \citep{chernozhukov2018generic}.
\end{example}


\subsection{A regret averse planner's problem}

We start from the planner's problem without sample data $Z^{n}$.
Since the policymaker can only allocate policy decisions conditional
on $W$, we call their action plan a \emph{$W$-individualized decision
rule}, i.e., a mapping $
\delta:\mathcal{W}\rightarrow[0,1]$ from the support of $W$ to the unit interval. Here, $\delta(w)$
is the fraction of the subpopulation $W=w$ to be treated. For example,
$\delta(w)=0.5$ means half of the subpopulation
with $W=w$ will be treated, leaving the rest untreated. Although the policy
rule of the planner can only condition on $W$, treatment effect heterogeneity
may still vary at the more refined level $X$. For each group with $X=x$, let its corresponding  covariate $W$
take a value at $W=w$. Suppose that applying a generic $W$-individualized rule
$\delta$ to $X=x$ yields a linearly additive welfare for
the planner:
\begin{equation}
W(x,\delta,\gamma_{1},\gamma_{0}):=\gamma_{1}(x)\delta(w)+\gamma_{0}(x)(1-\delta(w)).\label{eq:welfare}
\end{equation}
Note the form of the welfare in (\ref{eq:welfare}) implies that the
optimal level of the welfare for $X=x$ is achieved by the
infeasible rule $\mathbf{1}\left\{ \tau(x)\geq0\right\} $. Then, for  $X=x$, define the regret of rule $\delta$ as the welfare
gap between $\delta$ and $\mathbf{1}\left\{ \tau(x)\geq0\right\}$:
\begin{align*}
Reg(x,\delta,\tau) & :=\max\left\{ \gamma_{1}(x),\gamma_{0}(x)\right\} -W(x,\delta,\gamma_{1},\gamma_{0})\\
 & =\tau(x)\left[\mathbf{1}\left\{ \tau(x)\geq0\right\} -\delta(w)\right].
\end{align*}
We consider a regret-averse policy maker who chooses an optimal $W$-individualized
policy rule by minimizing a nonlinear 
transformation of regret, a notion advocated by \cite{kitagawa2022treatment}
and axiomatized by \cite{hayashi2008regret,stoye2011axioms}. Specifically, the policy maker
aggregates regrets among different subpopulation groups via the average
nonlinear regret loss:
\begin{align*}
L(\delta,\tau)  :=\int Reg^{\alpha}(x,\delta,\tau)dF_{X}(x) :=\mathbb{E}\left[Reg^{\alpha}(X,\delta,\tau)\right],
\end{align*}
where $\alpha\geq1$ is the degree of regret aversion, $F_{X}(\cdotp)$
denotes the marginal distribution of $X$ induced by the population
distribution $P$, and $\mathbb{E}[\cdotp]$ is the expectation operator
under $P$. A nonlinear regret optimal decision rule $\delta^{*}$
then solves 
\begin{equation}\label{eq:population.optimal}
\min_{\delta:\mathcal{W}\rightarrow[0,1]}L(\delta,\tau).   
\end{equation}
The rule $\delta^{*}$ characterizes an optimal action plan (conditional
on $W$) for the planner  if $P$ were known. Since $P$ in fact is unknown and needs to be learned from 
data $\ensuremath{Z^{n}\in\mathcal{Z}^{n}}$,
the decision of the planner becomes statistical, i.e., selecting a
$W$-individualized \emph{statistical} decision rule:
\[
\hat{\delta}:\mathcal{Z}^{n}\times\mathcal{W}\rightarrow[0,1]
\]
that instructs an action for each subgroup $W=w$ given each possible
realization of data $Z^{n}=z^{n}\in\mathcal{Z}^{n}$. Denote by $\mathbb{E}_{P^{n}}[\cdotp]$
the expectation with respect to the randomness of 
$Z^{n}\sim P^{n}$. The planner's ultimate goal is to find a ``good''
rule $\hat{\delta}$ from the training data so that 
\begin{equation}\label{eq:goal.risk}
\sup_{P^{n}\in\mathcal{P}^{n}}\mathbb{E}_{P^{n}}\left[L(\hat{\delta},\tau)-L(\delta^{*},\tau)\right]
\end{equation}
is small and converges to zero at a fast rate (hopefully the
fastest) uniformly across a set of possible distributions $P^{n}\in\mathcal{P}^{n}$. 

\subsection{Regret aversion as inequality aversion}\label{sec:inequalty.aversion}

We argue that our regret aversion loss
$L(\delta,\tau)$ has baked in an aversion to regret  inequality
in the population.\footnote{Regret measures the welfare loss of a group compared to what \emph{could have achieved} in terms of its best potential. Thus, our $L(\delta,\tau)$ reflects the preference of a planner who cares about to what extent personalized policies are equally fulfilling each sub-population's potential.} 
More concretely, write $Reg_{\delta}(x):=Reg(x,\delta,\tau)$. Inspired
by the seminal work of \cite{atkinson1970measurement}, let 
\begin{equation}
I_{\alpha}(Reg_{\delta}):=\frac{\left\{ \mathbb{E}[Reg_{\delta}(X)^{\alpha}]\right\} ^{1/\alpha}}{\mathbb{E}[Reg_{\delta}(X)]}-1\label{eq:inequality}
\end{equation}
be the Atkinson inequality measure of the regret
distribution $Reg_{\delta}(\cdotp)$ in the population induced by
rule $\delta$.\footnote{We may take $I_\alpha(Reg_\delta)=0$ if $\mathbb{E}[Reg_{\delta}(X)]=0$ as a convention.} Call $\left\{ \mathbb{E}[Reg_{\delta}(X)^{\alpha}]\right\} ^{1/\alpha}$
as the \textit{equally distributed equivalent} level of regret ---
the level of regret, if equally distributed to each subgroup in $X$, would yield
the same level of the loss as the actual distribution $Reg_{\delta}(\cdotp)$.
As $\alpha\geq1$, $I_{\alpha}(Reg_{\delta})\geq0$. The index $I_{\alpha}(Reg_{\delta})$
then measures how much larger the \textit{equally
distributed equivalent} is compared to the actual mean of regret $\mathbb{E}[Reg_{\delta}(X)]$.
A larger value of $I_{\alpha}$ indicates a higher degree of regret inequality. In this context, the regret aversion coefficient
$\alpha$ can be alternatively interpreted as a degree of inequality
aversion of the policymaker. A larger value of $\alpha$ means the
planner is more averse to regret inequality  among the population.
From (\ref{eq:inequality}), we may rewrite our nonlinear regret loss as
\begin{align*}
\mathbb{E}[Reg_{\delta}(X)^{\alpha}] & =\left[\underset{\text{mean regret}}{\underbrace{\mathbb{E}[Reg_{\delta}(X)]}}\left(1+\underset{\text{penalty for regret inequality}}{\underbrace{I_{\alpha}(Reg_{\delta})}}\right)\right]^{\alpha}.
\end{align*}
The mean regret paradigm corresponds
$\alpha=1$: $I_{1}(Reg_{\delta})=0$ for all distributions
of regret, meaning the policy maker displays no aversion to regret inequality
and ranks each distribution of regret only by their mean. If $\alpha>1$,
the planner is averse to regret inequality and penalizes rules that
lead to large inequality among the population. 
\begin{figure}
\caption{Equally distributed equivalent of regret in an illustrative example}
\begin{center}
\includegraphics[scale=0.2]{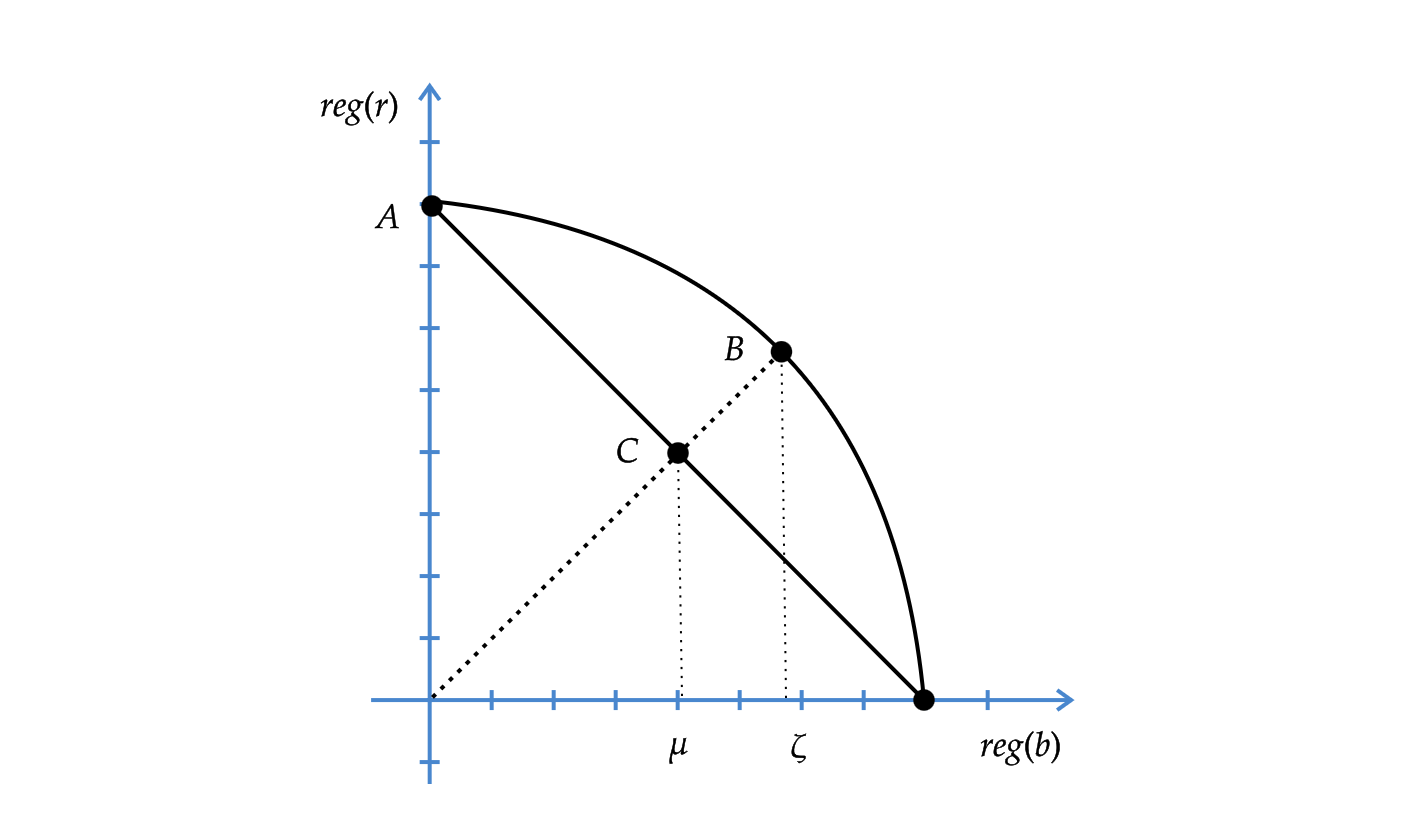}\label{fig:atkinson.inequality}
\end{center}
{\raggedright\footnotesize \textit{Notes}: The distribution of the regret for rule $\delta=1$ is represented at point $A$. Its equally distributed equivalent $\zeta$ can be found as the $x$-coordinate of point $B$, where the dotted $45^\circ$   line intersects the curved isoquant that shows the same level of the loss as point $A$. The actual mean of
the regret corresponds to the $x$-coordinate of point $C$, where the dotted line intersects the solid black line perpendicular to it through point $A$. Then, the inequality
index of $\delta=1$ is  $\frac{\zeta}{\mu}-1$.}

\end{figure}

\begin{example}\label{ex:simple}
Suppose $X=\left\{ b,r\right\} $ is a binary group identity (blue
or red) with equal shares in the population, with $CATE(b)=\tau_{b}>0$ and $CATE(r)=\tau_{r}<0$ ($0<\left|\tau_{r}\right|<\tau_{b}$). 
The policymaker cannot differentiate the two groups and can
only make a single treatment decision $\delta\in[0,1]$ to be applied to
the whole population. See Figure \ref{fig:atkinson.inequality} for an illustration of  
the equally distributed equivalent of the regret for rule $\delta=1$ (which benefits  group $b$ but hurts group $r$). 

\end{example}

\subsection{Population analysis}\label{sec:population}
We say a $W$-individualized
rule $\delta$ is a singleton rule if for almost all $w\in\mathcal{W}$,
it holds $\delta(w)\in\left\{ 0,1\right\} $; otherwise, we say $\delta$
is fractional. 
With a slight modification of notation, we write
$\tau(w):=\mathbb{E}[Y(1)-Y(0)\mid W=w]$. 
\begin{prop}
\label{prop:population}Consider the population optimal rule $\delta^{*}$
that solves (\ref{eq:population.optimal}).
\begin{itemize}
    \item[(i)] If $\alpha>1$, $\delta^{*}$ satisfies 
\[
\mathbb{E}\left[\tau(X)\left\{ \tau(X)\left[\mathbf{1}\left\{ \tau(X)\geq0\right\} -\delta^{*}(W)\right]\right\} ^{\alpha-1}\mid W=w\right]=0,
\]
for almost all $w\in\mathcal{W}$, and is fractional unless 
\[\min\left\{ Pr\{\tau(X)>0|W=w\},Pr\{\tau(X)<0|W=w\}\right\} =0\]
for almost all $w\in\mathcal{W}$;
 \item[(ii)] 
If $\alpha=1$,  then $\delta^{*}(w)=1$ if $\tau(w)>0$, $\delta^{*}(w)=0$ if $\tau(w)<0$, and $\delta^{*}(w)\in[0,1]$ if $\tau(w)=0$,
for almost all \textup{$w\in\mathcal{W}$.}
\end{itemize}
\end{prop}

Proposition \ref{prop:population} shows that a regret-averse planner
concerned about regret inequality will
often prefer a fractional $W$-individualized rule.  They would prefer a singleton
rule if, for almost all $w\in\mathcal{W}$, CATE($x$) shares the same sign for all $x$ with  the same value of $w$, which also nests the case when $W=X$.
Our results offer a novel justification of implementing fractional
rules at the population level: Treatment effect heterogeneity at the $X$ level plus a concern for regret inequality induces a planner
to ``diversify'' their treatment allocation. We illustrate the optimal rule of Example \ref{ex:simple} in Figure \ref{fig:ex.2}.
\begin{figure}[h!]
\caption{Optimal treatment rule in Example \ref{ex:simple}}
\begin{center}
\includegraphics[scale=0.27]{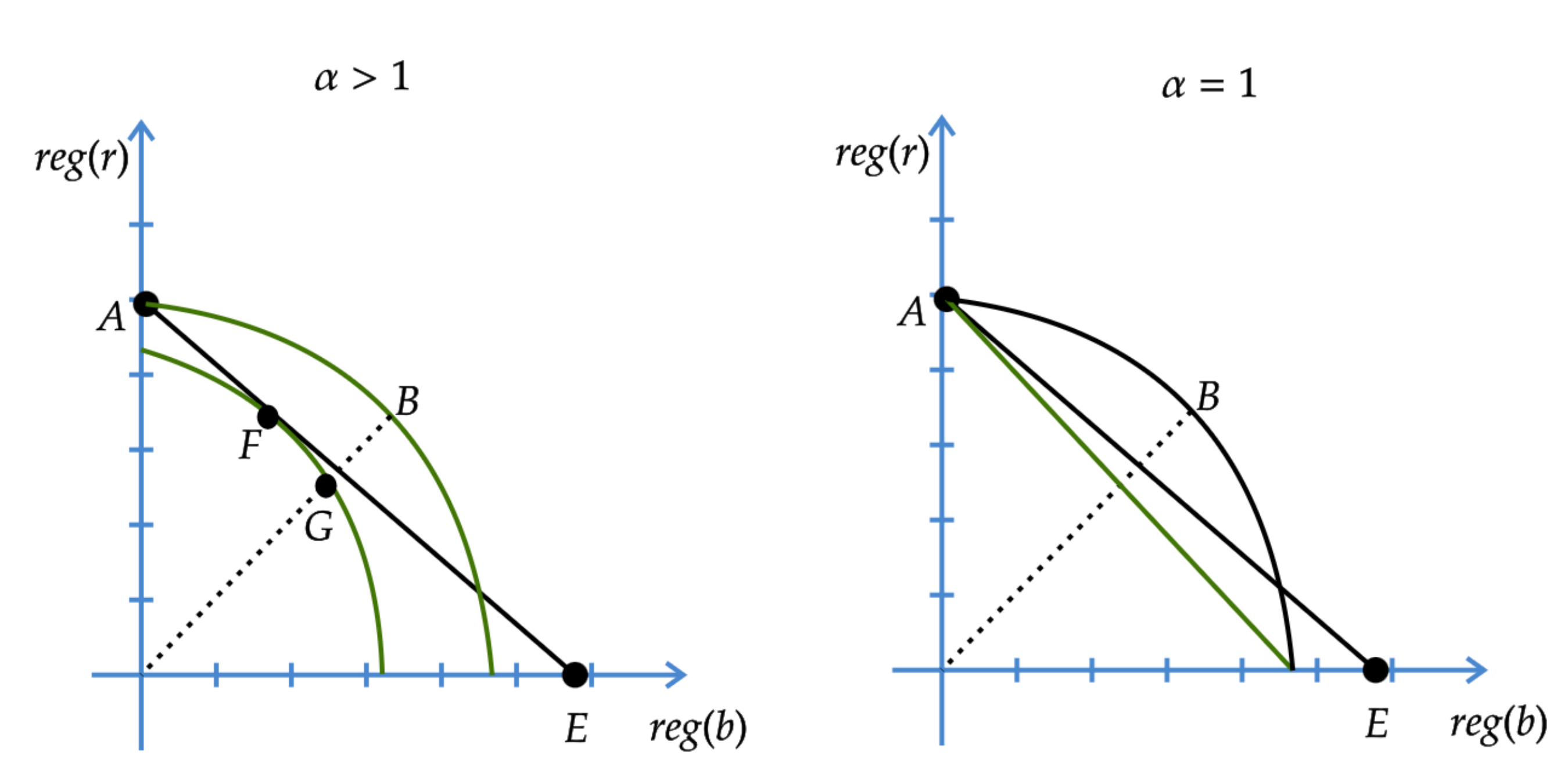}\label{fig:ex.2} 
\end{center}

{\raggedright\footnotesize
\textit{Notes}: Line AE, viewed as a budget line, collects all the feasible allocations
of regrets between groups $b$ and $r$ with a decision $\delta\in[0,1]$.
Point $A$ represents $\delta=1$, under which the regret of 
group $b$ is zero while the regret for group $r$ is $-\tau_{r}$.
Point $E$ is the regret distribution for rule $\delta=0$, in which the regret for group $r$ is zero but now group $b$ incurs
a regret of $\tau_{b}$. Every interior point of Line AE represents
a regret allocation of a fractional rule between the two groups. Since
$-\tau_{r}<\tau_{b}$, line AE tilts more heavily toward the horizontal
line. The green curves are the isoquants, each of them showing all
the allocations giving the same level of the loss function. The planner's
goal is to search for a point on Line AE that yields the smallest
loss. As long as $\alpha>1$, the isoquants will be strictly concave,
yielding an interior solution, i.e., point F in Figure \ref{fig:ex.2} (left), and
its equally distributed equivalent can be found as via point $G$. However, if $\alpha=1$,
the isoquants become linear, yielding a corner solution, i.e., point
A in Figure \ref{fig:ex.2} (right).}
\end{figure}


When $\alpha=2$, $L(\delta,\tau)$ becomes the average squared
regret, and the planner's problem becomes a weighted least squares
problem:
\[
\min_{\delta:\mathcal{W}\rightarrow[0,1]}\mathbb{E}\left\{ \tau^{2}(X)\left[\mathbf{1}\left\{ \tau(X)\geq0\right\} -\delta(W)\right]^{2}\right\} .
\]
Moreover, the associated
optimal rule also has an explicit form
\begin{equation}\label{eq:squared.regret.optimal}
\delta^{*}(w)=\frac{E\left[\tau^{2}(X)\mathbf{1}\{\tau(X)\geq0\}\mid W=w\right]}{E\left[\tau^{2}(X)\mid W=w\right]},   
\end{equation}
whenever $\mathbb{E}\left[\tau^{2}(X)\mid W=w\right]\neq0$.\footnote{If $\mathbb{E}\left[\tau^{2}(X)\mid W=w\right]=0$ for some $w\in\mathbb{R}^{d_W}$, then $\delta^*(w)\in[0,1]$, i.e., any action in $[0,1]$ is optimal for those $w$ values. Our theory accommodates this ``non-uniqueness'' situation to some extent. See Assumptions \ref{asm:reg} and \ref{asm:margin} and discussions therein. }
\eqref{eq:squared.regret.optimal} shows that the fractional nature of the optimal rule $\delta^*$ depends not only on the sign and magnitude of the average treatment effect $\tau(x)$, but also the conditional distribution of $X$ given $W$. The optimal treatment assignment would be more fractional (i.e., closer to 0.5) if the values of $\int_{\tau(x)>0}\tau^{2}(x)dF_{X|W}(x\mid w)$ and $\int_{\tau(x)<0}\tau^{2}(x)dF_{X|W}(x\mid w)$ are closer to each other.

\begin{rem}
\cite{atkinson1970measurement}'s original proposal concerns inequality of \emph{welfare levels}. In our context, that corresponds to picking a concave transformation $U(\cdotp)$
and solving:
\begin{equation}\label{eq:concave.welfare}
\max_{\delta:\mathcal{W}\rightarrow[0,1]}\int U\left[W(x,\delta,\gamma_{1},\gamma_{0})\right]dF_{X}(x).
\end{equation}
We adapt
\cite{atkinson1970measurement}'s framework to focus on  inequality of \emph{regret}. It would be easy to construct examples in which low inequality of welfare levels imply high inequality of regret, and vice versa.  Given regret measures how far away each group is compared to their optimal welfare level, we think our approach could be  suitable when the policy maker mainly cares about   supporting each group to their full potential and not considerably hurting any group.
%
\end{rem}

\begin{rem}
One possibility of incorporating regret aversion is through an alternative loss function $\tilde{L}(\delta, \tau):=\{\mathbb{E}\left[Reg(X,\delta,\tau)\right]\}^{\alpha}$ for the same $\alpha\geq 1$. That is, the planner first aggregates subgroup level regret $Reg(X,\delta,\tau)$ and then takes a nonlinear transformation of the aggregate average regret ${\mathbb{E}\left[Reg(X,\delta,\tau)\right]}$.\footnote{C.f. \cite{manski2016sufficient} for discussions on achieving an optimality criterion for each observed covariate group or within the overall population only.} However, in this case, minimizing $\tilde{L}$ for any $\alpha\geq1$  is the same as minimizing  $\mathbb{E}\left[Reg(X,\delta,\tau)\right]$, i.e., the case of $\alpha=1$ for our loss $L$. Such a planner also does not display aversion to regret inequality and only ranks rules according to their group-wise average regret.  One may also twist our loss function by redefining the regret for each subgroup $X=x$ as:
\begin{align*}
\widetilde{Reg}(x,\delta,\tau):= \tau(w)\left[\mathbf{1}\left\{ \tau(w)\geq0\right\} -\delta(w)\right],
\end{align*}
leading to an alternative loss 
$\widetilde{L}(\delta,\tau):=\mathbb{E}\left[\widetilde{Reg}^{\alpha}(X,\delta,\tau)\right]$.
That is, the regret of each subgroup $X=x$ is evaluated according to the welfare gap of its corresponding coarser $W=w$ group compared with the best welfare for the same coarser $W=w$ group. However, a planner with a loss $\widetilde{L}(\delta,\tau)$ is not concerned about regret inequality within the $W$ group.\footnote{For instance, in Example \ref{ex:simple}, as  $0<-\tau_r<\tau_b$, the overall average treatment effect is positive. Hence, according to $\widetilde{L}$, rule $\delta=1$ would yield a regret of $0$ for both $r$ and $b$ groups --- meaning there is no inequality between the two groups. However, we know $\delta=1$ actually hurts group $r$ dramatically as $CATE(r)<0$.}
\end{rem}

\begin{rem}
If $\alpha>1$ but the action space of the planner is restricted, i.e, $\delta(w)\in\left\{ 1,0\right\}$  for all $w\in\mathcal{W}$, the optimal rule $\delta^{*}$ would still be different from that of $\alpha=1$.\footnote{Moreover, to what extent one shall view the action space  as ``restricted'' is subject to the interpretation of the researcher. Even when a planner cannot take a fractional treatment allocation \emph{per se}, considering an extended action space $[0,1]$ is still valuable, as $\delta(w)\in[0,1]$ may be interpreted as a probabilistic recommendation, instead of an actual allocation of treatment. } For example, when $\alpha=2$, the average squared regret for action $1$, conditioning on  $W=w$,  is $
\mathbb{E}\left[\tau^{2}(X)\left(\mathbf{1}\left\{ \tau(X)\geq0\right\} -1\right)^{2}\mid W=w\right]$,
while that of  action $0$ is 
$\mathbb{E}\left[\tau^{2}(X)\left(\mathbf{1}\left\{ \tau(X)\geq0\right\} \right)^{2}\mid W=w\right]$.
Thus, the optimal restricted $W$-individualized rule is 
$
\delta^{*}_{\text{restricted}}(w)=\mathbf{1}\left\{ \mathbb{E}\left[\tau^{2}(X)\text{sgn}(\tau(X))\mid W=w\right]\geq0\right\}$.
\end{rem}




\section{Main proposal}\label{sec:proposal}

From now on and for the rest of
the paper, we focus on $\alpha=2$ for tractability. Other
values of $\alpha>1$ may  be analyzed analogously  with
more technicalities. Our
 proposal for learning a good rule $\hat{\delta}$
from training data involves two steps. The first step is
the efficient estimation of the loss function for each fixed $\delta$, i.e., 
\begin{equation}
L(\delta,\tau)=\mathbb{E}\left\{ \tau^{2}(X)\left[\mathbf{1}\left\{ \tau(X)\geq0\right\} -\delta(W)\right]^{2}\right\}.\label{eq:loss.proposal}
\end{equation}
Once
$L(\delta,\tau)$ is efficiently estimated from data, 
the second step is to minimize the estimated loss 
among a class of $W-$individualized rules that must be specified
by the researcher.
\begin{figure}[http]
\centering
\begin{minipage}{0.43\textwidth}
\centering
\includegraphics[width=0.9\textwidth]{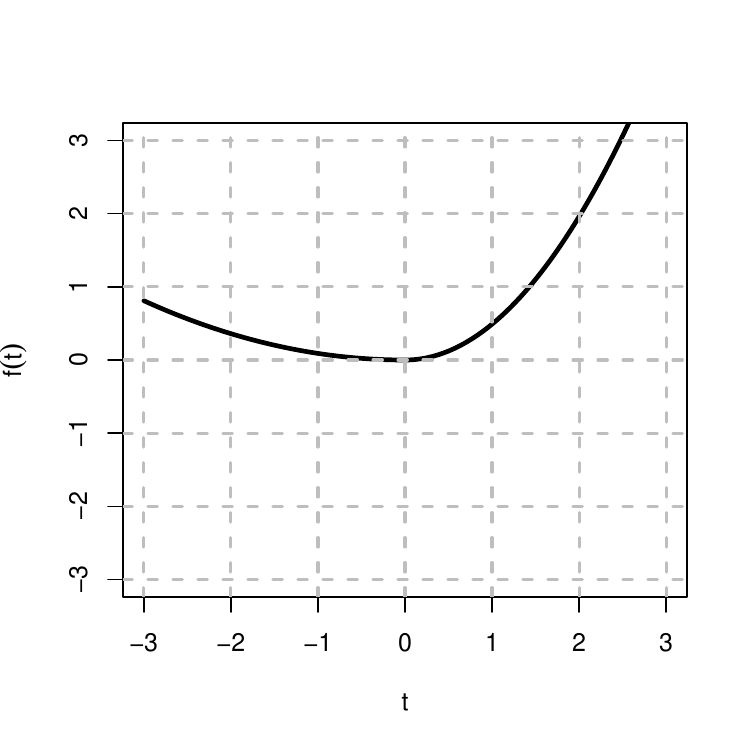} 
\text{$f(t)=t^{2}\left(\mathbf{1}\left\{ t\geq0\right\} -\delta\right)^{2}$}
\end{minipage}
\begin{minipage}{0.43\textwidth}
\centering
\includegraphics[width=0.9\textwidth]{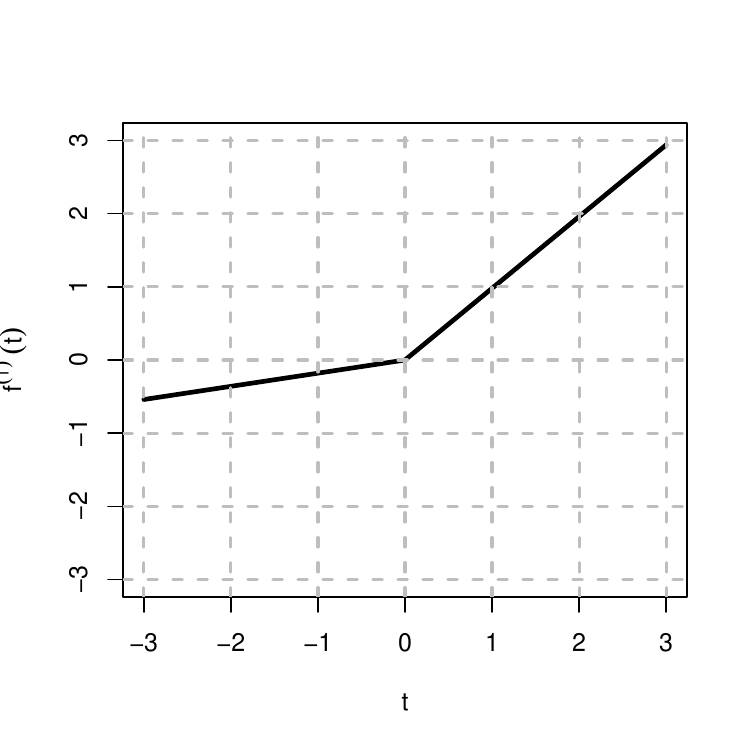} 
        \text{$f^{(1)}(t)=2t\left(\mathbf{1}\left\{ t\geq0\right\} -\delta\right)^{2}$}
    \end{minipage}
    \caption{Average squared regret functional is continuously differentiable}\label{fig:diff}
    \label{fig:enter-label}
\end{figure}

To allow for a wide range of ML algorithms in the first step and to potentially improve the statistical qualities of the second step, we consider
``debiasing'' (in a sense we make precise below) the loss function together with cross-fitting.
\footnote{See Remark \ref{rem:plug.in} for discussions on the connections with more direct ``plug-in'' approaches that do not involve debiasing.} Note although the nuisance function $\tau$ appears inside the indicator function, the loss $L(\delta, \tau)$ is still continuously differentiable in $\tau$ (though not twice continuously differentiable). This is due to the squaring of the term $[\mathbf{1}\{\tau(X) \geq 0\} - \delta(W)]$, which smooths out the discontinuity (See Figure \ref{fig:diff}).  Furthermore, we may view $L(\delta,\tau)$ as a finite dimensional parameter $\theta_0\in\mathbb{R}$
in a moment condition 
\begin{equation}\label{eq:moment.loss}
\mathbb{E}\left[m(Z,\theta_{0},\tau)\right]=0,\quad \text{where } m(Z,\theta,\tau):=\tau(X)^{2}\left(\mathbf{1}\{\tau(X)\geq0\}-\delta(W)\right)^{2}-\theta.    
\end{equation}
Therefore, following \citet[][Proposition 4]{newey1994asymptotic}, we can still derive the efficient influence function for any regular and asymptotically linear estimator of $L(\delta,\tau)$.  Let $\eta_{0}:=(\gamma_{1},\gamma_{0},\omega_{1},\omega_{0})$, where
 $\omega_{1}(x):=\frac{2\left(\gamma_{1}(x)-\gamma_{0}(x)\right)}{\pi(x)}$,
$\omega_{0}(x):=\frac{2\left(\gamma_{1}(x)-\gamma_{0}(x)\right)}{1-\pi(x)}$. Write
\[\xi(Z,\eta_{0}):=\left[\gamma_{1}(X)-\gamma_{0}(X)\right]^{2}+\left[D\omega_{1}(X)(Y-\gamma_{1}(X))-\left(1-D\right)\omega_{0}(X)(Y-\gamma_{0}(X))\right].\]

\begin{prop}\label{prop:debias}

Suppose Assumption \ref{asm:unconfounded} holds. For each $\delta$, the efficient influence function for any regular and asymptotically linear estimator of $\theta_0:=L(\delta,\tau)$ is 
\[
\psi(Z)=\xi(Z,\eta_{0})\left(\mathbf{1}\{\gamma_{1}(X)-\gamma_{0}(X)\geq0\}-\delta(W)\right)^{2}-\theta_0
\]
\end{prop}
As $Var(\psi(Z))$ defines the semiparametric efficiency bound for estimating $L(\tau,\delta)$, we think it makes sense to exploit the structure of $\psi(Z)$ and define our modified loss function as:\footnote{Moreover, with the margin condition in Assumption \ref{asm:margin} and other regularity conditions, $L^o(\delta,\eta_0)$ can also be verified to  satisfy the Neyman orthogonal condition (c.f., \citealt{chernozhukov2018double} and references therein).}
\[
L^{o}(\delta,\eta_{0}):=\mathbb{E}\left[\xi(Z,\eta_{0})\left(\mathbf{1}\{\gamma_{1}(X)-\gamma_{0}(X)\geq0\}-\delta(W)\right)^{2}\right].
\] 

Our theory does not restrict how the additional  nuisance  functions, $\omega_{1}$ and $\omega_{0}$, should be estimated. However, we note
they feature the following \emph{balancing} property (see, e.g., \citealt*{hainmueller2012entropy,zubizarreta2015stable,athey2018approximate}): for  all $g(\cdotp)$ such that $\mathbb{E}[g^{2}(X)]<\infty$,
\begin{align}
\mathbb{E}\left[D\omega_{1}(X)g(X)\right] & =\mathbb{E}\left[2\left(\gamma_{1}(X)-\gamma_{0}(X)\right)g(X)\right]=\mathbb{E}\left[(1-D)\omega_{0}(X)g(X)\right],\label{eq:balancing}
\end{align}
which may facilitate their estimation without the need to calculate propensity score. For example, to construct an estimator for $\omega_{1}$, denote
by $b(x):=(b_{1}(x),\ldots,b_{\text{dim}(b)}(x))^{\prime}$ a vector
of $\text{dim}(b)$ basis functions. Let  $\left\Vert \cdot\right\Vert$ be the vector $l_{2}$ norm. Note \eqref{eq:balancing} implies
\begin{equation}\label{eq:practical.omega}
\mathbb{E}\left[D\omega_{1}(X)b(X)\right]=\mathbb{E}\left[2\left(\gamma_{1}(X)-\gamma_{0}(X)\right)b(X)\right], 
\end{equation}
suggesting we may estimate $\omega_1$  by solving the following minimum distance estimator with a Tikhonov penalty (c.f., \citealt{chen2012estimation,qiu2022approximate}):
\begin{equation}\label{eq:minimum.distance.penalty}
\min_{\omega_{1}\in\Theta_{n}}\left\Vert \frac{1}{n}\sum_{i=1}^{n}\left[2\left(\hat{\gamma}_{1}(X_i)-\hat{\gamma}_{0}(X_i)\right)b(X_i)-D_i\omega_{1}(X_i)b(X_i)\right]\right\Vert^2 +\lambda_{1,n}\frac{1}{n}\sum_{i=1}^{n}\left[D_i\omega_{1}^{2}(X_i)\right],    
\end{equation}
where $\hat{\gamma}_{1},\hat{\gamma}_{0}$ are estimated versions
of $\gamma_{1}$ and $\gamma_{0}$, 
\[
\Theta_{n}=\left\{ f:\mathcal{X}\rightarrow\mathbb{R}\mid f(x)=a^{\prime}b(x),a\in\mathbb{R}^{\text{dim}(b)}\right\} ,
\]
and $\lambda_{1,n}\geq0$ is a tuning parameter specified by the researcher.\footnote{In the empirical application, we use cross validation to select the tuning parameter $\lambda_{1,n}$. See Appendix \ref{sec:compute} for additional computational details.}


We now describe our cross-fitting procedure to estimate $L^{o}(\delta,\eta_0)$ from
data $Z^{n}=\left\{ X_{i},D_{i},Y_{i}\right\} {}_{i=1}^{n}$. Let
$[n]:=\{1,\dots,n\}$ be the observation index set.  Randomly partition
$[n]$ into approximately equal-sized $K\geq2$ folds $\left(I_{k}\right)_{k=1}^{K}$. Without loss of generality, we assume each fold is of sample size $m:=n/K$. 
For each $k\in[K]:=\{1,\dots,K\}$, let $I_{k}^{c}:=[n]\backslash I_{k}$
only include observations $\textit{not}$ from fold $I_{k}$. For
each $I_{k}$, $k\in[K]$, we estimate $\eta_{0}$ by $\hat{\eta}^{k}:=(\hat{\gamma}_{1}^{k},\hat{\gamma}_{0}^{k},\hat{\omega}_{1}^{k},\hat{\omega}_{0}^{k})$,
where $\hat{\gamma}_{1}^{k}:=\hat{\gamma}_{1}\left(\left(Z_{j}\right)_{j\in I_{k}^{c}}\right)$,
$\hat{\gamma}_{0}^{k}:=\hat{\gamma}_{0}\left(\left(Z_{j}\right)_{j\in I_{k}^{c}}\right)$,
$\hat{\omega}_{1}^{k}:=\hat{\omega}_{1}\left(\left(Z_{j}\right)_{j\in I_{k}^{c}}\right)$
and $\hat{\omega}_{0}^{k}:=\hat{\omega}_{0}\left(\left(Z_{j}\right)_{j\in I_{k}^{c}}\right)$,
i.e., $\hat{\eta}^{k}$ is constructed only using data in $I_{k}^{c}$.
Then, for each $\delta$, an estimator of $L^{o}(\delta,\eta_{0})$
is
\[
\hat{L}_{n}^{o}(\delta):=\frac{1}{n}\sum_{i=1}^{n}\hat{\xi}(Z_{i})\left(\mathbf{1}\{\hat{\tau}(X_{i})\geq0\}-\delta(W_{i})\right)^{2},
\]
where 
\begin{align*}
\hat{\xi}(Z_{i}) & :=\sum_{k=1}^{K}\hat{\xi}^{k}(Z_{i})\mathbf{1}\left\{ i\in I_{k}\right\} ,\hat{\xi}^{k}(Z_{i}):=\xi(Z_{i},\hat{\eta}^{k}),\\
\hat{\tau}(X_{i}) & :=\sum_{k=1}^{K}\left(\hat{\gamma}_{1}^{k}(X_{i})-\hat{\gamma}_{0}^{k}(X_{i})\right)\mathbf{1}\left\{ i\in I_{k}\right\} 
\end{align*}
are estimated versions of the weight $\xi(Z_{i},\eta_{0})$ and CATE
$\tau(X_{i})$ for each $i\in[n]$. Next, let $p(w):=(p_{1}(w),p_{2}(w),\ldots p_{\sieved}(w))^{\prime}$
be a vector of basis functions with dimension
$\sieved:=\sieved(n)$ that may grow as $n\rightarrow\infty$. Write 
\begin{equation}
\hat{A}_{n}:=\frac{1}{n}\sum_{i=1}^{n}\hat{\xi}(Z_{i})p(W_{i})p(W_{i})^{\prime},\quad \hat{B}_{n}:=\frac{1}{n}\sum_{i=1}^{n}\hat{\xi}(Z_{i})p_{i}(W_{i})\mathbf{1}\left\{ \hat{\tau}(X_{i})\geq0\right\}.
\end{equation}
Our final estimated policy is defined as:
\begin{align}\label{eq:our.solution.trimmed}
\hat{\delta}^{\mathcal{T}}(w)	:=\begin{cases}
1, & \hat{\delta}(w)>1,\\
\hat{\delta}(w), & \hat{\delta}(w)\in[0,1],\\
0, & \hat{\delta}(w)<0,
\end{cases}  
\end{align}
where 
\begin{equation}
\hat{\delta}(w)=\hat{\beta}^{\prime}p(w),\quad \hat{\beta}:={\hat{A}_{n}}^{-}\hat{B}_{n},\label{eq:our.solution}
\end{equation}
and ${(\cdot)}^{-}$ denotes the Moore-Penrose inverse.\footnote{Our cross-fitted procedure is considered as DML2 \citep{chernozhukov2018double}. One may also consider a different cross-fitting approach, in which we solve for the optimal rule in  each fold before taking the averages over all folds. It would be interesting to compare these two approaches in policy learning problems, in light of the recent progress of \cite{velez2024asymptotic} for estimation problems.} 

The rationale behind \eqref{eq:our.solution.trimmed} is as follows. Given $\hat{L}_{n}^{o}(\delta)$, one may consider finding an optimal rule by solving 
\begin{equation}
\inf_{\delta\in\mathcal{D}}\hat{L}_{n}^{o}(\delta),\label{eq:our.proposal}
\end{equation}
where 
\begin{equation}
\mathcal{D}:=\mathcal{D}_{n}:=\left\{ f(w)=\sum_{j=1}^{\sieved}\beta_{j}p_{j}(w):\beta_{j}\in\mathbb{R},\forall j=1\ldots\sieved\right\} .\label{eq:policy.class.linear}
\end{equation}
is a  linear sieve policy class.\footnote{It is a common practice to use a class of linear functions to approximate
a probability function. See, e.g., \citet*{chen2008semiparametric}. Our theory in fact can also be extended to other policy classes, e.g.,
a class of logit functions, with more technicalities.} \eqref{eq:our.proposal} may be viewed as a \emph{weighted least squares (empirical projection)}  problem, in which we predict an estimated outcome  $\mathbf{1}\{\hat{\tau}(X_i)\geq0\}$ in space $\mathcal{D}$ with an estimated weight $\hat{\xi}(Z_i)$.  Due to the presence of the adjustment term to ``debias'', the weight $\hat{\xi}(Z_i)$ may be negative, and the Hessian matrix $\hat{A}_n$
may also not be positive semidefinite. As a result, the problem in
\eqref{eq:our.proposal} is not necessarily convex in finite sample. However, our theory shows that, whenever $\hat{\eta}^{k}$ is of
high quality and $n$ is sufficiently large (in a sense we make precise),
the probability of $\hat{A}_{n}$ not being positive definite is exponentially
small. In addition, on the event that $\hat{A}_{n}$ is positive definite, \eqref{eq:our.proposal} has a unique solution  \eqref{eq:our.solution}, in which the Moore-Penrose inverse reduces to a standard inverse.\footnote{Therefore, \eqref{eq:our.solution} also has an interesting IV interpretation with an outcome of interest $\mathbf{1}\{\hat{\tau}(X_i)\geq0\}$,
a vector of endogenous variables $p(W_i)$,
and a vector of instrument $\hat{\xi}(Z_i)p(W_i)$.} Finally, to guarantee the estimated policy is indeed a valid decision rule and also for technical tractability, \eqref{eq:our.solution.trimmed} takes a trimmed form.\footnote{See, e.g., \cite*{newey1994series,newey1999nonparametric} for  examples, in other contexts, of using trimming to improve the  theoretic performances of certain statistics.} Note \eqref{eq:our.solution.trimmed} is well-defined irrespective of whether $\hat{A}_n$ is positive definite or not.

\section{Performance guarantee}\label{sec:stat}

Let $e_{1}:=Y(1)-\gamma_{1}(X)$,
$e_{0}:=Y(0)-\gamma_{0}(X)$ and $A:=\mathbb{E}[\tau^{2}(X)p(W)p^{\prime}(W)]$. We first impose the following regularity
conditions. 
\begin{assumption}
\label{asm:reg}
\begin{itemize}
\item[(i)]  There exist some constants $C_{e}$ and $C_{\gamma}$ such that $\left|e_{1}\right|\leq C_{e}$, $\left|e_{0}\right|\leq C_{e}$, $\sup_{x\in\mathcal{X}}\left|\gamma_{1}(x)\right|\leq C_{\gamma}$, $\sup_{x\in\mathcal{X}}\left|\gamma_{0}(x)\right|\leq C_{\gamma}$. 
\item[(ii)] All the eigenvalues of $A$
are bounded from above and away from zero. 
\end{itemize}
\end{assumption}

Under Assumptions \ref{asm:unconfounded} and \ref{asm:reg}, there exists some $C_{\xi}$ such that $\sup_{z\in\mathcal{Z}}\left|\xi(z,\eta_0)\right|\leq C_{\xi}$, and denote by $\bar{\lambda}<\infty$ and $\underline{\lambda}>0$
the maximum and minimum eigenvalues of $A$. Notably, even if $\mathbb{E}\left[\tau^{2}(X)\mid W=w\right]=0$ for some $w\in\mathbb{R}^{d_W}$, Assumption \ref{asm:reg}(ii) may still hold, thus allowing $\delta^*(w)$ to be non-unique for some $w$ values. 
Next, we impose the following
statistical quality requirements on the learners of $\eta_0$. Let 
$
\mathbb{E}_{k}\left[\cdotp\right]:=\mathbb{E}_{P^{n}}\left[\cdotp\mid\left\{ Z_{j}\right\} _{j\in[n]\setminus I_{k}}\right]$.
\begin{assumption}
\label{asm:quality} For each $k\in[K]$,
the following holds:
\begin{itemize}
\item[(i)] for some constant $C_{M}$ and  some constants $r_{\gamma_{1}},r_{\gamma_{0}},r_{\omega_{1}}$
and $r_{\omega_{0}}$ in $(0,1]$,
\[
\begin{aligned}\mathbb{E}_{k}\left[\int\left(\hat{\gamma}_{1}^{k}(x)-\gamma_{1}(x)\right)^{2}dF_{X}(x)\right]\leq C_{M}n^{-r_{\gamma_{1}}}, & \mathbb{E}_{k}\left[\int\left(\hat{\gamma}_{0}^{k}(x)-\gamma_{0}(x)\right)^{2}dF_{X}(x)\right]\leq C_{M}n^{-r_{\gamma_{0}}},\\
\mathbb{E}_{k}\left[\int\left(\hat{\omega}_{1}^{k}(x)-\omega_{1}(x)\right)^{2}dF_{X}(x)\right]\leq C_{M}n^{-r_{\omega_{1}}}, & \mathbb{E}_{k}\left[\int\left(\hat{\omega}_{0}^{k}(x)-\omega_{0}(x)\right)^{2}dF_{X}(x)\right]\leq C_{M}n^{-r_{\omega_{0}}};
\end{aligned}
\]
\item[(ii)] conditional on $\left\{ Z_{j}\right\} _{j\in[n]\setminus I_{k}}$ and for some constant $\tilde{C}_{M}$, 
\begin{align*}
\sup_{x\in\mathcal{X}}\left|\hat{\gamma}_{1}^{k}(x)-\gamma_{1}(x)\right| & \leq\tilde{C}_{M},\sup_{x\in\mathcal{X}}\left|\hat{\gamma}_{0}^{k}(x)-\gamma_{0}(x)\right|\leq\tilde{C}_{M},\\
\sup_{x\in\mathcal{X}}\left|\hat{\omega}_{1}^{k}(x)-\omega_{1}(x)\right| & \leq\tilde{C}_{M},\sup_{x\in\mathcal{X}}\left|\hat{\omega}_{0}^{k}(x)-\omega_{0}(x)\right|\leq\tilde{C}_{M}.
\end{align*}

\end{itemize}

\end{assumption}

By Assumption \ref{asm:quality}(i), our cross-fitted learners
of $\eta_0$ are mean square consistent with certain
convergence rates. Moreover, Assumptions \ref{asm:unconfounded}, \ref{asm:reg} and \ref{asm:quality}(ii) together imply that  there exists some $\tilde{C}_{\xi}$ such that for all $k\in[K]$, $\sup_{z\in\mathcal{Z}}\left|\hat{\xi}^k(z)\right|\leq\tilde{C}_{\xi}$ conditional on $\left\{ Z_{j}\right\} _{j\in[n]\setminus I_{k}}$. 

We now present a high-level stability condition of $\hat{\delta}$ useful for deriving fast convergence rates of our proposal. 
\begin{assumption}\label{asm:stability}
For some positive constant $\underline{\lambda}_\varepsilon$, we have $\sup_{w\in\mathcal{W}}|\hat{\delta}(w)|\cdot\mathbf{1}\{\lambda_{\text{min}}(\hat{A}_n)\geq\underline{\lambda}_\varepsilon\}\leq C_{L}$ for some finite constant $C_L$ (which may depend on $\underline{\lambda}_\varepsilon$).  
\end{assumption}
Assumption~\ref{asm:stability} essentially requires that the solution to the weighted least squares problem~\eqref{eq:our.proposal} satisfies a stability property with respect to the sup norm.
\footnote{See, for example, the sup-norm stability property of empirical $L_2$ projections using certain basis functions (e.g., splines and wavelets), which has been exploited by \cite*{huang2003local,belloni2015some,chen2015optimal} for sharp bias control in least squares series estimation. Our weighted least squares problem~\eqref{eq:our.proposal}, however, differs from those studied in the preceding literature due to the presence of estimated weights and outcomes.}  With Assumptions \ref{asm:unconfounded}-\ref{asm:quality}, we verify in Appendix \ref{sec:Bspline} that Assumption \ref{asm:stability} holds if $\mathcal{D}$ is constructed with B-spline basis functions. 
Finally, we consider the following  margin condition that concerns  the distribution of $\tau(X)$ in the neighborhood of $\tau(X)=0$ (see also \citealt{tsybakov2004optimal,kitagawa2018should}):
\begin{assumption}\label{asm:margin}
There exist positive constants $C_{\tau}$, $\alpha$, and $t^{*}$ such that
\[
P\left( \left| \tau(X) \right| \leq t \right) \leq C_{\tau} t^{\alpha}, \quad \text{for all } 0 \leq t \leq t^{*}.
\]
\end{assumption}
Note Assumption \ref{asm:margin} rules out  $P\{\tau(X)=0\}>0$, implying that $\delta^*(w)$ will be unique a.e. with respect to the marginal distribution of $W$. We are now ready to state our main result. Denote by $\sievedstar$ the dimension of the basis functions for  $\{f^2:f\in\mathcal{D}\}$, where $\mathcal{D}$ is defined in \eqref{eq:policy.class.linear}, and write $\zeta_{p}:=\sup_{w\in\mathcal{W}}\left\Vert p(w)\right\Vert$.\footnote{The magnitudes of $\sievedstar$ and $\zeta_p$ depend the choice of the basis functions. An upper bound of $d^*_p$ is $\sieved^2$, but  $\sievedstar$ may be as small as  $O(\sieved)$, e.g., for B-splines. The quantity of $\zeta_p$ is a key object of interest in the series estimation literature. It is well known that $\zeta_p=O(\sqrt{\sieved})$ for general spline basis functions (see, e.g., \citealt{newey1997convergence}). For B-splines, its structural properties imply that in fact, $\zeta_p\leq1$. See Appendix \ref{sec:Bspline} for additional treatments.} 
\begin{thm}
\label{thm:main}Suppose Assumptions \ref{asm:unconfounded}-\ref{asm:stability}
hold. Fix $0<\varepsilon<\min\left\{ \underline{\lambda},6\tilde{C}_{\xi}\zeta_{p}^{2},3C_{\xi}\zeta_{p}^{2}/2\right\}$, and  let
\begin{align*}
R_{n,O} & :=\frac{\sievedstar}{n}+\sup_{P^{n}}\inf_{\delta\in\mathcal{D}}\left[L(\delta,\tau)-L(\delta^{*},\tau)\right],\\
R_{n,B} & :=\text{max}\{\left(\log2d_{p}\right)^{2}\zeta_{p}^{6},\sieved\zeta_{p}^{3}\}\left(n^{-2r_{\gamma_{1}}}+n^{-2r_{\gamma_{0}}}+n^{-\left(r_{\omega_{1}}+r_{\gamma_{1}}\right)}+n^{-\left(r_{\omega_{0}}+r_{\gamma_{0}}\right)}\right),\\
R_{n,V} &:=\zeta_{p}^{3}n^{-1},\\
R_{n,F}&:=4C_{\xi}d_{p}\left[\exp\left(\frac{-n\varepsilon^{2}}{4C_{\xi}^{2}\zeta_{p}^{4}}\right)+K\exp\left(\frac{-n\varepsilon^{2}}{16K\tilde{C}_{\xi}^{2}\zeta_{p}^{4}}\right)\right].
\end{align*}
Then, for each $n$ such that
\[
4C_{M}\zeta_{p}^{2}\left(n^{-r_{\gamma_{1}}}+n^{-r_{\gamma_{0}}}+n^{-\frac{r_{\omega_{1}}+r_{\gamma_{1}}}{2}}+n^{-\frac{r_{\omega_{0}}+r_{\gamma_{0}}}{2}}\right)<\varepsilon,
\]
the following statements hold. 
\begin{itemize}
\item[(i)] For some constant $\mathcal{C}$ that is independent of $n$, $d_p$, $d^*_p$ and $\zeta_p$,
\begin{align}\label{eq:rate.main}
\sup_{P_{n}}\mathbb{E}_{P_{n}}\left[L(\hat{\delta}^{\mathcal{T}},\tau)-L(\delta^{*},\tau)\right]& \leq\mathcal{C}\left(R_{n,O}+R_{n,B}+R_{n,V}\right)+R_{n,F}.
\end{align}
\item[(ii)] The right-hand side of \eqref{eq:rate.main} improves to 
\begin{align}\label{eq:rate.w.margin}
\mathcal{C}\left(R_{n,O}+R_{n,B}+R_{n,V}\left(n^{-r_{\gamma_{1}}}+n^{-r_{\gamma_{0}}}\right)^{\frac{\alpha}{\alpha+2}}\right)+R_{n,F},
\end{align}
with a constant $\mathcal{C}$ suitably redefined (but also independent of $n$, $d_p$, $d^*_p$ and $\zeta_p$), if in addition, Assumption
\ref{asm:margin} holds and $n$ is also such that 
$\left(4C_{M}C_{\tau}^{-1}\left(n^{-r_{\gamma_{1}}}+n^{-r_{\gamma_{0}}}\right)\right)^{\frac{1}{\alpha+2}}<t^{*}$.
\end{itemize}
\end{thm}

Theorem \ref{thm:main} provides an upper bound for  the uniform excess risk whenever $n$ is sufficiently large. As long as  $R_{n,O}$, $R_{n,B}$ and $R_{n,V}$
all go to zero at a polynomial rate as a function of $n$, the exponential term $R_{n,F}$ will be asymptotically negligible, implying immediately that  if Assumptions \ref{asm:unconfounded}-\ref{asm:stability} hold,
\[
\sup_{P_{n}}\mathbb{E}_{P_{n}}\left[L(\hat{\delta}^{\mathcal{T}},\tau)-L(\delta^{*},\tau)\right]=O(R_{n,O}+R_{n,B}+R_{n,V}),
\]
and with the additional Assumption \ref{asm:margin},
\[
\sup_{P_{n}}\mathbb{E}_{P_{n}}\left[L(\hat{\delta}^{\mathcal{T}},\tau)-L(\delta^{*},\tau)\right]=O\left(R_{n,O}+R_{n,B}+R_{n,V}\left(n^{-r_{\gamma_{1}}}+n^{-r_{\gamma_{0}}}\right)^{\frac{\alpha}{\alpha+2}}\right).
\]
Each part of  \eqref{eq:rate.main} and \eqref{eq:rate.w.margin} is interpretable. Term $R_{n,F}$ controls the excess risk even when $\hat{A}_n$ and its oracle version $A_n:=\frac{1}{n}\sum_{i=1}^{n}\xi(Z_{i},\eta_0)p(W_{i})p(W_{i})^{\prime}$ 
do not ``behave nicely'' (i.e., when they are not positive definite). When they do ``behave nicely'', consider the following oracle ``empirical risk minimization''
(ERM) problem with known $\eta_0$:
\begin{equation}
\min_{\delta\in\mathcal{D}}L_{n}^{o}(\delta,\eta_0),\label{eq:oracle}
\end{equation}
where
\begin{align}
L_{n}^{o}(\delta,\eta_0):=\frac{1}{n}\sum_{i=1}^{n}\left[\xi(Z_i,\eta_0)\left(\mathbf{1}\{\gamma_{1}(X_i)-\gamma_{0}(X_i)\geq0\}-\delta(W_i)\right)^{2}\right]\nonumber\label{eq:empirical.average.debiased}.
\end{align}
The oracle excess risk is of  $O\left(R_{n,O}\right)$,
containing an approximation error term $\sup_{P^{n}}\inf_{\delta\in\mathcal{D}}\left[L(\delta,\tau)-L(\delta^{*},\tau)\right]$ due to  using $\mathcal{D}$ to approximate $\delta^*$.\footnote{This approximation quality term may depend on whether Assumption \ref{asm:margin} is imposed or not, and may be further analyzed provided with suitable smoothness conditions on $\delta^*$, which we leave for future research. 
}  
Since $\eta_0$ is in fact unknown and needs to be estimated, we have to pay an additional price from the ``remainder estimation error''. Interestingly, the asymptotic order of this remainder error depends on whether the margin condition is imposed or not. Without margin condition, the ``remainder estimation error'' is $O\left(R_{n,B}+R_{n,V}\right)$, where $R_{n,B}$ is a bias term 
while $R_{n,V}$ is a variance term. 
If  the margin condition is imposed 
with some $\alpha>0$, the rate for the variance term improves to $O(R_{n,V}\left(n^{-r_{\gamma_{1}}}+n^{-r_{\gamma_{0}}}\right)^{\frac{\alpha}{\alpha+2}})$.

The optimality of our proposal depends on the complexity of $\delta^*$. If there exists some $\delta\in\mathcal{D}$ that solves \eqref{eq:loss.proposal} with $\sieved$ fixed, we say $\delta^{*}$ is parametric. If no rule in $\mathcal{D}$ solves \eqref{eq:loss.proposal}, we say $\delta^{*}$ is nonparametric. 
The following proposition suggests that,  when  $\delta^*$ is parametric,  our procedure 
is asymptotically optimal in terms of the rate with Assumptions \ref{asm:unconfounded}-\ref{asm:stability}. Moreover, it is also asymptotically semiparametrically efficient with the additional Assumption \ref{asm:margin}.

\begin{prop}\label{prop:parametric}
Consider the case when $\delta^*$ is parametric. 
\begin{itemize}
\item[(i)] Suppose Assumptions \ref{asm:unconfounded}, \ref{asm:reg}, \ref{asm:quality} and \ref{asm:stability} hold true, $r_{\gamma_{1}}>1/2$, $r_{\gamma_{0}}>1/2$, $r_{\omega_{1}}+r_{\gamma_{1}}>1$ and $r_{\omega_{0}}+r_{\gamma_{0}}>1$. 
Then,  $R_{n,O}=R_{n,V}=O(\frac{1}{n})$, $R_{n,B}=o(\frac{1}{n})$, and 
\begin{equation}\label{eq:rate.parametric}
\sup_{P_{n}}\mathbb{E}_{P_{n}}\left[L(\hat{\delta}^{\mathcal{T}},\tau)-L(\delta^{*},\tau)\right]=O\left(\frac{1}{n}\right).
\end{equation}

\item[(ii)] If in addition, Assumption \ref{asm:margin} also holds, then $R_{n,O}=O(\frac{1}{n})$, $R_{n,B}=R_{n,V}=o(\frac{1}{n})$ and \eqref{eq:rate.parametric} is still true. Moreover, suppose $\Omega :=A^{-1}VA^{-1}$ is  positive definite, where $V:=\mathbb{E}\left[SS^{\prime} \right]$, 
\begin{align*}
S&:=\xi(Z)p(W)\mathbf{1}\{\tau(X)\geq 0\}-\mathbb{E}[\xi(Z)p(W)\mathbf{1}\{\tau(X)\geq 0\}].
\end{align*}
Then, as $n\rightarrow\infty$,
\[
n\left(L(\hat{\delta}^{\mathcal{T}},\tau)-L(\delta^{*},\tau)\right)\overset{d}{\rightarrow}N(0,\Omega)^{\prime}AN(0,\Omega),
\]
where $N(0,\Omega)$ denotes a multivariate normal distribution with mean $0$ and covariance matrix $\Omega$.
\end{itemize}
\end{prop}

Note if $\delta^*$ is parametric, Assumption \ref{asm:reg}(ii) implies a unique $\beta^*\in\mathbb{R}^{\sieved}$ such that $\left(\beta^{*}\right)^{\prime}p(w)$ solves \eqref{eq:loss.proposal}. The study of $\beta^*$ has a known semiparametric efficiency bound $\Omega$. See e.g., \cite*{newey1994asymptotic,ackerberg2014asymptotic}.  By Proposition \ref{prop:parametric}(ii),  our procedure is asymptotically equivalent to the oracle that solves \eqref{eq:oracle}. In particular, $\sqrt{n}\left(\hat{\beta}-\beta^{*}\right)\overset{d}{\rightarrow}N(0,\Omega)$, achieving the semiparametric efficiency bound asymptotically. Moreover, with a parametric $\delta^*$,
\[n\left(L(\hat{\delta},\tau)-L(\delta^{*},\tau)\right)	=n\left(\hat{\beta}-\beta^{*}\right)^{\prime}A\left(\hat{\beta}-\beta^{*}\right).\] 
The asymptotic efficiency of $\beta^*$ translates to that of $L(\hat{\delta},\tau)$, implying that our  procedure is asymptotically efficient as well.\footnote{A parametric $\delta^*$ is not necessarily restrictive. Even if $\delta^*$ is nonparametric, one may wish to target the ``second best'' rule, i.e., $\delta^{SB}\in\arg\inf_{\delta\in\mathcal{D}}L(\delta,\tau)$, for which Proposition \ref{prop:parametric} can be shown to still apply.}

When $\delta^*$ is nonparametric, the discussion of the optimality of our procedure is more involved.   In Appendix \ref{sec:lower.bound}, we derive a minimax lower bound for $\delta^*$ in the style of \cite{stone1982optimal}, which we suspect is attainable by our procedure for certain high smoothness class of $\delta^*$ when $d_p$ grows sufficiently slowly. We leave the verification of this conjecture, as well as the asymptotic distribution of $(L(\hat{\delta},\tau)-L(\delta^{*},\tau))$ for future research. 

\begin{rem}\label{rem:proof.strategy}
The proof strategy of Theorem \ref{thm:main} is significantly different from the existing approaches in the policy learning literature (c.f.  \citealt{kitagawa2018should,AW20}). For the oracle problem, one may follow the classic theory developed by \cite{vapnik1971uniform,vapnik1974theory} to bound
\[
\sup_{\delta\in\mathcal{D}}\mathbb{E}_{P^{n}}\left|L_{n}^{o}(\delta,\eta_0)-L^o(\delta,\eta_0)\right|,
\]
resulting in an order of
$O(1/\sqrt{n})$ in general even in the case of a parametric $\delta^*$. Instead, we exploit the weighted least squares structure embedded in $L^o_{n}$ and adapt (in Lemma \ref{lem:main.1}) a refined maximal inequality developed by \citet[][Theorem 2]{KOHLER20001}, leading to a  convergence rate for the oracle that can be as fast as $O(1/n)$. For the remainder estimation error part, one possibility is to follow 
\citet[][Section 3.2]{AW20} and control the estimation error uniformly over all rules in $\mathcal{D}$. This approach, however, would only lead to a rate of $o(1/\sqrt{n})$ even with a parametric $\delta^*$,  much slower than our result of  $O(R_{B,n}+R_{V,n})$ even without margin condition. We, instead, utilize the fact that both \eqref{eq:our.proposal} and \eqref{eq:oracle} have explicit solutions in large sample that  satisfy certain first order optimality conditions, which allows us to derive a faster rate (Lemma \ref{lem:main.2}). These nice structures for the remainder estimation errors are only valid in large samples. Therefore, our results are asymptotic in nature, as opposed to the finite sample performance guarantee in \cite{kitagawa2018should}.
\end{rem}

\begin{rem} \label{rem:plug.in}Currently, it is not  entirely clear  to what extent our debiased approach is strictly needed for some of the results in Theorem \ref{thm:main} to hold. Indeed, debiasing is costly: $\hat{L}_n^o(\delta)$ is a low-bias, but more noisy estimator of $L(\delta, \tau)$ due to the indefiniteness of $\hat{A}_n$, which may compromise the finite-sample performance of debiasing. A natural alternative is to solve 
\begin{equation}
\inf_{\delta\in\mathcal{D}}\hat{L}_{n}(\delta),\label{eq:plug-in.1}
\end{equation}
where 
\begin{equation}
\hat{L}_{n}(\delta):=\frac{1}{n}\sum_{i=1}^{n}\hat{\tau}^2(X_{i})\left(\mathbf{1}\{\hat{\tau}(X_{i})\geq0\}-\delta(W_{i})\right)^{2},  
\end{equation}
and $\hat{\tau}$ is any estimator of $\tau$ that may or may not be cross-fitted. This plug-in approach maintains the positive semi-definiteness of the associated Hessian matrix. It is straightforward to extend our theory and establish the oracle rate for this plug-in approach, which would be the same as $R_{n,O}$.  Analogous analyses (c.f. proof of Lemma \ref{lem:main.2}) imply that the remainder estimation error is determined asymptotically by 
\begin{equation}\label{eq:plug.in.bias}
\left(\mathbb{E}_{P^{n}}\left\Vert \hat{A}_{n}^{\text{plug-in}}-A_{n}^{\text{plug-in}}\right\Vert ^{2}+\mathbb{E}_{P^{n}}\left\Vert \hat{B}_{n}^{\text{plug-in}}-B_{n}^{\text{plug-in}}\right\Vert ^{2}\right),   
\end{equation}
where
\begin{align*}
A_{n}^{\text{{plug-in}}}&:=\frac{1}{n}\sum_{i=1}^{n}\tau^{2}(X_{i})p(W_{i})p(W_{i})^{\prime},\\
\hat{A}_{n}^{\text{{plug-in}}}&:=\frac{1}{n}\sum_{i=1}^{n}\hat{\tau}^{2}(X_{i})p(W_{i})p(W_{i})^{\prime},\\
B_{n}^{\text{{plug-in}}}&:=\frac{1}{n}\sum_{i=1}^{n}\tau^{2}(X_{i})p(W_{i})\mathbf{1}\left\{ \tau(X_{i})\geq0\right\},\\
\hat{B}_{n}^{\text{{plug-in}}}&:=\frac{1}{n}\sum_{i=1}^{n}\hat{\tau}^{2}(X_{i})p_(W_{i})\mathbf{1}\left\{ \hat{\tau}(X_{i})\geq0\right\}.
\end{align*}
If cross-fitting is used, \eqref{eq:plug.in.bias} in general presents an asymptotic bias larger  than $R_{B,n}$ (c.f. Lemmas \ref{lem:remainder.5} and \ref{lem:remainder.6}). However, if cross-fitting is not used and conditional on the specific structure of the estimator $\hat{\tau}$, we suspect the asymptotic bias in \eqref{eq:plug.in.bias}  may be as fast as $R_{B,n}$, in light of the classic ``low bias'' results of certain plug-in approaches in the semiparametric estimation literature, e.g., \cite*{ai2003efficient,chen2003estimation,hirano2003efficient}.  Whether the plug-in approach preserves the same asymptotic remainder estimation error is an intricate but fascinating question that we leave for future research. In the empirical applications below, we present results with our main debiased approach as well as the plug-in alternative. 
\end{rem}

\section{Capacity constraint}\label{sec:capacity}

In this section, we consider a decision maker facing convex constraints for the allocation rules. As a leading case, suppose $W$ is discrete and a capacity constraint exists on how many people in the population can get treatment. With such capacity constraint, the problem is convex with
differentiable objective and constraint functions,  and the Slater's condition can be verified to hold. Therefore, the optimal solution is characterized by the well-known KKT condition (e.g., \citealt{boyd2004convex}, Chapter
5, p.244), as we show below.  
\begin{prop}\label{prop:capacity}
Suppose $W$ is discrete and takes values $\left\{ w_{j}\right\} _{j=1}^{J}$
with corresponding probabilities $\left\{ p_{j}\right\} _{j=1}^{J}$, where
$p_{j}>0$ for all $j=1,\ldots, J$. Consider solving (\ref{eq:population.optimal}) with $\alpha=2$
and a capacity constraint $\mathbb{E}[\delta(W)]\leq t$ for some
$0<t<1$. Let 
\begin{align*}
b_{j} & :=\mathbb{E}\left[\tau^{2}(X)\mathbf{1}\left\{ \tau(X)\geq0\right\} \mid W=w_{j}\right],\\
a_{j} & :=\mathbb{E}\left[\tau^{2}(X)\mid W=w_{j}\right].
\end{align*}
Wlog, suppose $a_{j}>0$ for all $j=1\ldots J$\footnote{The case of $a_{j}=0$ can be excluded as an action of 0 would be optimal and  not add to the capacity.}, and index groups
so that $b_{1}\geq b_{2}\ldots\geq b_{J}$. If the capacity constraint
is not binding (i.e., $\sum\limits _{j=1}^{J}(p_{j}b_{j}/a_{j})\leq t$),
then the unconstrained solution $\left\{ {b_{j}}/{a_{j}}\right\} _{j=1}^{J}$
is optimal. Otherwise, the optimal decision is 
\begin{align*}
\delta_{j}^{*}  =\frac{b_{j}}{a_{j}}-\frac{\lambda^{*}}{2a_{j}},\text{ for all }j\leq J^{*},\quad
\delta_{j}^{*} & =0,\text{ for all }j>J^{*},
\end{align*}
where $J^{*}\in\left\{ 1,\ldots,J\right\} $ and $\lambda^{*}\geq0$
are jointly determined such that
\[
\lambda^{*}=\frac{\sum_{j=1}^{J^{*}}\frac{p_{j}b_{j}}{a_{j}}-t}{\sum_{j=1}^{J^{*}}\frac{p_{j}}{2a_{j}}}.
\]
\end{prop}
Proposition \ref{prop:capacity} highlights an interesting insight:   with a capacity constraint, a regret-averse decision maker would reduce the fractional treatment for \emph{all} groups, possibly with some groups with smallest $b_j$ not treated at all if the capacity constraint is too severe. In contrast, when $\alpha=1$, the decision maker always prioritizes treating the $W$ groups with the largest  positive average treatment effect until the capacity constraint is filled, possibly with fractional allocation for the marginal group. In the hypothetical policy question  from \cite{resnjanskij2024can} considered in the introduction, suppose we have a capacity constraint that at most a $t\leq1$ fraction of the population can be offered with the mentoring program. Since there is only one $W$ group whose average treatment effect is positive, the optimal constrained rule is easy to calculate (see Table \ref{tab:constraint}). For example, if $\alpha=1$, the optimal rule is to treat $t$ fraction of the population; if $\alpha=2$, the optimal rule is to treat $t$ fraction of the population if $t<0.88$ and to treat 0.88 of the population if $t\geq0.88$ (as 0.88 is the unconstrained optimal which does not violate the capacity constraint). In this simple case with one $W$ group, $\alpha=1$ and $2$  would share the same optimal rule if $t<0.88$. 
\begin{table}[htbp]
\centering
\caption{Optimal allocation rule for the hypothetical policy question in \cite{resnjanskij2024can} with a capacity constraint}\label{tab:constraint}
\begin{tabular}{lccc}
\toprule
& \multicolumn{3}{c}{Atkinson inequality index} \\
 \cmidrule(lr){2-4}
Capacity constraint &  $\alpha=1$ & $\alpha=2$ & $\alpha=3$ \\
\midrule
$t\in[0.88,1]$  & $t$ & 0.88 & 0.82 \\
$t\in[0.82,0.88)$ &  $t$ &  $t$ & 0.82 \\
$t\in[0,0.82)$ &  $t$ &  $t$ &  $t$ \\
\bottomrule
\end{tabular}
\end{table}

In the setup of Proposition \ref{prop:capacity}, we can still learn the optimal constrained rule from data by solving \eqref{eq:our.proposal} and incorporating  additional constraints:\footnote{We do not need to impose  the constraints that  $\delta(w_{j})\leq1$ for $j=1,...,J$, as they  will be non-binding  at the true population constrained optimal rule.}
\begin{equation}\label{eq:add.constraint}
\frac{1}{n}\sum_{i=1}^{n}\delta(W_{i})\leq t,\delta(w_{j})\geq0,j=1,...,J,
\end{equation}
which is still a convex program with  differentiable objective and constraints and can be efficiently computed. However, establishing the statistical performance guarantee is more involved due to the known technical difficulty associated with not knowing  whether the constraints in \eqref{eq:add.constraint} are binding or not in general.

\section{Empirical applications}\label{sec:emp.app}

\subsection{Job Training Partnership Act (JTPA) Study}\label{sec:JTPA}

We revisit the experimental dataset of the National JTPA Study that aimed to
measure the benefit and cost of employment and training programs. Our sample consists of 9223 observations, in which the treatment $D$ was randomized to generate the applicants'  eligibility for receiving a mix of training, job-search assistance, and other
services provided by the JTPA. The outcome of interest $Y$ is the total individual earnings in the 30 months after program
assignment.\footnote{We take the intention-to-treat perspective. One may also consider an net-of-cost outcome, which would further deduct 774 dollars for each of  those assigned to treatment.} The study also collected a variety of the applicants' background information ($X$), some of which might be perceived  as sensitive, e.g.,  gender, race and marital
status. Following \cite{kitagawa2018should}, we consider a scenario in which a policymaker can only design treatment policies based on pre-program
years of education (``education'') and the pre-program earnings (``income'') --- these two variables become the $W$ in our setup.  As an illustration of our debiased approach, we choose $K=5$ and  estimate $\gamma_1$ and  $\gamma_0$ via lasso with 10-fold cross-validation, with all
interactions and squared terms of $X$. We estimate $\omega_1$ and $\omega_0$ with the minimum distance estimator with a Tikhonov
penalty \eqref{eq:minimum.distance.penalty}, where the tuning parameter is selected via cross validation. See Section \ref {sec:compute} for computational details and our algorithm to calculate $\hat{\xi}(Z_i)$.
\begin{figure}[!t]]
\includegraphics[width=0.49\linewidth]{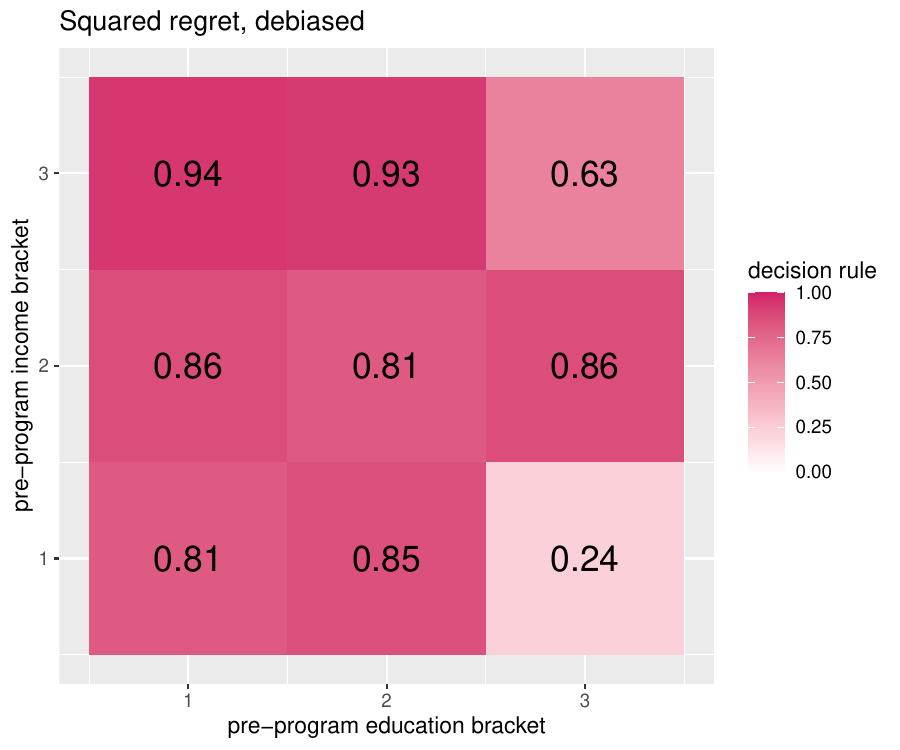}
\includegraphics[width=0.49\linewidth]{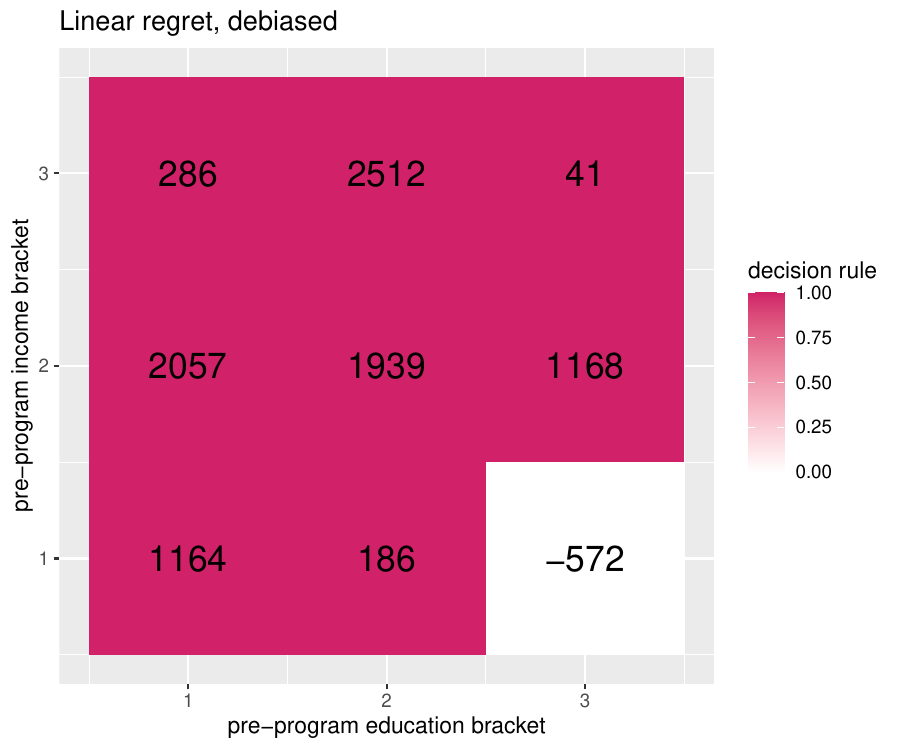}
\medskip
\includegraphics[width=0.49\linewidth]{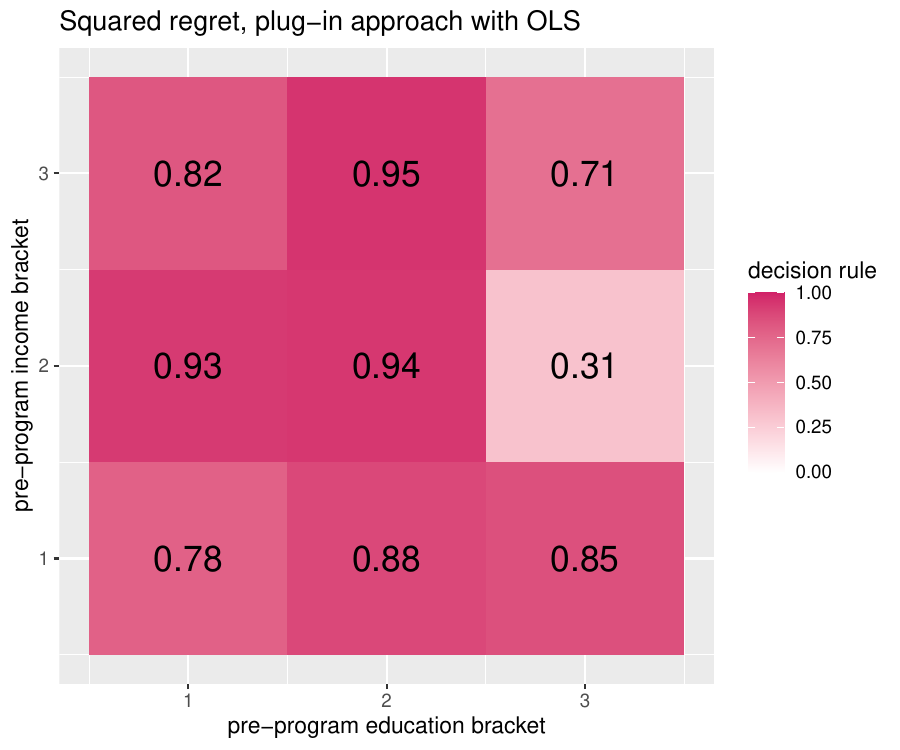}
\includegraphics[width=0.49\linewidth]{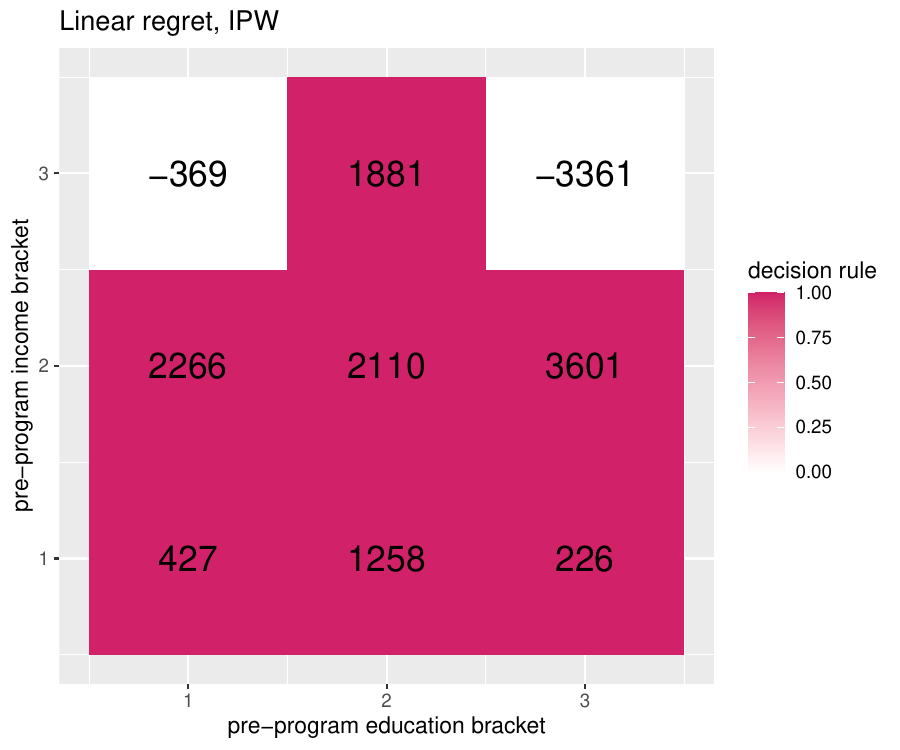}
{\footnotesize \textit{Notes}: Income brackets are defined according to pre-program  earnings as follows: 1 ($\leq\$220$),  2 ($>\$220$ and  $\leq\$3800$),  3 ($>\$3800$). Education brackets are defined according to pre-program years of schooling: 1 ($\leq11$, high school dropouts); 2 ($=12$,  high school graduates); 3 ($>12$, with higher education. Top left: our squared-regret approach with debiasing; Top right: linear regret approach, with $CATE(W)$ estimated with debiasing  and $\eta_0$ fitted with lasso. Down left: squared-regret approach with $\gamma_1$ and  $\gamma_0$ estimated by OLS; Down right: linear regret approach, with $CATE(W)$ estimated with inverse propensity score weighting with the known propensity score of $2/3$. The numbers in each of the brackets in the left two graphs refer to the corresponding estimated treatment assignment fractions, while the numbers in each of the brackets in the right two graphs refer to the estimated $CATE(W)$.
}
\caption{JTPA: estimated simple bracket rules }\label{fig:JTPA.bracket}
\end{figure}

To start with, suppose the policymaker is interested in implementing a simple rule based on nine pre-determined income and education brackets (defined in the note of Figure \ref{fig:JTPA.bracket}). In this case, $W$ is discrete, and the optimal rule can be solved bracket-by-bracket. Figure \ref{fig:JTPA.bracket} reports the results for our squared regret debiased approach,   a squared regret plug-in approach, as well as two  linear regret approaches. 
Although the majority of the estimated CATEs conditional on $W$ are positive, the fractional nature of our estimated policies reveals plausible and considerable treatment effect heterogeneity at the $X$ level for some brackets, demonstrating the value of our approach compared to the standard mean regret paradigm. For example, for those units in education bracket 3 and income bracket 3, the debiased CATE estimate is slightly positive (41), implying all units shall be treated. However, an IPW estimate of the same CATE (-3361) would imply that no-one should be treated. For this group of workers, our squared regret debiased optimal policy is 0.63, indicating that workers in the high-education and high-income bracket may display drastically different treatment effects 
from each other, which can lead to a high regret-inequality should a non-fractional policy be applied. The pattern of the squared-regret policy estimates between the plug-in and debiased approaches are similar for many brackets, although some disparities do exist.

Next, we consider a policymaker designing a class of linear sieve policies based on education and income. As an illustration, for each of the education and income variables, we create cubic B-splines with a total of 5 degrees of freedom. The multivariate B-splines are then generated as tensor products of the two. We present estimated policies of the debiased and plug-in approaches for selected values of the income and education variables in Figure \ref{fig:JTPA.heatmap}. Both approaches again indicate considerable effect heterogeneity in the population, although  disparities remain in the exact fitted values for some $W$ groups. 

\begin{figure}[http]
\centering
\includegraphics[width=0.49\linewidth]{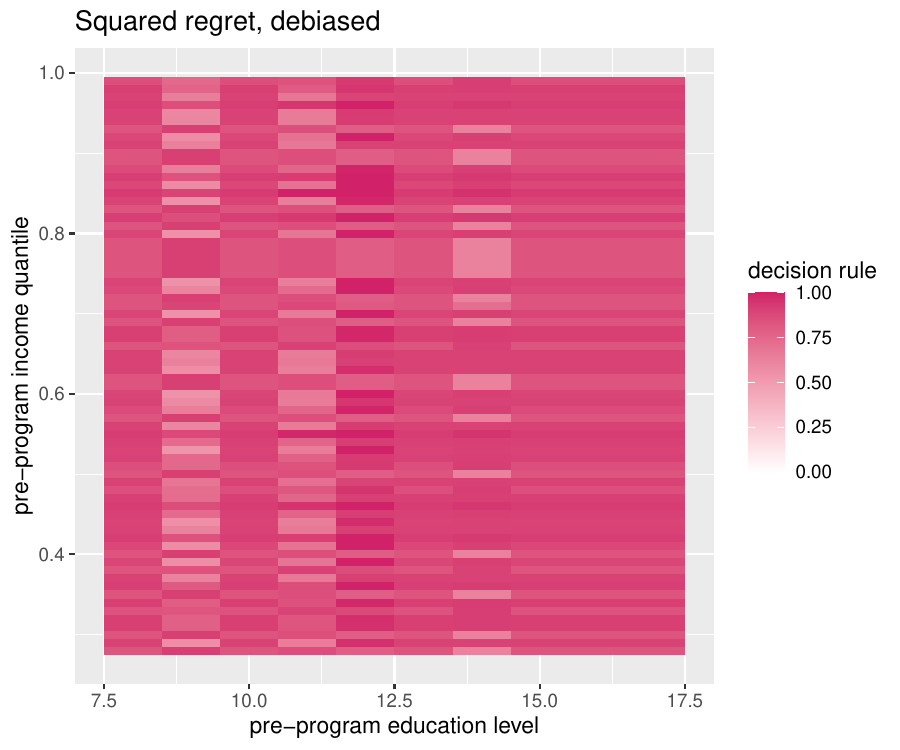}
\includegraphics[width=0.49\linewidth]{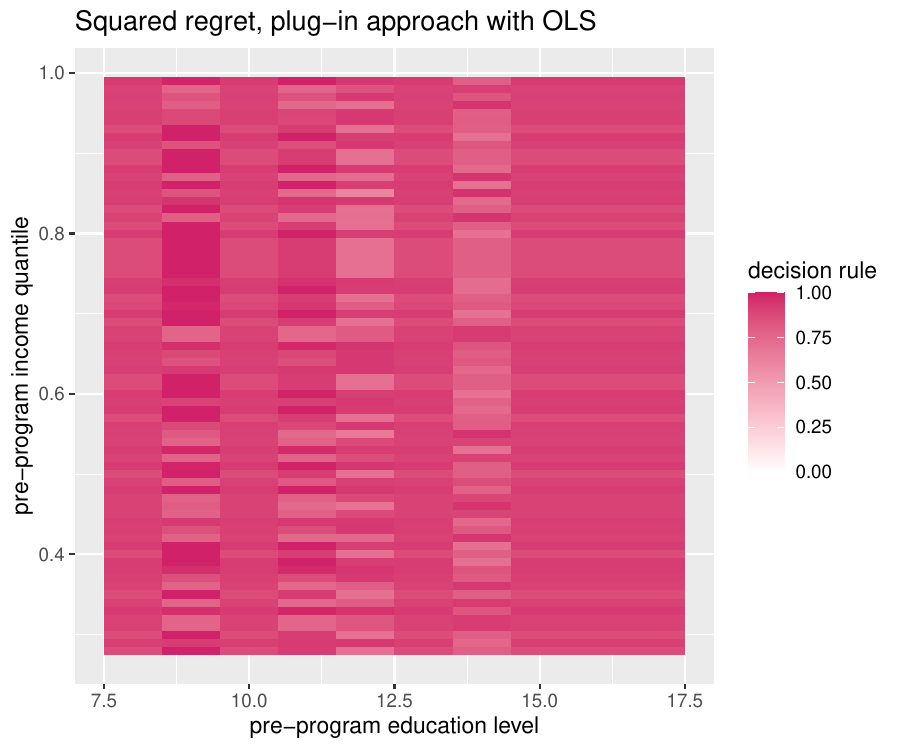}
\caption{JTPA: estimated linear sieve policy rules with multivariate B-splines}
\label{fig:JTPA.heatmap}
\end{figure}

\subsection{International Stroke Trial}
As a second application and to demonstrate the value of our approach in medical studies, we analyze the International Stroke Trial (IST, \citealt{international1997international}) that assessed the effect of aspirin and other treatments  for patients with presumed acute ischemic stroke. Following \cite*{yadlowsky2025evaluating}, we focus on the treatment of aspirin only on the outcome of whether there is death
or dependency at 6 months. This leaves us with a sample of 18304 patients  from over 30 countries. For each patient, we also observe a vector of 39 covariates ($X$), including  their gender, age as well as some of their medical history and geographical information. In this exercise, we consider a hypothetical scenario in which a doctor determines whether a patient should be treated with aspirin only based on their age ($W$). The aim is to assess whether our approach would generate significantly different treatment fractions compared to the mean regret approach.

\begin{figure}[http]
\centering
\includegraphics[width=0.7\linewidth]{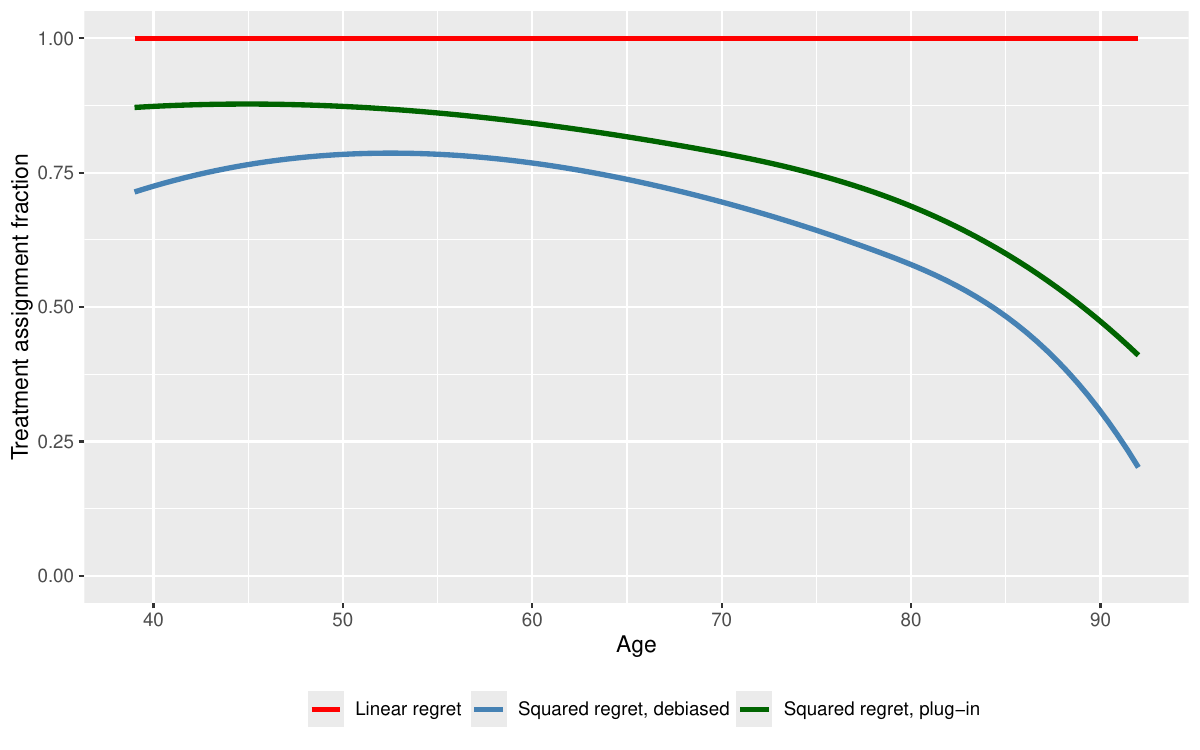}
    \caption{IST: estimated optimal treatment fractions based on age with B-splines}
    \label{fig:rate}
\end{figure}

We estimate the nuisance parameters with the same methodology described in Section \ref{sec:JTPA}. For the age variable, we create a cubic B-spline with a total degree freedom of 6. Figure \ref{fig:rate} reports our estimated optimal treatment fractions for patients with age between 39 and 92.  As the CATEs are all positive for all considered age groups, a linear regret approach will recommend to treat everyone for all age groups. In sharp contrast, the estimated treatment fractions with our debiased approach is between  25\% and 75\% for most age values, revealing considerable treatment heterogeneity among those sharing the same ages. The treatment proportion is especially close to 0.5 for patients with age between 75 to 85, suggesting that  a singleton ``treat everyone''  rule would potentially harm significantly some of those patients, leaving some of them with large regrets. The fitted curve with the plug-in approach shares the same downward sloping pattern as the debiased approach, although the estimated treatment proportions is slightly higher for all age groups. In light of these findings, we think that our squared regret approach to policy learning reveals additional important information that cannot be assessed with the mean regret approach alone.

\bibliographystyle{ecta}
\bibliography{bib}

\appendix

\section{Additional results on Theorem \ref{thm:main}}\label{sec:technical.discussion}

We now discuss the main proof steps of Theorem \ref{thm:main}.
\paragraph*{Step 1} Preparations. To ease notational burden, for any function
$f$, let $f_{i}:=f(Z_{i})$, and write 
\begin{align*}
\hat{A}_{n} & =\frac{1}{n}\sum_{i=1}^{n}\hat{\xi}_{i}p_{i}p_{i}^{\prime}, & \hat{B}_{n}=\frac{1}{n}\sum_{i=1}^{n}\hat{\xi}_{i}p_{i}\mathbf{1}\left\{ \hat{\tau}_{i}\geq0\right\} ,\\
A_{n} & =\frac{1}{n}\sum_{i=1}^{n}\xi_{i}p_{i}p_{i}^{\prime}, & B_{n}=\frac{1}{n}\sum_{i=1}^{n}\xi_{i}p_{i}\mathbf{1}\left\{ \tau_{i}\geq0\right\} ,\\
A & =\mathbb{E}\left[\xi_{i}p_{i}p_{i}^{\prime}\right].
\end{align*}
Pick any $0<\varepsilon<\min\left\{ \underline{\lambda},6\tilde{C}_{\xi}\zeta_{p}^{2},3C_{\xi}\zeta_{p}^{2}/2\right\} $.
Define event 
\[
\mathcal{E}_{n}:=\left\{ \lambda_{\min}\left(A_{n}\right)\geq\underline{\lambda}-\varepsilon,\lambda_{\min}\left(\hat{A}_{n}\right)\geq\underline{\lambda}-\varepsilon\right\} .
\]
On event $\mathcal{E}_{n}$, note the problem of (\ref{eq:our.proposal})
is convex with a unique solution $\hat{\delta}(w)=\hat{\beta}^{\prime}p(w)$, where
$\hat{\beta}=\hat{A}_{n}^{-1}\hat{B}_{n}$. Furthermore, the oracle
problem  $\min_{\delta=\beta^{\prime}p,\beta\in\mathbb{R}^{\sieved}}L_{n}^{o}(\delta)$
is also convex with a unique solution $\tilde{\delta}(w)=\tilde{\beta}^{\prime}p(w)$,
$\tilde{\beta}=A_{n}^{-1}B_{n}$. 
Next, we decompose 
\begin{align*}
 & \mathbb{E}_{P_{n}}\left[L(\hat{\delta}^{\mathcal{T}},\tau)-L(\delta^{*},\tau)\right]=\mathcal{P}_{1}E_{1,n}+\left(1-\mathcal{P}_{1}\right)E_{2,n},
\end{align*}
where 
\begin{align*}
\mathcal{P}_{1}:= & Pr\{\mathcal{E}_{n}\}\leq1,\\
E_{1,n}:= & \mathbb{E}_{P_{n}}\left[L(\hat{\delta}^{\mathcal{T}},\tau)-L(\delta^{*},\tau)\mid\mathcal{E}_{n}\text{ holds}\right],\\
E_{2,n}:= & \mathbb{E}_{P_{n}}\left[L(\hat{\delta}^{\mathcal{T}},\tau)-L(\delta^{*},\tau)\mid\mathcal{E}_{n}\text{ does not hold}\right].
\end{align*}

\paragraph*{Step 2} We invoke Lemmas \ref{lem:remainder.1} and \ref{lem:remainder.2}
to conclude that, for each $n$ such that 
\[
4C_{M}\zeta_{p}^{2}\left(n^{-r_{\gamma_{1}}}+n^{-r_{\gamma_{0}}}+n^{-\frac{r_{\omega_{1}}+r_{\gamma_{1}}}{2}}+n^{-\frac{r_{\omega_{0}}+r_{\gamma_{0}}}{2}}\right)<\varepsilon,
\]
we have 
\[
P\{\mathcal{E}_{n}\}\geq1-2d_{p}\left[\exp\left(\frac{-n\varepsilon^{2}}{4C_{\xi}^{2}\zeta_{p}^{4}}\right)+K\exp\left(\frac{-n\varepsilon^{2}}{16K\tilde{C}_{\xi}^{2}\zeta_{p}^{4}}\right)\right].
\]
Also, note $E_{2,n}\leq2C_{\xi}$ under Assumptions \ref{asm:unconfounded}, \ref{asm:reg} and due to trimming. Conclude that
\[
\left(1-\mathcal{P}_{1}\right)E_{2,n}\leq4C_{\xi}d_{p}\left[\exp\left(\frac{-n\varepsilon^{2}}{4C_{\xi}^{2}\zeta_{p}^{4}}\right)+K\exp\left(\frac{-n\varepsilon^{2}}{16K\tilde{C}_{\xi}^{2}\zeta_{p}^{4}}\right)\right].
\]
\paragraph*{Step 3} 
We now bound $E_{1,n}$. To simplify notation, for the rest of the paper, whenever we use ``$\mathbb{E}_{P^{n}}[\cdot]$'' on event $\mathcal{E}_n$, we  mean ``$\mathbb{E}_{P^{n}}[\cdot\mid\mathcal{E}_{n}\text{ holds}]$''. Note
for each $x\in\mathcal{X}$, 
\[
\tau^{2}(x)\left(\mathbf{1}\left\{ \tau(x)\geq0\right\} -\hat{\delta}^{\mathcal{T}}(w)\right)^{2}\leq\tau^{2}(x)\left(\mathbf{1}\left\{ \tau(x)\geq0\right\} -\hat{\delta}(w)\right)^{2}.
\]
Therefore, $L(\hat{\delta}^{\mathcal{T}},\tau)\leq L(\hat{\delta},\tau)$, implying on event $\mathcal{E}_n$,
\[
\mathbb{E}_{P^{n}}\left[L(\hat{\delta}^{\mathcal{T}},\tau)-L(\delta^{*},\tau)\right]\leq\mathbb{E}_{P^{n}}\left[L(\hat{\delta},\tau)-L(\delta^{*},\tau)\right].
\]
Therefore, it suffices to bound $\mathbb{E}_{P^{n}}\left[L(\hat{\delta},\tau)-L(\delta^{*},\tau)\right]$ on event $\mathcal{E}_n$. To this end, observe
\begin{align}
 & L(\hat{\delta},\tau)-L(\delta^{*},\tau)\nonumber \\
= & \underset{T_{1n}}{\underbrace{L(\hat{\delta},\tau)-L(\delta^{*},\tau)-2\left[L_{n}^{o}(\hat{\delta},\eta_{0})-L_{n}^{o}(\delta^{*},\eta_{0})\right]}}+2\left(\underset{T_{2n}}{\underbrace{L_{n}^{o}(\tilde{\delta},\eta_{0})-L_{n}^{o}(\delta^{*},\eta_{0})}}\right)\nonumber \\
+ & 2\left(\underset{T_{3n}}{\underbrace{L_{n}^{o}(\hat{\delta},\eta_{0})-L_{n}^{o}(\tilde{\delta},\eta_{0})}}\right),\label{pf:lem1.1}
\end{align}
where for term $T_{2n}$, we have $L_{n}^{o}(\tilde{\delta},\eta_{0})=\inf_{\delta\in\mathcal{D}}L_{n}^{o}(\delta,\eta_{0})$
on event $\mathcal{E}_{n}$. Therefore,
on event $\mathcal{E}_{n}$,
\begin{align}
\mathbb{E}_{P^{n}}[T_{2n}] & =\mathbb{E}_{P^{n}}\left[\inf_{\delta\in\mathcal{D}}L_{n}^{o}(\delta,\eta_{0})-L_{n}^{o}(\delta^{*},\eta_{0})\right]\nonumber \\
 & \leq\inf_{\delta\in\mathcal{D}}\left\{ \mathbb{E}_{P^{n}}\left[L_{n}^{o}(\delta,\eta_{0})-L_{n}^{o}(\delta^{*},\eta_{0})\right]\right\} \nonumber \\
 & =\underset{\text{approximation error}}{\underbrace{\inf_{\delta\in\mathcal{D}}\left[L(\delta,\tau)-L(\delta^{*},\tau)\right]}},\label{pf:lem1.2}
\end{align}
where the equality used the observation that $L^{o}(\delta,\eta_{0})=L(\delta,\tau)$
for all $\delta\in\mathcal{D}$ and for $\delta^{*}$ as well. Conclude with the following lemma.

\begin{lem}
\label{lem:basic} On event $\mathcal{E}_{n}$, the following holds
for each $P^{n}\in\mathcal{P}^{n}$: 
\begin{align*}
\mathbb{E}_{P^{n}}\left[L(\hat{\delta},\tau)-L(\delta^{*},\tau)\right]
\leq \underset{\text{oracle rate}}{\underbrace{\mathbb{E}_{P^{n}}\left[T_{1n}\right]+2\sup_{P^{n}}\inf_{\delta\in\mathcal{D}}\left[L(\delta,\tau)-L(\delta^{*},\tau)\right]}+}\underset{\text{remainder estimation error}}{2\underbrace{\mathbb{E}_{P^{n}}\left[T_{3n}\right]}}.
\end{align*}
\end{lem}
\paragraph*{Step 4} We bound $\mathbb{E}_{P^{n}}\left[T_{1n}\right]$ and $\mathbb{E}_{P^{n}}\left[T_{3n}\right]$
on event $\mathcal{E}_n$ by establishing the following two lemmas below, and the conclusion
of the theorem follows immediately. 
\begin{lem}
\label{lem:main.1} Under Assumptions \ref{asm:unconfounded}-\ref{asm:stability}
and on event $\mathcal{E}_{n}$, we have, for some constant $\mathcal{C}_{1}$,
\begin{equation}
\sup_{P^{n}\in\mathcal{P}^{n}}\mathbb{E}_{P^{n}}\left[T_{1n}\right]\leq\mathcal{C}_{1}\frac{\sievedstar}{n}.\label{eq:fast.rate}
\end{equation}
\end{lem}
\begin{lem}
\label{lem:main.2} Under Assumptions \ref{asm:unconfounded}-\ref{asm:quality}
and on event $\mathcal{E}_{n}$, the following statements hold:
\begin{itemize}
\item[(i)] For some constant $\mathcal{C}_{2}$, we have
\[
\sup_{P_{n}}\mathbb{E}_{P_{n}}\left[T_{3n}\right]\leq\mathcal{C}_{2}\left(R_{n,B}+R_{n,V}^{\text{}}\right).
\]
\item[(ii)] The above rate improves to

\[
\sup_{P_{n}}\mathbb{E}_{P_{n}}\left[T_{3n}\right]\leq\mathcal{C}_{3}\left(R_{n,B}+R_{n,V}\left(n^{-r_{\gamma_{1}}}+n^{-r_{\gamma_{0}}}\right)^{\frac{\alpha}{\alpha+2}}\right),
\]
with a suitably refined constant $\mathcal{C}_{3}$, if in addition,
Assumption \ref{asm:margin} holds and $n$ is also such that 
\[
\left(4C_{M}C_{\tau}^{-1}\left(n^{-r_{\gamma_{1}}}+n^{-r_{\gamma_{0}}}\right)\right)^{\frac{1}{\alpha+2}}<t^{*}.
\]

\end{itemize}
\end{lem}

For completeness, we also restate
the result of \citet[Theorem 2]{KOHLER20001} below. 
\begin{lem}
\label{lem:kohler} Let $Z,Z_{1},\ldots Z_{n}$ be iid random variables
with support $\mathscr{Z}$. Let $K_{1},K_{2}\geq1$ and let $\mathscr{F}$
be a permissible class of functions $f:\mathscr{Z}\rightarrow\mathbb{R}$
such that 
\[
\left|f(z)\right|\leq K_{1},\text{ and }\mathbb{E}f^{2}(Z)\leq K_{2}\mathbb{E}f(Z).
\]
Denote by $N(\epsilon,\mathcal{G},d_{2,n})$ the $\epsilon-$covering
number of function class $\mathcal{G}$ with respect to the empirical
$L_{2}$ distance 
\[
d_{2,n}(g_{1},g_{2}):=\left\{ \frac{1}{n}\sum_{i=1}^{n}\left[g_{1}(Z_{i})-g_{1}(Z_{i})\right]^{2}\right\} ^{1/2},\text{ }g_{1},g_{2}\in\mathcal{G}.
\]
Let $0<\varepsilon<1$ and $\alpha>0$. Assume that 
\[
\sqrt{n}\varepsilon\sqrt{1-\varepsilon}\sqrt{\alpha}\geq288\max\left\{ 2K_{1},\sqrt{2K_{2}}\right\} ,
\]
and for all $z_{1},\ldots,z_{n}\in\mathscr{Z}$ and all $\delta\geq\frac{\alpha}{4}$,
\[
\sqrt{n}\varepsilon\left(1-\varepsilon\right)\delta\gtrsim288\max\left\{ K_{1},2K_{2}\right\} \int_{0}^{\sqrt{\delta}}\left(\log N\left(u,\left\{ f\in\mathscr{F}:\frac{1}{n}\sum_{i=1}^{n}f^{2}(Z_{i})\leq4\delta\right\} ,d_{2,n}\right)\right)^{1/2}du.
\]
Then, 
\begin{align*}
 & Pr\left\{ \sup_{f\in\mathscr{F}}:\frac{\left|\mathbb{E}f(Z)-\frac{1}{n}\sum_{i=1}^{n}f(Z_{i})\right|}{\alpha+\mathbb{E}f(Z)}>\varepsilon\right\}
\leq 50\exp\left(-\frac{n\alpha}{128\cdotp2304\max\left\{ K_{1}^{2},K_{2}\right\} }\right).
\end{align*}
\end{lem}

\section{Structural properties of B-splines}\label{sec:Bspline}
The performance of our proposal depends on the choice of the basis functions via $\sievedstar$, $\zeta_k$ and validity of Assumption \ref{asm:stability}.   If $W$ is bounded, the B-splines basis functions have the following important properties:\footnote{See, e.g.,  Lemmas 14.2,  14.4 and 15.2 in \cite*{gyorfi2006distribution} for a detailed discussion of the properties of univariate and multivariate B-splines.}
\begin{enumerate}
    \item[i)] $p_j(w)\geq0$ for all $j=1,...\sieved$ and all $w\in\mathcal{W}$,
    \item[ii)] $\sum_{j=1}^{\sieved}p_{j}(w)=1$ for all $w\in\mathcal{W}$,
\end{enumerate}
implying  that $\zeta_{p} \leq 1$ for B-splines. With the above two properties, we can verify that Assumption \ref{asm:stability} holds.  Note 
\[\sup_{w\in\mathcal{W}}|\hat{\delta}(w)|\leq\left\Vert \hat{\beta}\right\Vert _{\infty}\sup_{w\in\mathcal{W}}\left(\sum_{j=1}^{\sieved}|p_{j}(w)|\right),\]
where 
$\sup_{w\in\mathcal{W}}\left(\sum_{j=1}^{\sieved}|p_{j}(w)|\right)=1$ by i) and ii) above. Conclude that \[
\left\Vert \hat{B}_{n}\right\Vert _{1}\leq\tilde{C}_{\xi}\sup_{w\in\mathcal{W}}\left(\sum_{j=1}^{\sieved}|p_{j}(w)|\right)\leq\tilde{C}_{\xi}.\] 
Let  $\left\Vert A\right\Vert _{\max}:=\max_{i,j}{\left| a_{ij}\right|}$ for matrix $A=\{a_{ij}\}$. Suppose now $\hat{A}_n$ is positive definite with a  minimum eigenvalue bounded away from a constant $\underline{\lambda}_\varepsilon>0$. Then, we have
\begin{align*}
\left\Vert \hat{\beta}\right\Vert _{\infty} & \leq\left\Vert \hat{A}_{n}^{-1}\right\Vert _{\max}\left\Vert \hat{B}_{n}\right\Vert _{1}\leq\tilde{C}_{\xi}/\underline{\lambda}_\varepsilon.
\end{align*}
Therefore, $\sup_{w\in\mathcal{W}}|\hat{\delta}(w)|\cdot\mathbf{1}\{\lambda_{\text{min}}(\hat{A}_n)\geq\underline{\lambda}_\varepsilon\}\leq C_L:=\tilde{C}_{\xi}/\underline{\lambda}_\varepsilon$ as required by Assumption \ref{asm:stability}.  Moreover, for B-splines, each $f\in\mathcal{D}$ is a multivariate piecewise polynomial  with respect to some partition in $\mathcal{W}$. This implies that $f^{2}$ is also a multivariate piecewise polynomial with respect to the same partition. Thus, $\sievedstar\leq c_{B}\sieved$ for some constant $c_{B}$ that depends on the degree of the polynomial as well as the dimension of $W$. See also  \citet[][proof of Lemma 2 on p.11]{KOHLER20001} for additional discussions.

\section{Proofs of main propositions}

\subsection*{Proof of Proposition \ref{prop:population}}

Since $\delta$ can only condition on $W$, (\ref{eq:population.optimal})
can be solved by conditioning on almost all $w\in\mathcal{W}$:
\begin{equation}
\min_{\delta(w)\in[0,1]}\int\left\{ \tau(x)\left[\mathbf{1}\left\{ \tau(x)\geq0\right\} -\delta(w)\right]\right\} ^{\alpha}dF_{X|W}(x\mid w),\label{eq:population.optimal.conditioning}
\end{equation}
where $F_{X|W}(\cdotp\mid\cdotp)$ denotes the conditional cdf of
$X$ given $W$.

\textbf{Proof of statement (i)}. When $\alpha>1$, the objective function (\ref{eq:population.optimal.conditioning})
is convex and differentiable in $\delta(w)$. If  $\delta^*(w)\in(0,1)$, it 
must satisfy the following FOC:
\begin{equation}
-\int\left\{ \tau(x)\left[\mathbf{1}\left\{ \tau(x)\geq0\right\} -\delta^{*}(w)\right]\right\} ^{\alpha-1}\tau(x)dF_{X|W}(x\mid w)=0,\label{pf:prop.frac.1}
\end{equation}
for almost all $w\in\mathcal{W}$. Moreover, even if $\delta^*(w)\in\{0,1\}$, (\ref{pf:prop.frac.1}) still holds, as the LHS of (\ref{pf:prop.frac.1})
can be written as 
\begin{align*}
- & \int_{\tau(x)>0}\left\{ \tau(x)\left[1-\delta(w)\right]\right\} ^{\alpha-1}\tau(x)dF_{X|W}(x\mid w)\\
- & \int_{\tau(x)<0}\left\{ -\tau(x)\delta(w)\right\} ^{\alpha-1}\tau(x)dF_{X|W}(x\mid w).
\end{align*}
At $\delta(w)=1$, the FOC becomes 
\begin{align}\label{pf:foc.1}
- & \int_{\tau(x)<0}\underset{>0}{\underbrace{\left\{ -\tau(x)\right\} ^{\alpha-1}}}\tau(x)dF_{X|W}(x\mid w)\geq0,
\end{align}
implying $\delta(w)=1$ solves 
\eqref{eq:population.optimal} only when \eqref{pf:prop.frac.1} holds. Analogously, when $\delta(w)=0$, the FOC becomes 
\begin{align}\label{pf:foc.0}
- & \int_{\tau(x)>0}\underset{>0}{\underbrace{\left\{ \tau(x)\right\} ^{\alpha}}}dF_{X|W}(x\mid w)\leq0,
\end{align}
implying $\delta(w)=0$ solves 
\eqref{eq:population.optimal} also only when \eqref{pf:prop.frac.1} holds. Finally, note the inequality in \eqref{pf:foc.1} is strict unless $\int_{\tau(x)<0}dF_{X|W}(x\mid w)=0$, and  the inequality in \eqref{pf:foc.0} is strict unless $\int_{\tau(x)>0}dF_{X|W}(x\mid w)=0$.
This completes the proof for statement (i). 

\textbf{Proof of statement (ii)}. When $\alpha=1$, (\ref{eq:population.optimal.conditioning})
is written as 
\[
\min_{\delta(w)\in[0,1]}\int\tau(x)\mathbf{1}\left\{ \tau(x)\geq0\right\} dF_{X|W}(x\mid w)-\delta(w)\int\tau(x)dF_{X|W}(x\mid w).
\]

The conclusion follows by noting, due to law of iterated expectations, 
\begin{align*}
\int\tau(x)dF_{X|W}(x\mid w) & =\mathbb{E}\left[\mathbb{E}\left[Y(1)-Y(0)|X\right]\mid W=w\right]\\
 & =\mathbb{E}\left[Y(1)-Y(0)\mid W=w\right]=\tau(w).
\end{align*}

\subsection*{Proof of Proposition \ref{prop:debias}}
Under Assumption \ref{asm:unconfounded}, let
\begin{align*}
\gamma_{1}(x) & =\mathbb{E}[Y\mid X=x,D=1]\in\mathcal{H}_{1},\quad
\gamma_{0}(x)  =\mathbb{E}[Y\mid X=x,D=0]\in\mathcal{H}_{0},
\end{align*}
where 
\begin{align*}
\mathcal{H}_{1} & :=\left\{ g(x,1):\mathbb{E}[g^{2}(X,1)]<\infty\mid g(x,d):\mathcal{X}\times\{0,1\}\rightarrow\mathbb{R}\right\} ,\\
\mathcal{H}_{0} & :=\left\{ g(x,0):\mathbb{E}[g^{2}(X,0)]<\infty\mid g(x,d):\mathcal{X}\times\{0,1\}\rightarrow\mathbb{R}\right\} .
\end{align*}
As  $\gamma_{1}$ and $\gamma_{0}$ are conditional
expectation functions, and $\tau(x)=\gamma_1(x)-\gamma_0(x)$, we follow \citet[][Proposition 4]{newey1994asymptotic}
to pin down the form of the efficient influence function. Slightly violating notations,  write
$m(z,\gamma_{1},\gamma_{0}):=m(z,\theta_{0},\gamma_1-\gamma_0)$, where the second $m(\cdot,\cdot,\cdot)$ is defined in \eqref{eq:moment.loss}. 
\paragraph{Step 1} Linearization. We calculate the pathwise derivative of $\mathbb{E}[m(z,\cdotp,\gamma_{0})]:\mathcal{H}_1\rightarrow\mathbb{R}$
at $\gamma_{1}$ along direction $(\tilde{\gamma}_{1}-\gamma_{1})\in\mathcal{H}_1$, as follows:
\begin{align*}
 & \frac{\partial\mathbb{E}[m(z,\gamma_{1}+t(\tilde{\gamma}_{1}-\gamma_{1}),\gamma_{0})])}{\partial t}\mid_{t=0}\\
= & \mathbb{E}\left[2\left(\gamma_{1}(X)-\gamma_{0}(X)\right)\left(\mathbf{1}\left\{ \gamma_{1}(X)-\gamma_{0}(X)\geq0\right\} -\delta(W)\right)^{2}\left(\tilde{\gamma}_{1}(X)-\gamma_{1}(X)\right)\right]\\
= & \mathbb{E}\left[D_{1}(Z,(\tilde{\gamma}_{1}-\gamma_{1}))\right],
\end{align*}
where 
\begin{align*}
 & D_{1}(z,\gamma_{1}+t(\tilde{\gamma}_{1}-\gamma_{1}))\\
:= & 2\left(\gamma_{1}(x)-\gamma_{0}(x)\right)\left(\mathbf{1}\left\{ \gamma_{1}(x)-\gamma_{0}(x)\geq0\right\} -\delta(w)\right)^{2}\left(\gamma_{1}(x)+t(\tilde{\gamma}_{1}(x)-\gamma_{1}(x))\right).
\end{align*}
Analogously, the pathwise derivative of $\mathbb{E}[m(z,\gamma_{1},\cdotp)]:\mathcal{H}_0\rightarrow\mathbb{R}$
at $\gamma_{0}$ along direction $(\tilde{\gamma}_{0}-\gamma_{0})\in\mathcal{H}_0$
is calculated as
\begin{align*}
 & \frac{\partial\mathbb{E}[m(z,\gamma_{1},\gamma_{0}+t(\tilde{\gamma}_{0}-\gamma_{0}))])}{\partial t}\mid_{t=0}\\
= & -\mathbb{E}\left[2\left(\gamma_{1}(X)-\gamma_{0}(X)\right)\left(\mathbf{1}\left\{ \gamma_{1}(X)-\gamma_{0}(X)\geq0\right\} -\delta(W)\right)^{2}\left(\tilde{\gamma}_{0}(X)-\gamma_{0}(X)\right)\right]\\
= & \mathbb{E}\left[D_{0}(Z,(\tilde{\gamma}_{0}-\gamma_{0}))\right],
\end{align*}
where 
\begin{align*}
 & D_{0}(z,\gamma_{0}+t(\tilde{\gamma}_{0}-\gamma_{0}))\\
:= & -2\left(\gamma_{1}(x)-\gamma_{0}(x)\right)\left(\mathbf{1}\left\{ \gamma_{1}(x)-\gamma_{0}(x)\geq0\right\} -\delta(w)\right)^{2}\left(\gamma_{0}(x)+t(\tilde{\gamma}_{0}(x)-\gamma_{0}(x))\right).
\end{align*}

\paragraph{Step 2} Derive the integral forms of $\mathbb{E}\left[D_{1}(Z,\cdotp)\right]$
and $\mathbb{E}\left[D_{0}(Z,\cdotp)\right]$. In our case, for all $\tilde{\gamma}_{1}\in\mathcal{H}_1$, we have
\begin{align*}
\mathbb{E}\left[D_{1}(Z,\tilde{\gamma}_{1}(X))\right] & =\mathbb{E}\left[2\left(\gamma_{1}(X)-\gamma_{0}(X)\right)\left(\mathbf{1}\left\{ \gamma_{1}(X)-\gamma_{0}(X)\geq0\right\} -\delta(W)\right)^{2}\tilde{\gamma}_{1}(X)\right]\\
 & =\mathbb{E}\left[2\left(\gamma_{1}(X)-\gamma_{0}(X)\right)\left(\mathbf{1}\left\{ \gamma_{1}(X)-\gamma_{0}(X)\geq0\right\} -\delta(W)\right)^{2}\frac{D}{\pi(X)}\tilde{\gamma}_{1}(X)\right]\\
 & =\mathbb{E}\left[\varsigma_{1}(D,X)\tilde{\gamma}_{1}(X)\right],
\end{align*}
where $\varsigma_{1}(d,x):=2\left(\gamma_{1}(x)-\gamma_{0}(x)\right)\left(\mathbf{1}\left\{ \gamma_{1}(x)-\gamma_{0}(x)\geq0\right\} -\delta(w)\right)^{2}\frac{d}{\pi(x)}$.
Moreover, let $\mathbb{E}_{t}[\cdotp]$ be the expectation under $P_{t}$,
a one-dimensional subfamily of $P\in\mathcal{P}$ that equals the
true distribution when $t=0$. Let $\gamma_{1}(x,t):=\arg\min_{\tilde{\gamma}_{1}\in\mathcal{H}_{1}}\mathbb{E}_{t}\left[\left(Y-\tilde{\gamma}_{1}(x)\right)^{2}\right]$.
Then, the following holds:
\[
\mathbb{E}_{t}\left[\varsigma_{1}(d,x)\gamma_{1}(X,t)\right]=\mathbb{E}_{t}\left[\varsigma_{1}(d,x)Y\right].
\]
Analogously, For all $\tilde{\gamma}_{0}\in\mathcal{H}_{0}$ such
that $\mathbb{E}[\tilde{\gamma}_{0}^{2}(X)]<\infty$, we have
\begin{align*}
\mathbb{E}\left[D_{0}(Z,\tilde{\gamma}_{0}(X))\right] & =-\mathbb{E}\left[2\left(\gamma_{1}(X)-\gamma_{0}(X)\right)\left(\mathbf{1}\left\{ \gamma_{1}(X)-\gamma_{0}(X)\geq0\right\} -\delta(W)\right)^{2}\tilde{\gamma}_{0}(X)\right]\\
 & =-\mathbb{E}\left[2\left(\gamma_{1}(X)-\gamma_{0}(X)\right)\left(\mathbf{1}\left\{ \gamma_{1}(X)-\gamma_{0}(X)\geq0\right\} -\delta(W)\right)^{2}\frac{1-D}{1-\pi(X)}\tilde{\gamma}_{0}(X)\right]\\
 & =\mathbb{E}\left[\varsigma_{0}(D,X)\tilde{\gamma}_{0}(X)\right],
\end{align*}
where $\varsigma_{0}(d,x):=-2\left(\gamma_{1}(x)-\gamma_{0}(x)\right)\left(\mathbf{1}\left\{ \gamma_{1}(x)-\gamma_{0}(x)\geq0\right\} -\delta(w)\right)^{2}\frac{1-d}{1-\pi(x)}$
is such that 
\[
\mathbb{E}_{t}\left[\varsigma_{0}(d,x)\gamma_{0}(X,t)\right]=\mathbb{E}_{t}\left[\varsigma_{0}(d,x)Y\right],
\]
and $\gamma_{0}(X,t):=\arg\min_{\tilde{\gamma}_{0}\in\mathcal{H}_{0}}\mathbb{E}_{t}\left[\left(Y-\tilde{\gamma}_{0}(x)\right)^{2}\right]$. The conclusion then follows by invoking \citet[][Proposition 4]{newey1994asymptotic}. 


\subsection*{Proof of Proposition \ref{prop:parametric}}

Statement (i) and the first part of statement (ii) directly follows
from applying Theorem \ref{thm:main}. We show the asymptotic
distribution result under Assumptions \ref{asm:unconfounded}-\ref{asm:margin}. With a slight abuse of notation, let $\delta^{*}(w)=\left(\beta^{*}\right)^{\prime}p(w)$,
where $\beta^{*}=A^{-1}B$. Note this $\delta^*$ is the unique rule in $\mathcal{D}$ that solves \eqref{eq:loss.proposal}. Moreover, note
\begin{align*}
A & =\mathbb{E}\left[\tau^{2}(Z)p(W)p(W)^{\prime}\right]=\mathbb{E}\left[\xi(Z)p(W)p(W)^{\prime}\right],\\
B & =\mathbb{E}\left[\tau^{2}p(W)\mathbf{1}\left\{ \tau(X)\geq0\right\} \right]=\mathbb{E}\left[\xi(Z)p(W)\mathbf{1}\left\{ \tau(X)\geq0\right\} \right].
\end{align*}
\paragraph*{Step 1:} Conclude from Lemmas \ref{lem:remainder.1}, \ref{lem:remainder.2},
\ref{lem:remainder.5} and \ref{lem:remainder.6} that 
$\sqrt{n}\left\Vert \hat{\beta}-\tilde{\beta}\right\Vert =o_{p}(1)$,$\sqrt{n}\left(\tilde{\beta}-\beta^{*}\right)\overset{d}{\rightarrow}N(0,A^{-1}VA^{-1})$,
implying $\sqrt{n}\left(\hat{\beta}-\beta^{*}\right)\overset{d}{\rightarrow}N(0,A^{-1}VA^{-1})$ and $\hat{\beta}-\beta^{*}=O_{p}\left(\frac{1}{\sqrt{n}}\right)$.

\paragraph*{Step 2:}  We show that $\hat{\mathbf{p}}:=\int_{w:\hat{\delta}(w)\notin[0,1]}dF_{W}(w)=o_{p}(1)$.
Note $\left|\hat{\delta}(w)-\delta^{*}(w)\right|\leq\left\Vert \hat{\beta}-\beta^{*}\right\Vert \zeta_{p}$
for all $w\in\mathcal{W}$. With the conclusions of step 1, we have
\[
\sup_{w\in\mathcal{W}}\left|\hat{\delta}(w)-\delta^{*}(w)\right|=O_{p}\left(\frac{1}{\sqrt{n}}\right),
\]
implying that for some $C_{\epsilon}>0$ and all $n$ sufficiently
large,
\[
Pr\left\{ \sup_{w\in\mathcal{W}}\left|\hat{\delta}(w)-\delta^{*}(w)\right|>\frac{C_{\epsilon}}{\sqrt{n}}\right\} <\epsilon
\]
for any $\epsilon>0$. Let $\mathcal{E}_{\hat{\delta}}$ be the event
that $\sup_{w\in\mathcal{W}}\left|\hat{\delta}(w)-\delta^{*}(w)\right|\leq\frac{C_{\epsilon}}{\sqrt{n}}$.
Then, for any $\nu>0$
\begin{align*}
Pr\left\{ \hat{\mathbf{p}}>\nu\right\} \leq & Pr\left\{ \hat{\mathbf{p}}>\nu\mid\mathcal{E}_{1,n}\right\} Pr\left\{ \mathcal{E}_{1,n}\right\} +\left(1-Pr\left\{ \mathcal{E}_{1,n}\right\} \right)\\
\leq & Pr\left\{ \hat{\mathbf{p}}>\nu\mid\mathcal{E}_{1,n}\right\} +\epsilon.
\end{align*}
Of note,
\begin{align*}
Pr\left\{ \hat{\mathbf{p}}>\nu\mid\mathcal{E}_{1,n}\right\} = & Pr\left\{ \int_{w:\hat{\delta}(w)\in\left[-\frac{C_{\epsilon}}{\sqrt{n}},0\right]\cup\left[1,\frac{C_{\epsilon}}{\sqrt{n}}+1\right]}dF_{W}(w)>\nu\mid\mathcal{E}_{1,n}\right\}\leq \epsilon
\end{align*}
for $n$ sufficiently large. Therefore, for all $n$ sufficiently large,
$Pr\left\{ \hat{\mathbf{p}}>\nu\right\} \leq2\epsilon.$ Conclude
$\hat{\mathbf{p}}=o_{p}(1)$ by the arbitrariness of $\nu$ and $\epsilon$. 
\paragraph{Step 3:} We now analyze the asymptotic behavior of $n\left(L(\hat{\delta}^{\mathcal{T}},\tau)-L(\delta^{*},\tau)\right)$.
Note $\delta^{*}$ satisfies the following condition, 
\[
\mathbb{E}\left[\tau^{2}(X)\left(\mathbf{1}\left\{ \tau(X)\geq0\right\} -\delta^{*}(W)\right)\mid W\right]=0,
\]
implying $n\left(L(\hat{\delta}^{\mathcal{T}},\tau)-L(\delta^{*},\tau)\right) =n\mathbb{E}\left[\tau^{2}(X)\left(\hat{\delta}^{\mathcal{T}}(W)-\delta^{*}(W)\right)^{2}\right]$
by Taylor's theorem. Conclude that 
\begin{align*}
n\left(L(\hat{\delta}^{\mathcal{T}},\tau)-L(\delta^{*},\tau)\right) & =T_{4,n}+T_{5,n},
\end{align*}
where 
\begin{align*}
T_{4,n} & =n\mathbb{E}\left[\tau^{2}(X)\left(\hat{\delta}^{\mathcal{T}}(W)-\delta^{*}(W)\right)^{2}\mathbf{1}\left\{ \hat{\delta}(W)\in[0,1]\right\} \right],\\
T_{5,n} & =n\mathbb{E}\left[\tau^{2}(X)\left(\hat{\delta}^{\mathcal{T}}(W)-\delta^{*}(W)\right)^{2}\mathbf{1}\left\{ \hat{\delta}(W)\notin[0,1]\right\} \right].
\end{align*}
For $T_{5,n}$, note 
\begin{align*}
\left|T_{5,n}\right| & \leq n\mathbb{E}\left[\tau^{2}(X)\left(\hat{\delta}(W)-\delta^{*}(W)\right)^{2}\mathbf{1}\left\{ \hat{\delta}(W)\notin[0,1]\right\} \right]\\
 & \leq n\left(\hat{\beta}-\beta^{*}\right)^{\prime}\hat{\mathcal{A}}\left(\hat{\beta}-\beta^{*}\right),
\end{align*}
where
\[
\hat{\mathcal{A}}:=\mathbb{E}\left[\tau^{2}(X)p(W)p^{\prime}(W)\mathbf{1}\left\{ \hat{\delta}(W)\notin[0,1]\right\} \right].
\]
Furthermore, note $\left\Vert \hat{\mathcal{A}}\right\Vert \leq C_{\xi}\zeta_{p}^{2}\hat{\mathbf{p}}=o_{p}(1)$
by conclusions of Step 2, and $\hat{\beta}-\beta^{*}=O_{p}\left(\frac{1}{\sqrt{n}}\right)$
by conclusions from Step 1. Conclude that $T_{5,n}=o_{p}(1)$. For term $T_{4,n}$, note 
\begin{align*}
T_{4,n} & =\mathbb{E}\left[\tau^{2}(X)\left(\hat{\delta}(W)-\delta^{*}(W)\right)^{2}\mathbf{1}\left\{ \hat{\delta}(W)\in[0,1]\right\} \right]=T_{4,1,n}-T_{4,2,n},
\end{align*}
where
\begin{align*}
T_{4,1,n} & :=\left(\sqrt{n}\left(\hat{\beta}-\beta^{*}\right)\right)^{\prime}A\left(\sqrt{n}\left(\hat{\beta}-\beta^{*}\right)\right)\overset{d}{\rightarrow}N(0,\Omega)AN(0,\Omega),
\end{align*}
and 
\begin{align*}
T_{4,2,n} & :=\left(\sqrt{n}\left(\hat{\beta}-\beta^{*}\right)\right)^{\prime}\hat{\mathcal{A}}\left(\sqrt{n}\left(\hat{\beta}-\beta^{*}\right)\right) =o_{p}(1)
\end{align*}
by conclusions from step 2. As $T_{5,n}=o_{p}(1),$ $T_{4,2,n}=o_{p}(1)$,
we conclude that the asymptotic distribution of $n\left(L(\hat{\delta}^{\mathcal{T}},\tau)-L(\delta^{*},\tau)\right)$
is determined by $T_{4,1,n}$, completing the proof.

\newpage
\global\long\def\thepage{OA-\arabic{page}}%
 \setcounter{page}{1}

\part*{Online Supplement}

\section{Additional results for Appendix \ref{sec:technical.discussion}}\label{sec:OS.1}

\subsection{Proof of Proposition \ref{prop:capacity}}
If the capacity constraint is not binding, the unconstrained optimal
is also the constrained optimal. Thus, focus on the case when the
capacity constraint is binding, i.e., $\sum\limits _{j=1}^{J}\frac{p_{j}b_{j}}{a_{j}}-t>0$.
Consider the following primal (\textbf{P}):
\begin{align*}
\min_{\left\{ \delta_{j}\right\} _{j=1}^{J}} & \sum_{j=1}^{J}p_{i}\left[b_{j}-2b_{j}\delta_{j}+a_{j}\delta_{j}^{2}\right]\nonumber \\
s.t. & \sum_{j=1}^{J}p_{j}\delta_{j}\leq t,\\
 & \delta_{j}\geq0,j=1\ldots J.
\end{align*}
We characterize the solution of \textbf{P}, denoted as $\left\{ \delta_{j}^{*}\right\} _{j=1}^{J}$,
which, as we will show below, does not violate constraints $\delta_{j}\leq1$
for all $j=1\ldots J$. Therefore, \textbf{$\left\{ \delta_{j}^{*}\right\} _{j=1}^{J}$
}is also optimal even with the additional constraints. Next, since
\textbf{P} is convex with differentible objective and constraint functions,
and the Slater's condition is satisfied, it follows that (\citealt{boyd2004convex}, Chapter
5, p.244) strong duality holds, and \textbf{$\left\{ \delta_{j}^{*}\right\} _{j=1}^{J}$}
is a solution of \textbf{P} if and only if the following KKT conditions
are met:
\begin{align*}
-2p_{j}b_{j}+2p_{j}a_{j}\delta_{j}^{*}+\lambda^{*}p_{j}-h_{j}^{*} & =0,j=1\ldots J,\\
\delta_{j}^{*}\geq0,h_{j}^{*}\geq0,h^{*}_{j}\delta_{j}^{*} & =0,j=1\ldots J,\\
\lambda^{*}\left(\sum_{j=1}^{J}p_{j}\delta_{j}^{*}-t\right)=0,\sum_{j=1}^{J}p_{j}\delta_{j}^{*}\leq t,\lambda^{*} & \geq0.
\end{align*}
The first equation implies
\[
\delta_{j}^{*}=\frac{b_{j}}{a_{j}}-\frac{\lambda^{*}}{2a_{j}}+\frac{h_{j}^{*}}{2p_{j}a_{j}},j=1\ldots J.
\]
If $\sum_{j=1}^{J}p_{j}\delta_{j}^{*}<t$, then $\lambda^{*}=0$,
and $\delta_{j}^{*}=\frac{b_{j}}{a_{j}}+\frac{h_{j}^{*}}{2p_{j}a_{j}}\geq\frac{b_{j}}{a_{j}}$.
It follows then $\sum_{j=1}^{J}p_{j}\delta_{j}^{*}\geq\sum_{j=1}^{J}\frac{p_{j}b_{j}}{a_{j}}>t$,
a contradiction. Conclude that the capacity constraint must be binding,
i.e., $\sum_{j=1}^{J}p_{j}\delta_{j}^{*}=t$. Furthermore, for any
$j$ such that $h_{j}^{*}=0$, we must have 
\[
\delta_{j}^{*}=\frac{b_{j}}{a_{j}}-\frac{\lambda^{*}}{2a_{j}}\geq0.
\]
Since $b_{j}$ is decreasing in $j$, we conclude that 
\[
\delta_{j}^{*}\geq0\Rightarrow\delta_{i}^{*}\geq0,\forall i\leq j,i\in\left\{ 1,\ldots J\right\} .
\]
Therefore, there must exist some $J^{*}\in\left\{ 1,\ldots J\right\} $
such that 
\begin{align*}
\delta_{j}^{*} & \geq0,\forall j\leq J^{*};\delta_{j}^{*}=0,\forall j>J^{*},
\end{align*}
implying 
\begin{align*}
h_{j}^{*} & =0,\forall j\leq J^{*};h_{j}^{*}>0,\forall j>J^{*},
\end{align*}
and 
\[
\delta_{j}^{*}=\frac{b_{j}}{a_{j}}-\frac{\lambda^{*}}{2a_{j}},\forall j\leq J^{*}.
\]
With the binding capacity constraint, we can back out the value of $\lambda^*$, which must also satisfy the nonnegativity constraint.

\subsection{Proof of Lemma \ref{lem:main.1}}

Recall
\[
T_{1n}=L(\hat{\delta},\tau)-L(\delta^{*},\tau)-2\left[L_{n}^{o}(\hat{\delta},\eta_{0})-L_{n}^{o}(\delta^{*},\eta_{0})\right].
\]
And note for all $\delta\in\mathcal{D}$ as well as $\delta^{*}$,
\[
L^{o}(\delta,\eta_{0})=L(\delta,\tau),\quad L^{o}(\delta^{*},\eta_{0})=L(\delta^{*},\tau).
\]
To ease notational burden, in what follows, we write $L^{o}(\delta):=L^{o}(\delta,\eta_{0})$,
$L(\delta):=L(\delta,\tau)$ for any $\delta$. On event $\mathcal{E}_n$, the minimum eigenvalue of $\hat{A}_n$ is bounded away from $\underline{\lambda}-\varepsilon$. Therefore, Assumption \ref{asm:stability} implies that, for some $C_L$, $\hat{\delta}$ is an element of the space 
\[
\mathcal{D}_{C_L}:=\mathcal{D}_{n,C_L}:=\left\{ f(w)=\sum_{j=1}^{\sieved}\beta_{j}p_{j}(w):\sup_{w}|f(w)|\leq C_L\right\} .
\] 
Then, for all $t>0$ and on event $\mathcal{E}_{n}$, we have\footnote{The corresponding probability refers to the conditional probability on event $\mathcal{E}_{n}$.} 
\begin{align*}
Pr\{T_{1n}>t\} & \leq Pr\left\{ \exists\delta\in\mathcal{D}_{C_L}:L(\delta)-L(\delta^{*})-2\left[L_{n}^{o}(\delta)-L_{n}^{o}(\delta^{*})\right]>t\right\} \\
 & =Pr\left\{ \exists\delta\in\mathcal{D}_{C_L}:L^{o}(\delta)-L^{o}(\delta^{*})-2\left[L_{n}^{o}(\delta)-L_{n}^{o}(\delta^{*})\right]>t\right\} \\
 & =Pr\{\exists\delta\in\mathcal{D}_{C_L}:2\left[L^{o}(\delta)-L^{o}(\delta^{*})\right]-2\left[L_{n}^{o}(\delta)-L_{n}^{o}(\delta^{*})\right]\\
 & >t+L^{o}(\delta)-L^{o}(\delta^{*})\}\\
 & =Pr\left\{ \exists\delta\in\mathcal{D}_{C_L}:\frac{L^{o}(\delta)-L^{o}(\delta^{*})-\left[L_{n}^{o}(\delta)-L_{n}^{o}(\delta^{*})\right]}{t+L^{o}(\delta)-L^{o}(\delta^{*})}>\frac{1}{2}\right\} \\
 & \leq Pr\left\{ \sup_{\delta\in\mathcal{D}_{C_L}}\frac{\left|L^{o}(\delta)-L^{o}(\delta^{*})-\left[L_{n}^{o}(\delta)-L_{n}^{o}(\delta^{*})\right]\right|}{t+L^{o}(\delta)-L^{o}(\delta^{*})}>\frac{1}{2}\right\} .
\end{align*}
We now apply Lemma \ref{lem:kohler} to bound the above term and to
derive an upper bound for $\mathbb{E}_{P^{n}}T_{1n}$. Let $\left(Z_{1},Z_{2},\ldots Z_{n}\right)$
be iid copies of $Z=(Y,D,X)$. For each $f\in\mathcal{D}_{C_L}$, write
\[
g_{f}(z):=\xi(z)\left(\mathbf{1}\{\tau(x)\geq0\}-f(w))\right)^{2}-\xi(z)\left(\mathbf{1}\{\tau(x)\geq0\}-\delta^{*}(w)\right)^{2},
\]
where recall 
\begin{align*}
\xi(z) & =\left[\gamma_{1}(x)-\gamma_{0}(x)\right]^{2}+d\omega_{1}(x)(y-\gamma_{1}(x))-\left(1-d\right)\omega_{0}(x)(y-\gamma_{0}(x)),\\
f(w) & =\beta^{\prime}p(w),\sup_{w\in\mathcal{W}}\left|f(w)\right|\leq C_L.
\end{align*}
Consider the following functional class $\mathcal{G}:=\left\{ g=g_{f}(z)\mid f\in\mathcal{D}_{C_L}\right\}$. Under Assumptions \ref{asm:unconfounded} and \ref{asm:reg}, there
exists some finite $K_{1}\geq1$ such that 
\begin{align*}
\left|g_{f}(z)\right| & \leq\left|\xi(z)\right|\left(1+L\right)^{2}+\left|\xi(z)\right|\leq K_{1},
\end{align*}
for all $g_{f}\in\mathcal{G}$ and $z\in\mathcal{Z}$. Furthermore, we can equivalently write each $g_{f}\in\mathcal{G}$ as 
\begin{align*}
g_{f}(z) & =\xi(z)\left[\left(\delta^{*}(w)-f(w)\right)^{2}+2\left(\mathbf{1}\{\tau(x)\geq0\}-\delta^{*}(w)\right)\left(\delta^{*}(w)-f(w)\right)\right].
\end{align*}
Moreover, for each $g_{f}\in\mathcal{G}$, 
\begin{align*}
 & \mathbb{E}\xi(Z)\left[\left(\mathbf{1}\{\tau(X)\geq0\}-\delta^{*}(W)\right)\left(\delta^{*}(W)-f(W)\right)\right]\\
= & \mathbb{E}\left[\tau(X)^{2}\left(\mathbf{1}\{\tau(X)\geq0\}-\delta^{*}(W)\right)\left(\delta^{*}(W)-f(W)\right)\right]=  0,
\end{align*}
where the last equality follows from the fact that $\delta^{*}$ must
satisfy the following condition (due to Proposition \ref{prop:population}(i))
\begin{align*}
\mathbb{E}\left[\tau(X)^{2}\left(\mathbf{1}\{\tau(X)\geq0\}-\delta^{*}(W)\right)\mid W\right]=0.
\end{align*}
As a result, we conclude 
\begin{align*}
\mathbb{E}[g_{f}(Z)] & =\mathbb{E}\left[\xi(Z)\left(\delta^{*}(W)-f(W)\right)^{2}\right]=\mathbb{E}\left[\tau^{2}(X)\left(\delta^{*}(W)-f(W)\right)^{2}\right],\forall g_{f}\in\mathcal{G}.
\end{align*}
Moreover, note 
\[
g_{f}^{2}(z)=\xi^{2}(z)\left(\delta^{*}(w)-f(w)\right)^{2}\left[\left(\delta^{*}(w)-f(w)\right)+2\left(\mathbf{1}\{\tau(x)\geq0\}-\delta^{*}(w)\right)\right]^{2}.
\]
It follows then 
\begin{align*}
\mathbb{E}[g_{f}^{2}(Z)] & =\mathbb{E}\left\{ \xi^{2}(Z)\left[2\mathbf{1}\{\tau(X)\geq0\}-f(W)-\delta^{*}(W)\right]^{2}\left[\delta^{*}(W)-f(W)\right]^{2}\right\} \\
 & \leq\left(3+L\right)^{2}\mathbb{E}\left\{ \xi^{2}(Z)\left[\delta^{*}(W)-f(W)\right]^{2}\right\} ,
\end{align*}
where note 
\begin{align*}
 & \mathbb{E}\left\{ \xi^{2}(Z)\left[\delta^{*}(W)-f(W)\right]^{2}\right\} \\
= & \mathbb{E}\left\{ \tau^{4}(X)\left[\delta^{*}(W)-f(W)\right]^{2}\right\} \\
+ & \mathbb{E}\left\{ \left(\frac{2De_{1}}{\pi(X)}-\frac{\left(1-D\right)2e_{0}}{1-\pi(X)}\right)^{2}\tau^{2}(X)\left[\delta^{*}(W)-f(W)\right]^{2}\right\} .
\end{align*}
Furthermore, under Assumptions \ref{asm:unconfounded} and \ref{asm:reg},
observe that 
\[
\mathbb{E}\left\{ \tau^{4}(X)\left[\delta^{*}(W)-f(W)\right]^{2}\right\} \leq\sup_{x\in\mathcal{X}}\tau^{2}(x)\mathbb{E}[g_{f}(Z)],
\]
and 
\begin{align*}
 & \mathbb{E}\left\{ \left(\frac{2De_{1}}{\pi(X)}-\frac{\left(1-D\right)2e_{0}}{1-\pi(X)}\right)^{2}\tau^{2}(X)\left[\delta^{*}(W)-f(W)\right]^{2}\right\} \\
\leq & \sup_{x\in\mathcal{X}}\left(\frac{4\mathbb{E}[e_{1}^{2}\mid X=x]}{\pi^{2}(x)}+\frac{4\mathbb{E}[e_{1}^{2}\mid X=x]}{\left(1-\pi(x)\right)}\right)\mathbb{E}[g_{f}(Z)].
\end{align*}
Thus, we conclude that there exists a finite number $K_{2}\geq1$
such that 
$\mathbb{E}[g_{f}^{2}(Z)]\leq K_{2}\mathbb{E}g_{f}(Z)
$.
Denote by $N(\epsilon,\mathcal{G},d_{2,n})$ the covering number for
$\mathcal{G}$ with respect to the empirical $L_{2}$ distance 
\[
d_{2,n}(g_{f_{1}},g_{f_{2}}):=\left\{ \frac{1}{n}\sum_{i=1}^{n}\left[g_{f_{1}}(Z_{i})-g_{f_{2}}(Z_{i})\right]^{2}\right\} ^{1/2}.
\]
Note for each $g_{f}\in\mathcal{G}$, 
\begin{align*}
g_{f}(z) & =\xi(z)\left(\mathbf{1}\{\tau(x)\geq0\}-f(w)\right)^{2}-\xi(z)\left(\mathbf{1}\{\tau(x)\geq0\}-\delta^{*}(w)\right)^{2}\\
 & =\xi(z)\left(\mathbf{1}\{\tau(x)\geq0\}-2f(w)+f^{2}(w)\right)-\xi(z)\left(\mathbf{1}\{\tau(x)\geq0\}-\delta^{*}(w)\right)^{2},
\end{align*}
implying $\mathcal{G}$ is a subset of a linear vector space with dimension $d_\mathcal{G}$, where $d_\mathcal{G}\leq1+\sieved+\sievedstar\leq1+2\sievedstar$. 
It follows from \citet[Corollary 2.6]{geer2000empirical}
that 
\[
N\left(\epsilon,\left\{ g_{f}\in\mathcal{G}:\frac{1}{n}\sum_{i=1}^{n}\left[g_{f}(Z_{i})\right]^{2}\leq4\delta\right\} ,d_{2,n}\right)\leq\left(\frac{8\sqrt{\delta}+\epsilon}{\epsilon}\right)^{d_\mathcal{G}+1},
\]
Hence, integration by change-of-variable yields 
\begin{align*}
 & \int_{0}^{\sqrt{\delta}}\left\{ \log N\left(\epsilon,\left\{ g_{f}\in\mathcal{G}:\frac{1}{n}\sum_{i=1}^{n}\left[g_{f}(Z_{i})\right]^{2}\leq4\delta\right\} ,d_{2,n}\right)\right\} ^{1/2}d\epsilon\\
\leq & \left(d_\mathcal{G}+1\right)^{1/2}\int_{0}^{\sqrt{\delta}}\left\{ \log\left(\frac{8\sqrt{\delta}+\epsilon}{\epsilon}\right)\right\} ^{1/2}d\epsilon\\
= & \left(d_\mathcal{G}+1\right)^{1/2}\sqrt{\delta}\int_{0}^{1}\left\{ \log\left(1+\frac{8}{t}\right)\right\} ^{1/2}dt\\
= & \left(d_\mathcal{G}+1\right)^{1/2}\sqrt{\delta}\int_{1}^{\infty}\frac{\sqrt{\log\left(1+8u\right)}}{u^{2}}du\\
\leq & \left(d_\mathcal{G}+1\right)^{1/2}\sqrt{\delta}\int_{1}^{\infty}\frac{\sqrt{8u}}{u^{2}}du(\log\left(1+x\right)\leq x\text{ for }x>0)\\
= & 2\sqrt{2}\left(d_\mathcal{G}+1\right)^{1/2}\sqrt{\delta}\int_{1}^{\infty}u^{-\frac{3}{2}}du\\
= & 4\sqrt{2}\left(d_\mathcal{G}+1\right)^{1/2}\sqrt{\delta}.
\end{align*}
Thus, for $\varepsilon=\frac{1}{2}$ in the conditions of Lemma \ref{lem:kohler},
there exists some finite constant $c_{K_{1},K_{2}}>0$ such that and
for all 
\[
\alpha\geq c_{K_{1},K_{2}}\frac{d_\mathcal{G}+1}{n},\delta\geq4\alpha,
\]
all the conditions of Lemma \ref{lem:kohler} are met so that we conclude
for all $\alpha=c_{K_{1},K_{2}}\frac{d_\mathcal{G}+1}{n}$, we have that
there exists some finite constant $\overline{c}_{K_{1},K_{2}}>0$
such that for all $t\geq c_{K_{1},K_{2}}\frac{d_\mathcal{G}+1}{n}$, we
have 
\begin{align*}
Pr\{T_{1n}>t\}\leq & 50\exp\left(-\frac{nt}{\overline{c}_{K_{1},K_{2}}}\right).
\end{align*}
Then, for each $P^{n}\in\mathcal{P}^{n}$, 
\begin{align*}
\mathbb{E}_{P^{n}}[T_{1n}] & \leq\int_{0}^{\infty}Pr\{T_{1n}>t\}dt\\
 & \leq c_{K_{1},K_{2}}\frac{d_\mathcal{G}+1}{n}+\int_{C_{K_{1},K_{2}}\frac{d_\mathcal{G}+1}{n}}^{\infty}50\exp\left(-\frac{nt}{\overline{c}_{K_{1},K_{2}}}\right)dt\\
 & =c_{K_{1},K_{2}}\frac{d_\mathcal{G}+1}{n}+\frac{50\overline{c}_{K_{1},K_{2}}}{n}\exp\left(-\frac{c_{K_{1},K_{2}}\left(d_\mathcal{G}+1\right)}{\overline{c}_{K_{1},K_{2}}}\right),
\end{align*}
implying there exists some constant $\mathcal{C}_{1}$ that only depends
on $K_{1}$, $K_{2}$ 
such that 
\[
\mathbb{E}_{P^{n}}[T_{1n}]\leq\mathcal{C}_{1}\frac{\sievedstar}{n}.
\]
As $\sievedstar$ and $\mathcal{C}_{1}$ does not depend on $P^{n}$,
the conclusion of the lemma follows.

\subsection{Proof of Lemma \ref{lem:main.2}}
On event $\mathcal{E}_{n}$, we have
\[
\hat{\beta}=\hat{A}_{n}^{-1}\hat{B}_{n},\quad\tilde{\beta}=A_{n}^{-1}B_{n},
\]
and  $\tilde{\delta}$ solves the oracle
problem  $\min_{\delta=\beta^{\prime}p,\beta\in\mathbb{R}^{\sieved}}L_{n}^{o}(\delta)$, satisfying the following FOC:
\[
\frac{1}{n}\sum_{i=1}^{n}\left[\xi_{i}\left(\mathbf{1}\left\{ \tau_{i}\geq0\right\} -\tilde{\delta}_{i}\right)p_i\right]=0.
\]
Then, Taylor's theorem implies
\begin{align*}
 0\leq& L_{n}^{o}(\hat{\delta},\eta_{0})-L_{n}^{o}(\tilde{\delta},\eta_{0})\\
= & \frac{1}{n}\sum_{i=1}^{n}\xi_{i}\left(\mathbf{1}\left\{ \tau_{i}\geq0\right\} -\hat{\delta}_{i}\right)^{2}-\frac{1}{n}\sum_{i=1}^{n}\xi_{i}\left(\mathbf{1}\left\{ \tau_{i}\geq0\right\} -\tilde{\delta}_{i}\right)^{2}\\
= & (\hat{\beta}-\tilde{\beta})^{\prime}A_{n}(\hat{\beta}-\tilde{\beta}),
\end{align*}
where 
\[
\hat{\beta}-\tilde{\beta}=A_{n}^{-1}\left(A_{n}-\hat{A}_{n}\right)\hat{A}_{n}^{-1}\hat{B}_{n}+A_{n}^{-1}\left(\hat{B}_{n}-B_{n}\right).
\]
Furthermore, on event $\mathcal{E}_{n}$, algebra shows 
\begin{align*}
 L_{n}^{o}(\hat{\delta},\eta_{0})-L_{n}^{o}(\tilde{\delta},\eta_{0})= & \hat{B}_{n}^{\prime}\hat{A}_{n}^{-1}\left(A_{n}-\hat{A}_{n}\right)A_{n}^{-1}\left(A_{n}-\hat{A}_{n}\right)\hat{A}_{n}^{-1}\hat{B}_{n}\\
+&2\hat{B}_{n}^{\prime}\hat{A}_{n}^{-1}\left(A_{n}-\hat{A}_{n}\right)A_{n}^{-1}\left(\hat{B}_{n}-B_{n}\right)\\
+&\left(\hat{B}_{n}-B_{n}\right)^{\prime}A_{n}^{-1}\left(\hat{B}_{n}-B_{n}\right).
\end{align*}
Also note $\left\Vert \hat{B}_{n}\right\Vert \leq\tilde{C}_{\xi}\zeta_{p}$.
Applying triangle and Cauchy-Schwarz inequalities yields:
\begin{align*}
 & \mathbb{E}_{P^{n}}\left[L_{n}^{o}(\hat{\delta},\eta_{0})-L_{n}^{o}(\tilde{\delta},\eta_{0})\right]\\
\leq & \left(\frac{\tilde{C}_{\xi}^{2}\zeta_{p}^2}{\left(\underline{\lambda}-\varepsilon\right)^{3}}+\frac{\tilde{C}_{\xi}\zeta_{p}}{\left(\underline{\lambda}-\varepsilon\right)^{2}}\right)\mathbb{E}_{P^{n}}\left\Vert \hat{A}_{n}-A_{n}\right\Vert ^{2}\\
+ & \left(\frac{1}{\underline{\lambda}-\varepsilon}+\frac{\tilde{C}_{\xi}\zeta_{p}}{\left(\underline{\lambda}-\varepsilon\right)^{2}}\right)\mathbb{E}_{P^{n}}\left\Vert \hat{B}_{n}-B_{n}\right\Vert ^{2}\\
\leq & c_1\left(\zeta_{p}^2\mathbb{E}_{P^{n}}\left\Vert \hat{A}_{n}-A_{n}\right\Vert ^{2}+\zeta_{p}\mathbb{E}_{P^{n}}\left\Vert \hat{B}_{n}-B_{n}\right\Vert ^{2}\right),
\end{align*}
where $c_1$ is a constant that depends on $\tilde{C}_{\xi}$, $\varepsilon$
and $\underline{\lambda}$. Then, the conclusion of statement (i)
follows from invoking Lemmas \ref{lem:remainder.5} and \ref{lem:remainder.6}(i), and that of statement (ii) follows from Lemmas \ref{lem:remainder.5} and \ref{lem:remainder.6}(ii).

\subsection{Additional Lemmas supporting Theorem \ref{thm:main}}

\begin{lem}\label{lem:remainder.1}
Under Assumptions \ref{asm:unconfounded}-\ref{asm:quality}, we have 
\[
P\left\{ \lambda_{\min}\left(A_{n}\right)<\underline{\lambda}-\varepsilon\right\} \leq2\sieved\exp\left(\frac{-n\varepsilon^{2}}{4C_{\xi}^{2}\zeta_{p}^{4}}\right),
\]
for all $0<\varepsilon\leq\min\left\{ 3C_{\xi}\zeta_{p}^{2}/2,\underline{\lambda}\right\} $. 
\end{lem}

\begin{proof}
As $\left|\lambda_{\min}\left(A_{n}\right)-\underline{\lambda}\right|\leq\left\Vert A_{n}-A\right\Vert $,
we have 
\begin{align*}
P\left\{ \lambda_{\min}\left(A_{n}\right)<\underline{\lambda}-\varepsilon\right\}  & \leq P\left\{ \left|\lambda_{\min}\left(A_{n}\right)-\underline{\lambda}\right|>\varepsilon\right\} \leq P\left\{ \left\Vert A_{n}-A\right\Vert >\varepsilon\right\} .
\end{align*}
Note $\left\Vert \xi_{i}p(W_{i})p^{\prime}(W_{i})\right\Vert\leq C_{\xi}\zeta_{p}^{2}$, $\left\Vert \mathbb{E}\xi_{i}^{2}p(W_{i})p^{\prime}(W_{i})p(W_{i})p^{\prime}(W_{i})\right\Vert  \leq C_{\xi}^{2}\zeta_{p}^{4}$.
\citet[][Corollary 6.2.1]{tropp2015introduction} implies that 
\[
P\left\{ \left\Vert A_{n}-A\right\Vert >\varepsilon\right\} \leq2\sieved\exp\left(\frac{-n\varepsilon^{2}/2}{\left(C_{\xi}^{2}\zeta_{p}^{4}+C_{\xi}\zeta_{p}^{2}\varepsilon/3\right)}\right).
\]
The conclusion follows by picking $\varepsilon<\underline{\lambda}$
such that $2C_{\xi}\zeta_{p}^{2}\varepsilon/3<C_{\xi}^{2}\zeta_{p}^{4}$,
i.e., $\varepsilon<\min\left\{ 3C_{\xi}\zeta_{p}^{2}/2,\underline{\lambda}\right\} $. 
\end{proof}
\begin{lem}\label{lem:remainder.2}
Under Assumptions \ref{asm:unconfounded}-\ref{asm:quality}, for all $0<\varepsilon<\min\left\{ \underline{\lambda},6\tilde{C}_{\xi}\zeta_{p}^{2}\right\} $
and all $n$ such that 
\[
4C_{M}\zeta_{p}^{2}\left(n^{-r_{\gamma_{1}}}+n^{-r_{\gamma_{0}}}+n^{-\frac{r_{\omega_{1}}+r_{\gamma_{1}}}{2}}+n^{-\frac{r_{\omega_{0}}+r_{\gamma_{0}}}{2}}\right)<\varepsilon,
\]
we have
\[
P\left\{ \lambda_{\min}\left(\widehat{A}_{n}\right)<\underline{\lambda}-\varepsilon\right\} \leq2\sieved K\exp\left(\frac{-n\varepsilon^{2}}{16K\tilde{C}_{\xi}^{2}\zeta_{p}^{4}}\right).
\]
\end{lem}

\begin{proof}
Note 
\[
\widehat{A}_{n}=\frac{1}{K}\sum_{k=1}^{K}\left[\frac{1}{m}\sum_{i\in I_{k}}\hat{\xi}_{i}p_{i}p_{i}^{\prime}\right]=\frac{1}{K}\sum_{k=1}^{K}\left[\frac{1}{m}\sum_{i\in I_{k}}\hat{\xi}_{i}^{k}p_{i}p_{i}^{\prime}\right].
\]
Weyl\textquoteright s inequality implies
\[
\lambda_{\min}\left(\widehat{A}_{n}\right)\geq\frac{1}{K}\sum_{k=1}^{K}\lambda_{\min}\left[\frac{1}{m}\sum_{i\in I_{k}}\hat{\xi}_{i}^{k}p_{i}p_{i}^{\prime}\right].
\]
First, fix $0<\varepsilon<\underline{\lambda}$, and  write $\hat{A}_{k}:=\sum_{i\in I_{k}}\frac{\hat{\xi}_{i}^{k}p(W_{i})p^{\prime}(W_{i})}{m}$
and $P_{k}\left\{ \cdotp\right\} :=P_{k}\left\{ \cdotp\mid\left\{ Z_{j}\right\} _{j\in[n]\setminus I_{k}}\right\} $.
Also recall $\mathbb{E}_{k}\left[\cdotp\right]=\mathbb{E}_{k}\left[\cdotp\mid\left\{ Z_{j}\right\} _{j\in[n]\setminus I_{k}}\right]$.
We have
\begin{align*}
&P\left\{ \lambda_{\min}\left(\widehat{A}_{n}\right)<\underline{\lambda}-\varepsilon\right\}  \leq P\left\{ \frac{1}{K}\sum_{k=1}^{K}\lambda_{\min}\left[\hat{A}_{k}\right]<\underline{\lambda}-\varepsilon\right\} \\
 \leq & P\left\{ \lambda_{\min}\left[\hat{A}_{k}\right]<\underline{\lambda}-\varepsilon,\text{for some }k\in[K]\right\}  \leq\sum_{k=1}^{K}P\left\{ \lambda_{\min}\left[\hat{A}_{k}\right]<\underline{\lambda}-\varepsilon\right\} \\
  =&\sum_{k=1}^{K}\mathbb{E}\left[P_{k}\left\{ \lambda_{\min}\left(\hat{A}_{k}\right)<\underline{\lambda}-\varepsilon\right\} \right].
\end{align*}
Then, it suffices to bound $P_{k}\left\{ \lambda_{\min}\left(\hat{A}_{k}\right)<\underline{\lambda}-\varepsilon\right\} $
for each $\forall k\in[K]$. To this end, note Lemma \ref{lem:remainder.3} implies that
\begin{equation}
P_{k}\left\{ \lambda_{\min}\left(\hat{A}_{k}\right)<\underline{\lambda}-\varepsilon\right\} \leq P_{k}\left\{ \lambda_{\min}\left(\hat{A}_{k}\right)<\lambda_{\min}\left(\mathbb{E}_{k}\left[\widehat{A}_{k}\right]\right)-\varepsilon/2\right\} \label{pf:min.eigen.estimated.1}
\end{equation}
for all $n$ such that 
\[
4C_{M}\zeta_{p}^{2}\left(n^{-r_{\gamma_{1}}}+n^{-r_{\gamma_{0}}}+n^{-\frac{r_{\omega_{1}}+r_{\gamma_{1}}}{2}}+n^{-\frac{r_{\omega_{0}}+r_{\gamma_{0}}}{2}}\right)<\varepsilon.
\]
Furthermore, 

\begin{equation}
P_{k}\left\{ \lambda_{\min}\left[\hat{A}_{k}\right]<\lambda_{\min}\left(\mathbb{E}_{k}\left[\widehat{A}_{k}\right]\right)-\varepsilon/2\right\} \leq P_{k}\left\{ \left\Vert \widehat{A}_{k}-\mathbb{E}_{k}\widehat{A}_{k}\right\Vert >\varepsilon/2\right\} .\label{pf:min.eigen.estimated.2}
\end{equation}
And for each $k\in[K]$ and conditional on $\left\{ Z_{j}\right\} _{j\in[n]\setminus I_{k}}$,
$\hat{A}_{k}$ is a sum of $m$ random matrices with independent entries.
Since
\begin{align*}
\left\Vert \hat{\xi}_{i}^{k}p_{i}p_{i}^{\prime}\right\Vert  & \lesssim\tilde{C}_{\xi}\zeta_{p}^{2},\quad
\left\Vert \mathbb{E}_{k}\left[\left(\hat{\xi}_{i}^{k}\right)^{2}p_{i}p_{i}^{\prime}p_{i}p_{i}^{\prime}\right]\right\Vert   \lesssim\zeta_{p}^{4}\tilde{C}_{\xi}^{2},
\end{align*}
it follows by \citet[][Corollary 6.2.1]{tropp2015introduction}  that 
\begin{align*}
P\left\{ \left\Vert \widehat{A}_{k}-\mathbb{E}_{k}\widehat{A}_{k}\right\Vert >\varepsilon\right\}  & \leq2d_{p}\exp\left(\frac{-\frac{n}{K}\frac{\varepsilon^{2}}{8}}{\left(\zeta_{p}^{4}\tilde{C}_{\xi}^{2}+\tilde{C}_{\xi}\zeta_{p}^{2}\varepsilon/6\right)}\right).
\end{align*}
Thus, picking $\varepsilon<\min\left\{ \underline{\lambda},6\tilde{C}_{\xi}\zeta_{p}^{2}\right\} $,
we conclude that 
\begin{equation}
P_{k}\left\{ \left\Vert \widehat{A}_{k}-\mathbb{E}_{k}\widehat{A}_{k}\right\Vert >\varepsilon/2\right\} \leq2d_{p}\exp\left(-\frac{\varepsilon^{2}n}{16K\tilde{C}_{\xi}^{2}\zeta_{p}^{4}}\right)\label{pf:min.eigen.estimated.3}
\end{equation}
for all $k\in[K]$. The conclusion follows by  combining (\ref{pf:min.eigen.estimated.1}),
(\ref{pf:min.eigen.estimated.2}), and (\ref{pf:min.eigen.estimated.3})
\end{proof}
\begin{lem}\label{lem:remainder.3}
Suppose Assumptions \ref{asm:unconfounded}-\ref{asm:quality} hold. For each $\varepsilon>0$, and for each $n$ such that
\[
4C_{M}\zeta_{p}^{2}\left(n^{-r_{\gamma_{1}}}+n^{-r_{\gamma_{0}}}+n^{-\frac{r_{\omega_{1}}+r_{\gamma_{1}}}{2}}+n^{-\frac{r_{\omega_{0}}+r_{\gamma_{0}}}{2}}\right)<\varepsilon,
\]
we have, for all $k\in[K]$,
\[
\lambda_{\min}\left(\mathbb{E}_{k}\left[\widehat{A}_{k}\right]\right)\geq\underline{\lambda}-\frac{\varepsilon}{2}.
\]
\end{lem}

\begin{proof}
Note $\mathbb{E}_{k}\left[\widehat{A}_{k}\right]=\mathbb{E}_{k}\left[\sum_{i\in I_{k}}\frac{\hat{\xi}_{i}^{k}p_{i}p_{i}^{\prime}}{m}\right]=\mathbb{E}_{k}\left[\hat{\xi}^{k}pp^{\prime}\right]$.
Thus,
\begin{align*}
\lambda_{\min}\left(\mathbb{E}_{k}\left[\widehat{A}_{k}\right]\right) & \geq\lambda_{\min}\left(\mathbb{E}_{k}\left[\xi pp^{\prime}\right]\right)+\lambda_{\min}\left(\mathbb{E}_{k}\left[\left(\hat{\xi}^{k}-\xi\right)pp^{\prime}\right]\right)\\
 & =\underline{\lambda}+\lambda_{\min}\left(\mathbb{E}_{k}\left[\left(\hat{\xi}^{k}-\xi\right)pp^{\prime}\right]\right),
\end{align*}
and
\begin{align*}
 & \left|\lambda_{\min}\left(\mathbb{E}_{k}\left(\hat{\xi}^{k}-\xi\right)pp^{\prime}\right)\right|= \left|\min_{\left\Vert a\right\Vert =1}\mathbb{E}_{k}\left[\left(\hat{\xi}^{k}-\xi\right)\left(a^{\prime}p\right)^{2}\right]\right|\\
\leq & \sup_{\left\Vert a\right\Vert =1}\left|\mathbb{E}_{k}\left[\left(\hat{\xi}^{k}-\xi\right)\left(a^{\prime}p\right)^{2}\right]\right|
\leq 2C_{M}\zeta_{p}^{2}\left(n^{-r_{\gamma_{1}}}+n^{-r_{\gamma_{0}}}+n^{-\frac{r_{\omega_{1}}+r_{\gamma_{1}}}{2}}+n^{-\frac{r_{\omega_{0}}+r_{\gamma_{0}}}{2}}\right),
\end{align*}
where the last inequality follows by Lemma \ref{lem:remainder.4}. The conclusion follows
by picking 
\[
2C_{M}\zeta_{p}^{2}\left(n^{-r_{\gamma_{1}}}+n^{-r_{\gamma_{0}}}+n^{-\frac{r_{\omega_{1}}+r_{\gamma_{1}}}{2}}+n^{-\frac{r_{\omega_{0}}+r_{\gamma_{0}}}{2}}\right)\leq\varepsilon/2.
\]
\end{proof}
\begin{lem}\label{lem:remainder.4}
Under Assumptions \ref{asm:unconfounded}-\ref{asm:quality}, we have for each $i\in I_{k}$, $k\in[K]$,
\begin{align*}
\sup_{\left\Vert a\right\Vert =1}\left|\mathbb{E}_{k}\left[\left(\hat{\xi}^{k}-\xi\right)\left(a^{\prime}p_{i}\right)^{2}\right]\right| & \leq2C_{M}\zeta_{p}^{2}\left(n^{-r_{\gamma_{1}}}+n^{-r_{\gamma_{0}}}+n^{-\frac{r_{\omega_{1}}+r_{\gamma_{1}}}{2}}+n^{-\frac{r_{\omega_{0}}+r_{\gamma_{0}}}{2}}\right).
\end{align*}
\end{lem}

\begin{proof}
For each $k\in[K]$, the following decomposition holds:
\begin{align*}
\hat{\xi}^{k}-\xi= & \left[2\left(\gamma_{1}-\gamma_{0}\right)-D\omega_{1}\right]\left(\hat{\gamma}_{1}^{k}-\gamma_{1}\right)+\left[(1-D)\omega_{0}-2\left(\gamma_{1}-\gamma_{0}\right)\right]\left(\hat{\gamma}_{0}^{k}-\gamma_{0}\right)\\
+ & \left(\hat{\gamma}_{1}^{k}-\hat{\gamma}_{0}^{k}-\gamma_{1}+\gamma_{0}\right)^{2}\\
+ & D\left(\hat{\omega}_{1}^{k}-\omega_{1}\right)e_{1}+D\left(\hat{\omega}_{1}^{k}-\omega_{1}\right)(\gamma_{1}-\hat{\gamma}_{1}^{k})\\
- & (1-D)\left(\hat{\omega}_{0}^{k}-\omega_{0}\right)e_{0}-(1-D)\left(\hat{\omega}_{0}^{k}-\omega_{0}\right)(\gamma_{0}-\hat{\gamma}_{0}^{k}).
\end{align*}
For each $a$ such that $\left\Vert a\right\Vert =1$, note:
\begin{align*}
\mathbb{E}_{k}\left[\left[2\left(\gamma_{1}-\gamma_{0}\right)-D\omega_{1}\right]\left(\hat{\gamma}_{1}^{k}-\gamma_{1}\right)\left(a^{\prime}p\right)^{2}\right] & =0,\\
\mathbb{E}_{k}\left[\left[(1-D)\omega_{0}-2\left(\gamma_{1}-\gamma_{0}\right)\right]\left(\hat{\gamma}_{0}^{k}-\gamma_{0}\right)\left(a^{\prime}p\right)^{2}\right] & =0,\\
\mathbb{E}_{k}\left[D\left(\hat{\omega}_{1}^{k}-\omega_{1}\right)e_{1}\left(a^{\prime}p\right)^{2}\right] & =0,\\
\mathbb{E}_{k}\left[(1-D)\left(\hat{\omega}_{0}^{k}-\omega_{0}\right)e_{0}\left(a^{\prime}p\right)^{2}\right] & =0.
\end{align*}
The conclusion follows by noting 
\begin{align*}
 & \sup_{\left\Vert a\right\Vert =1}\left|\mathbb{E}_{k}\left[\left(\hat{\gamma}_{1}^{k}-\hat{\gamma}_{0}^{k}-\gamma_{1}+\gamma_{0}\right)^{2}\left(a^{\prime}p\right)^{2}\right]\right|\\
\leq & \zeta_{p}^{2}\left|\mathbb{E}_{k}\left[\left(\hat{\gamma}_{1}^{k}-\hat{\gamma}_{0}^{k}-\gamma_{1}+\gamma_{0}\right)^{2}\right]\right|\\
\leq & 2\zeta_{p}^{2}\left(\mathbb{E}_{k}\left[\left(\hat{\gamma}_{1}^{k}-\gamma_{1}\right)^{2}\right]+\mathbb{E}_{k}\left[\left(\hat{\gamma}_{0}^{k}-\gamma_{0}\right)^{2}\right]\right)\\
\leq & 2C_{M}\zeta_{p}^{2}\left(n^{-r_{\gamma_{1}}}+n^{-r_{\gamma_{0}}}\right),
\end{align*}
and 
\[
\begin{aligned}\sup_{\left\Vert a\right\Vert =1}\left|\mathbb{E}_{k}\left[D\left(\hat{\omega}_{1}^{k}-\omega_{1}\right)(\gamma_{1}-\hat{\gamma}_{1}^{k})\left(a^{\prime}p\right)^{2}\right]\right| & \leq C_{M}\zeta_{p}^{2}n^{-\frac{r_{\omega_{1}}+r_{\gamma_{1}}}{2}},\\
\sup_{\left\Vert a\right\Vert =1}\left|\mathbb{E}_{k}\left[(1-D)\left(\hat{\omega}_{0}^{k}-\omega_{0}\right)(\gamma_{0}-\hat{\gamma}_{0}^{k})\right]\right| & \leq C_{M}\zeta_{p}^{2}n^{-\frac{r_{\omega_{0}}+r_{\gamma_{0}}}{2}}.
\end{aligned}
\]
\end{proof}
\begin{lem}\label{lem:remainder.5}
Under Assumptions \ref{asm:unconfounded}-\ref{asm:quality}, we have
\begin{align*}
\mathbb{E}\left\Vert \hat{A}_{n}-A_{n}\right\Vert ^{2} & \leq c_3\left(\log2d_{p}\right)^{2}\zeta_{p}^{4}\left(n^{-2r_{\gamma_{1}}}+n^{-2r_{\gamma_{0}}}+n^{-\left(r_{\omega_{1}}+r_{\gamma_{1}}\right)}+n^{-\left(r_{\omega_{0}}+r_{\gamma_{0}}\right)}\right),
\end{align*}
where $c_3$ is a constant that depends on $K$, $\underline{\pi}$, $C_{e},C_{\gamma}$,
$\tilde{C}_{M}$ and $C_{M}$.
\end{lem}
\begin{proof}
As 
\begin{align*}
\hat{A}_{n}-A_{n} & =\frac{1}{K}\sum_{k=1}^{K}\left(\frac{1}{m}\sum_{i\in I_{k}}\left(\hat{\xi}_{i}-\xi_{i}\right)p_{i}p_{i}^{\prime}\right) =\frac{1}{K}\sum_{k=1}^{K}\left(\frac{1}{m}\sum_{i\in I_{k}}\left(\hat{\xi}_{i}^{k}-\xi_{i}\right)p_{i}p_{i}^{\prime}\right),
\end{align*}
triangle inequality implies 
\[
\mathbb{E}_{P^{n}}\left\Vert \hat{A}_{n}-A_{n}\right\Vert ^{2}\leq\frac{1}{K}\sum_{k=1}^{K}\mathbb{E}_{P^{n}}\left[\left\Vert \frac{1}{m}\sum_{i\in I_{k}}\left(\hat{\xi}_{i}^{k}-\xi_{i}\right)p_{i}p_{i}^{\prime}\right\Vert ^{2}\right].
\]
Furthermore, for each $k\in[K]$,
\[
\mathbb{E}_{P^{n}}\left[\left\Vert \frac{1}{m}\sum_{i\in I_{k}}\left(\hat{\xi}_{i}^{k}-\xi_{i}\right)p_{i}p_{i}^{\prime}\right\Vert ^{2}\right]\leq\mathbb{E}_{P^{n}}\left\{ \mathbb{E}_{k}\left[\left\Vert S_{k,1}\right\Vert ^{2}\right]+\left\Vert S_{k,2}\right\Vert ^{2}\right\} ,
\]
where 
\begin{align*}
S_{k,1} & =\frac{1}{m}\sum_{i\in I_{k}}\left(\hat{\xi}_{i}^{k}-\xi_{i}\right)p_{i}p_{i}^{\prime}-\mathbb{E}_{k}\left[\left(\hat{\xi}^{k}-\xi\right)pp^{\prime}\right],\\
S_{k,2} & =\mathbb{E}_{k}\left[\left(\hat{\xi}^{k}-\xi\right)pp^{\prime}\right].
\end{align*}
Conditional on $\left\{ Z_{j}\right\} _{j\in[n]\setminus I_{k}}$,
$S_{k,1}$ is a centered random matrix with independent entries. We
calculate:
\begin{align*}
v_{1} & :=\frac{1}{m}\left\Vert \mathbb{E}_{k}\left[\left(\hat{\xi}^{k}-\xi\right)^{2}pp^{\prime}pp^{\prime}\right]\right\Vert \leq\frac{\zeta_{p}^{4}}{m}\mathbb{E}_{k}\left[\left(\hat{\xi}^{k}-\xi\right)^{2}\right],\\
v_{2} & :=\mathbb{E}_{k}\left[\max_{i\in I_{k}}\left\Vert \frac{\left(\hat{\xi}_{i}^{k}-\xi_{i}\right)p_{i}p_{i}^{\prime}}{m}\right\Vert ^{2}\right] \leq\frac{\mathbb{E}_{k}\left[\left\Vert \left(\hat{\xi}^{k}-\xi\right)pp^{\prime}\right\Vert ^{2}\right]}{m} \leq\frac{\zeta_{p}^{4}}{m}\mathbb{E}_{k}\left[\left(\hat{\xi}^{k}-\xi\right)^{2}\right].
\end{align*}
Applying \citet[][Theorem A.1]{chen2012masked} yields
\begin{align*}
\mathbb{E}_{k}\left[\left\Vert S_{k,1}\right\Vert ^{2}\right] & \leq2\left[2ev_{1}\log2d_{p}+16e^{2}v_{2}\left(\log2d_{p}\right)^{2}\right]\\
 & \leq2\left[2e\frac{\zeta_{p}^{4}}{m}\mathbb{E}_{k}\left[\left(\hat{\xi}^{k}-\xi\right)^{2}\right]\log2d_{p}+16e^{2}\left(\frac{\zeta_{p}^{4}}{m}\mathbb{E}_{k}\left[\left(\hat{\xi}^{k}-\xi\right)^{2}\right]\right)\left(\log2d_{p}\right)^{2}\right]\\
 & =\left[4e\log2d_{p}+16e^{2}\left(\log2d_{p}\right)^{2}\right]\frac{\zeta_{p}^{4}K}{n}\mathbb{E}_{k}\left[\left(\hat{\xi}^{k}-\xi\right)^{2}\right],
\end{align*}
where note under Assumptions \ref{asm:unconfounded}-\ref{asm:quality},
\[
\mathbb{E}_{k}\left[\left(\hat{\xi}^{k}-\xi\right)^{2}\right]\leq c_2C_{M}\left(n^{-r_{\gamma_{1}}}+n^{-r_{\gamma_{0}}}+n^{-r_{\omega_{1}}}+n^{-r_{\omega_{0}}}\right)
\]
for some constant $c_2$ that only depends on $\underline{\pi},C_{e},C_{\gamma}$,
and $\tilde{C}_{M}$. For $\left\Vert S_{k,2}\right\Vert$, note Lemma \ref{lem:remainder.4} implies
\begin{align*}
\left\Vert S_{k,2}\right\Vert  & =\sup_{\left\Vert a\right\Vert =1}\mathbb{E}_{k}\left[\left(\hat{\xi}^{k}-\xi\right)\left(a^{\prime}p\right)^{2}\right]\\
 & \leq\sup_{\left\Vert a\right\Vert =1}\left|\mathbb{E}_{k}\left[\left(\hat{\xi}^{k}-\xi\right)\left(a^{\prime}p\right)^{2}\right]\right|\\
 & \leq2C_{M}\zeta_{p}^{2}\left(n^{-r_{\gamma_{1}}}+n^{-r_{\gamma_{0}}}+n^{-\frac{r_{\omega_{1}}+r_{\gamma_{1}}}{2}}+n^{-\frac{r_{\omega_{0}}+r_{\gamma_{0}}}{2}}\right).
\end{align*}
Thus, we conclude
\begin{align*}
\mathbb{E}\left\Vert \hat{A}_{n}-A_{n}\right\Vert ^{2} & \leq\left[4e\log2d_{p}+16e^{2}\left(\log2d_{p}\right)^{2}\right]\frac{Kc_{2}C_{M}\zeta_{p}^{4}}{n}\left(n^{-r_{\gamma_{1}}}+n^{-r_{\gamma_{0}}}+n^{-r_{\omega_{1}}}+n^{-r_{\omega_{0}}}\right)\\
 & +16C_{M}^{4}\zeta_{p}^{4}\left(n^{-2r_{\gamma_{1}}}+n^{-2r_{\gamma_{0}}}+n^{-\left(r_{\omega_{1}}+r_{\gamma_{1}}\right)}+n^{-\left(r_{\omega_{0}}+r_{\gamma_{0}}\right)}\right)\\
 & \leq c\left(\log2d_{p}\right)^{2}\zeta_{p}^{4}\left(n^{-2r_{\gamma_{1}}}+n^{-2r_{\gamma_{0}}}+n^{-\left(r_{\omega_{1}}+r_{\gamma_{1}}\right)}+n^{-\left(r_{\omega_{0}}+r_{\gamma_{0}}\right)}\right),
\end{align*}
where $c_3$ depends on $K,c_{2}$ and $C_{M}$.
\end{proof}

\begin{lem}\label{lem:remainder.6}
\begin{itemize}
\item[(i)] Under Assumptions \ref{asm:unconfounded}-\ref{asm:quality}, we have
\begin{align*}
\mathbb{E}\left\Vert \hat{B}_{n}-B_{n}\right\Vert ^{2} & \leq c_4\left[\frac{\zeta_{p}^{2}}{n}+\sieved\zeta_{p}^{2}\left[\left(n^{-2r_{\gamma_{1}}}+n^{-2r_{\gamma_{0}}}\right)+n^{-(r_{\omega_{1}}+r_{\gamma_{1}})}+n^{-(r_{\omega_{0}}+r_{\gamma_{0}})}\right]\right]
\end{align*}
for some $c_4$ that depends on $K$, $C_{M}$, $C_{\xi}$ and $\tilde{C}_{\xi}$.

\item[(ii)] Suppose in addition,  Assumption \ref{asm:margin} also holds, and  $n$ is such that
\[
\left(4C_{M}C_{\tau}^{-1}\left(n^{-r_{\gamma_{1}}}+n^{-r_{\gamma_{0}}}\right)\right)^{\frac{1}{\alpha+2}}<t^{*}.
\]
Then, we have
\begin{align*}
\mathbb{E}\left\Vert \hat{B}_{n}-B_{n}\right\Vert ^{2} & \leq\frac{c_5\zeta_{p}^{2}}{n}\left[\left(n^{-r_{\gamma_{1}}}+n^{-r_{\gamma_{0}}}\right)^{\frac{\alpha}{\alpha+2}}+n^{-r_{\omega_{1}}}+n^{-r_{\omega_{0}}}\right]\\
 & +c_5d_{p}\zeta_{p}^{2}\left[\left(n^{-2r_{\gamma_{1}}}+n^{-2r_{\gamma_{0}}}\right)+n^{-(r_{\omega_{1}}+r_{\gamma_{1}})}+n^{-(r_{\omega_{0}}+r_{\gamma_{0}})}\right],
\end{align*}
where $c_5$ is a suitably redefined constant that also depends on $\underline{\pi},C_{\tau},C_{e}$,
$\tilde{C}_{M}$, $C_{\tau}$ and $\alpha$.
\end{itemize}
\end{lem}

\begin{proof}
\textbf{Statement (i)}. As
\begin{align*}
\hat{B}_{n}-B_{n} & =\left[\frac{1}{K}\sum_{k=1}^{K}\frac{1}{m}\sum_{i\in I_{k}}\left(\hat{\xi}_{i}\mathbf{1}\left\{ \hat{\tau}_{i}\geq0\right\} -\xi_{i}\mathbf{1}\left\{ \tau_{i}\geq0\right\} \right)p_{i}\right]\\
 & =\left[\frac{1}{K}\sum_{k=1}^{K}\frac{1}{m}\sum_{i\in I_{k}}\left(\hat{\xi}_{i}^{k}\mathbf{1}\left\{ \hat{\tau}_{i}^{k}\geq0\right\} -\xi_{i}\mathbf{1}\left\{ \tau_{i}\geq0\right\} \right)p_{i}\right],
\end{align*}
analogous arguments to Lemma \ref{lem:remainder.4} imply that it suffices to bound for each $k\in[K]$
\begin{align*}
 & \mathbb{E}_{k}\left\Vert \frac{1}{m}\sum_{i\in I_{k}}\left(\hat{\xi}_{i}^{k}\mathbf{1}\left\{ \hat{\tau}_{i}^{k}\geq0\right\} -\xi_{i}\mathbf{1}\left\{ \tau_{i}\geq0\right\} \right)p_{i}\right\Vert ^{2}
\leq \left|\mathbb{E}_{k}\left[S_{k,3}\right]\right|+\left|\mathbb{E}_{k}\left[S_{k,4}\right]\right|,
\end{align*}
where
\begin{align*}
\mathbb{E}_{k}\left[S_{k,3}\right] & :=\frac{1}{m}\mathbb{E}_{k}\left[\left(\hat{\xi}^{k,+}-\xi^{+}\right)^{2}p^{\prime}p\right],\\
\mathbb{E}_{k}\left[S_{k,4}\right] & :=\frac{1}{m(m-1)}\sum_{i,j\in I_{k},i\neq j}\mathbb{E}_{k}\left(\hat{\xi}_{i}^{k,+}-\xi_{i}^{+}\right)\left(\hat{\xi}_{j}^{k,+}-\xi_{j}^{+}\right)p_{i}^{\prime}p_{j},\\
\hat{\xi}^{k,+} & :=\hat{\xi}^{k}\mathbf{1}\left\{ \hat{\tau}^{k}(X)\geq0\right\} ,\\
\xi^{+} & :=\xi\mathbf{1}\left\{ \tau(X)\geq0\right\} .
\end{align*}
For term $S_{k,3}$, note under Assumptions \ref{asm:unconfounded}-\ref{asm:quality}:
\begin{align}
\left|\mathbb{E}_{k}\left[S_{k,3}\right]\right| & \leq\frac{\zeta_{p}^{2}}{m}\mathbb{E}_{k}\left[\left(\hat{\xi}^{k,+}-\xi^{+}\right)^{2}\right]\leq\frac{K\zeta_{p}^{2}}{n}\left(2C_{\xi}^{2}+2\tilde{C}_{\xi}^{2}\right).\label{pf:hat.B.1}
\end{align}
For term $S_{k,4}$, note conditional on $\left\{ Z_{j}\right\} _{j\in[n]\setminus I_{k}}$,
$\left(\hat{\xi}_{i}^{k,+}-\xi_{i}^{+}\right)p_{i}$ and $\left(\hat{\xi}_{j}^{k,+}-\xi_{j}^{+}\right)p_{j}$
are independent. Therefore, for each $i,j\in I_{k},i\neq j$:
\begin{align*}
 & \mathbb{E}_{k}\left(\hat{\xi}_{i}^{k,+}-\xi_{i}^{+}\right)\left(\hat{\xi}_{j}^{k,+}-\xi_{j}^{+}\right)p_{i}^{\prime}p_{j}\\
= & \left[\mathbb{E}_{k}\left(\hat{\xi}_{i}^{k,+}-\xi_{i}^{+}\right)p_{i}\right]^{\prime}\mathbb{E}_{k}\left[\left(\hat{\xi}_{j}^{k,+}-\xi_{j}^{+}\right)p_{j}\right]\\
= & \sum_{t=1}^{d_{p}}\mathbb{E}_{k}\left[\left(\hat{\xi}_{i}^{k,+}-\xi_{i}^{+}\right)p_{t,i}\right]\mathbb{E}_{k}\left[\left(\hat{\xi}_{j}^{k,+}-\xi_{j}^{+}\right)p_{t,j}\right]\\
= & \sum_{t=1}^{d_{p}}\left\{ \mathbb{E}_{k}\left[\left(\hat{\xi}^{k,+}-\xi^{+}\right)p_{t}\right]\right\} ^{2}.
\end{align*}
Applying Lemma \ref{lem:remainder.7} yields
\begin{equation}
\sum_{t=1}^{d_{p}}\left\{ \mathbb{E}_{k}\left[\left(\hat{\xi}^{k,+}-\xi^{+}\right)p_{t}\right]\right\} ^{2}\leq d_{p}\zeta_{p}^{2}C_{M}^{2}\left[12\left(n^{-r_{\gamma_{1}}}+n^{-r_{\gamma_{0}}}\right)+n^{-\frac{r_{\omega_{1}}+r_{\gamma_{1}}}{2}}+n^{-\frac{r_{\omega_{0}}+r_{\gamma_{0}}}{2}}\right]^{2}.\label{pf:hat.B.2}
\end{equation}
The conclusion follows by combining (\ref{pf:hat.B.1}) and (\ref{pf:hat.B.2}).

\noindent \textbf{Statement (ii)}. The proof is analogous to that of statement (i), but with a new upper bound for term $\mathbb{E}_{k}\left[S_{k,3}\right]$ established in Lemma \ref{lem:remainder.8}. 
\end{proof}

\begin{lem}\label{lem:remainder.7}
Under Assumptions \ref{asm:unconfounded}-\ref{asm:quality}, we have, for all $k\in[K]$ and each $j\in[\sieved]$:
\[
\left|\mathbb{E}_{k}\left[\left(\hat{\xi}^{k,+}-\xi^{+}\right)p_{j}\right]\right|\leq\zeta_{p}C_{M}\left[n^{-\frac{r_{\omega_{1}}+r_{\gamma_{1}}}{2}}+n^{-\frac{r_{\omega_{0}}+r_{\gamma_{0}}}{2}}+12\left(n^{-r_{\gamma_{1}}}+n^{-r_{\gamma_{0}}}\right)\right]
\]
\end{lem}

\begin{proof}
Since the function $f(x)=x^{2}\mathbf{1}\{x\geq0\}$ is continuously
differentiable, first order Taylor expansion yields, for each $k\in[K]$:
\begin{align*}
 & \left(\hat{\gamma}_{1}^{k}-\hat{\gamma}_{0}^{k}\right)^{2}\mathbf{1}\{\hat{\gamma}_{1}^{k}-\hat{\gamma}_{0}^{k}\geq0\}-\left(\gamma_{1}-\gamma_{0}\right)^{2}\mathbf{1}\{\gamma_{1}-\gamma_{0}\geq0\}\\
= & \left(\tilde{\gamma}_{1}^{k}-\tilde{\gamma}_{0}^{k}\right)\mathbf{1}\{\tilde{\gamma}_{1}^{k}-\tilde{\gamma}_{0}^{k}\geq0\}\left(\left(\hat{\gamma}_{1}^{k}-\gamma_{1}\right)-\left(\hat{\gamma}_{0}^{k}-\gamma_{0}\right)\right),
\end{align*}
where 
\begin{align*}
\tilde{\gamma}_{1}^{k}(x) & :=\gamma_{1}(x)+\tilde{t}^{k}(x)(\hat{\gamma}_{1}^{k}(x)-\gamma_{1}(x)),\\
\tilde{\gamma}_{0}^{k}(x) & :=\gamma_{0}(x)+\tilde{t}^{k}(x)(\hat{\gamma}_{0}^{k}(x)-\gamma_{0}(x)),
\end{align*}
and $\tilde{t}^{k}(x)\in(0,1)$ for all $x\in\mathcal{X}$. After
lengthy algebra, we arrive at the following decomposition:
\begin{align*}
\hat{\xi}^{k,+}-\xi^{+}= & \Delta_{k,1}+\Delta_{k,2}+\Delta_{k,3}+\Delta_{k,4}+\Delta_{k,5}+\Delta_{k,6}\\
- & \Delta_{k,7}-\Delta_{k,8}-\Delta_{k,9}-\Delta_{k,10},
\end{align*}
where
\begin{align*}
\Delta_{k,1} & :=\left(2\left(\gamma_{1}-\gamma_{0}\right)-D\omega_{1}\right)\left(\hat{\gamma}_{1}^{k}-\gamma_{1}\right)\mathbf{1}\{\hat{\gamma}_{1}^{k}-\hat{\gamma}_{0}^{k}\geq0\},\\
\Delta_{k,2} & =D\omega_{1}e_{1}\mathbf{1}\{\hat{\gamma}_{1}^{k}-\hat{\gamma}_{0}^{k}\geq0\}-D\omega_{1}e_{1}\mathbf{1}\{\gamma_{1}-\gamma_{0}\geq0\},\\
\Delta_{k,3} & =D\left(\hat{\omega}_{1}^{k}-\omega_{1}\right)e_{1}\mathbf{1}\{\hat{\gamma}_{1}^{k}-\hat{\gamma}_{0}^{k}\geq0\},\\
\Delta_{k,4} & =D\left(\hat{\omega}_{1}^{k}-\omega_{1}\right)(\gamma_{1}-\hat{\gamma}_{1}^{k})\mathbf{1}\{\hat{\gamma}_{1}^{k}-\hat{\gamma}_{0}^{k}\geq0\},\\
\Delta_{k,5} & =2\left[\left(\hat{\gamma}_{1}^{k}-\hat{\gamma}_{0}^{k}-\left(\gamma_{1}-\gamma_{0}\right)\right)^{2}\mathbf{1}\{\hat{\gamma}_{1}^{k}-\hat{\gamma}_{0}^{k}\geq0\}\right],\\
\Delta_{k,6} & =\left[\left(\tilde{\gamma}_{1}^{k}-\tilde{\gamma}_{0}^{k}\right)\mathbf{1}\{\tilde{\gamma}_{1}^{k}-\tilde{\gamma}_{0}^{k}\geq0\}-\left(\hat{\gamma}_{1}^{k}-\hat{\gamma}_{0}^{k}\right)\mathbf{1}\{\hat{\gamma}_{1}-\hat{\gamma}_{0}\geq0\}\right]\left(\left(\hat{\gamma}_{1}^{k}-\gamma_{1}\right)-\left(\hat{\gamma}_{0}^{k}-\gamma_{0}\right)\right),
\end{align*}
and 
\begin{align*}
\Delta_{k,7} & :=\left(2\left(\gamma_{1}-\gamma_{0}\right)-(1-D)\omega_{0})(\hat{\gamma}_{0}^{k}-\gamma_{0})\right)\mathbf{1}\{\hat{\gamma}_{1}^{k}-\hat{\gamma}_{0}^{k}\geq0\},\\
\Delta_{k,8} & :=(1-D)\omega_{0}e_{0}\mathbf{1}\{\hat{\gamma}_{1}^{k}-\hat{\gamma}_{0}^{k}\geq0\}-(1-D)\omega_{0}e_{0}\mathbf{1}\{\gamma_{1}-\gamma_{0}\geq0\},\\
\Delta_{k,9} & :=(1-D)(\hat{\omega}_{0}^{k}-\omega_{0})e_{0}\mathbf{1}\{\hat{\gamma}_{1}^{k}-\hat{\gamma}_{0}^{k}\geq0\},\\
\Delta_{k,10} & :=(1-D)(\hat{\omega}_{0}^{k}-\omega_{0})(\gamma_{0}-\hat{\gamma}_{0}^{k})\mathbf{1}\{\hat{\gamma}_{1}^{k}-\hat{\gamma}_{0}^{k}\geq0\}.
\end{align*}
Note 
\begin{align*}
\mathbb{E}_{k}\left[\Delta_{k,1}p_{j}\right] & =0, & \mathbb{E}_{k}\left[\Delta_{k,7}p_{j}\right],=0,\\
\mathbb{E}_{k}\left[\Delta_{k,2}p_{j}\right] & =0, & \mathbb{E}_{k}\left[\Delta_{k,8}p_{j}\right]=0,\\
\mathbb{E}_{k}\left[\Delta_{k,3}p_{j}\right] & =0, & \mathbb{E}_{k}\left[\Delta_{k,9}p_{j}\right]=0.
\end{align*}
Thus, $\left|\mathbb{E}_{k}\left[\left(\hat{\xi}^{k,+}-\xi^{+}\right)p_{j}\right]\right|$
is determined by
\begin{align*}
\mathbb{E}_{k}\left[\Delta_{k,4}p_{j}\right] & , & \mathbb{E}_{k}\left[\Delta_{k,5}p_{j}\right],\\
\mathbb{E}_{k}\left[\Delta_{k,6}p_{j}\right] & =0, & \mathbb{E}_{k}\left[\Delta_{k,10}p_{j}\right].
\end{align*}
It is straigtforward to show
\begin{align*}
\left|\mathbb{E}_{k}\left[\Delta_{k,4}p_{j}\right]\right| & \leq\zeta_{p}C_{M}n^{-\frac{r_{\omega_{1}}+r_{\gamma_{1}}}{2}},\\
\left|\mathbb{E}_{k}\left[\Delta_{k,10}p_{j}\right]\right| & \leq\zeta_{p}C_{M}n^{-\frac{r_{\omega_{0}}+r_{\gamma_{0}}}{2}},\\
\left|\mathbb{E}_{k}\left[\Delta_{k,5}p_{j}\right]\right| & \leq4\zeta_{p}C_{M}\left(n^{-r_{\gamma_{1}}}+n^{-r_{\gamma_{0}}}\right).
\end{align*}
For tem $\left|\mathbb{E}_{k}\left[\Delta_{k,6}p_{t}\right]\right|$,
note
\begin{align*}
\mathbb{E}_{k}\left[\Delta_{k,6}p_{j}\right]= & \mathbb{E}_{k}\left[\Delta_{k,6}p_{j}\mathbf{1}\{\tilde{\gamma}_{1}^{k}-\tilde{\gamma}_{0}^{k}\geq0,\hat{\gamma}_{1}^{k}-\hat{\gamma}_{0}^{k}\geq0\}\right]\\
+ & \mathbb{E}_{k}\left[\Delta_{k,6}p_{j}\mathbf{1}\{\tilde{\gamma}_{1}^{k}-\tilde{\gamma}_{0}^{k}<0,\hat{\gamma}_{1}^{k}-\hat{\gamma}_{0}^{k}<0\}\right]\\
+ & \mathbb{E}_{k}\left[\Delta_{k,6}p_{j}\mathbf{1}\{\tilde{\gamma}_{1}^{k}-\tilde{\gamma}_{0}^{k}\geq0,\hat{\gamma}_{1}^{k}-\hat{\gamma}_{0}^{k}<0\}\right]\\
+ & \mathbb{E}_{k}\left[\Delta_{k,6}p_{j}\mathbf{1}\{\tilde{\gamma}_{1}^{k}-\tilde{\gamma}_{0}^{k}<0,\hat{\gamma}_{1}^{k}-\hat{\gamma}_{0}^{k}<0\}\right].
\end{align*}
Furthermore, we have
\begin{align*}
 & \left|\mathbb{E}_{k}\left[\Delta_{k,6}p_{j}\mathbf{1}\{\tilde{\gamma}_{1}^{k}-\tilde{\gamma}_{0}^{k}\geq0,\hat{\gamma}_{1}^{k}-\hat{\gamma}_{0}^{k}\geq0\}\right]\right|\\
\leq & \mathbb{E}_{k}\left[\left(\left(\hat{\gamma}_{1}^{k}-\gamma_{1}\right)-\left(\hat{\gamma}_{0}^{k}-\gamma_{0}\right)\right)^{2}\left|p_{j}\right|\right]\\
\leq & 4\zeta_{p}C_{M}\left(n^{-r_{\gamma_{1}}}+n^{-r_{\gamma_{0}}}\right),
\end{align*}
and 
\[
\mathbb{E}_{k}\left[\Delta_{k,6}p_{j}\mathbf{1}\{\tilde{\gamma}_{1}^{k}-\tilde{\gamma}_{0}^{k}<0,\hat{\gamma}_{1}^{k}-\hat{\gamma}_{0}^{k}<0\}\right]=0.
\]
When $\tilde{\gamma}_{1}^{k}-\tilde{\gamma}_{0}^{k}\geq0$ and $\hat{\gamma}_{1}^{k}-\hat{\gamma}_{0}^{k}<0$,
note
\[
\left|\Delta_{k,6}\right|\leq\left(\left(\hat{\gamma}_{1}^{k}-\gamma_{1}\right)-\left(\hat{\gamma}_{0}^{k}-\gamma_{0}\right)\right)^{2}.
\]
Thus, 
\begin{align*}
 & \mathbb{E}_{k}\left[\Delta_{k,6}p_{j}\mathbf{1}\{\tilde{\gamma}_{1}^{k}-\tilde{\gamma}_{0}^{k}\geq0,\hat{\gamma}_{1}^{k}-\hat{\gamma}_{0}^{k}<0\}\right]
\leq 2\zeta_{p}C_{M}\left(n^{-r_{\gamma_{1}}}+n^{-r_{\gamma_{0}}}\right),
\end{align*}
and analogously, 
\begin{align*}
 & \mathbb{E}_{k}\left[\Delta_{k,6}p_{j}\mathbf{1}\{\tilde{\gamma}_{1}^{k}-\tilde{\gamma}_{0}^{k}<0,\hat{\gamma}_{1}^{k}-\hat{\gamma}_{0}^{k}\geq0\}\right]
\leq  2\zeta_{p}C_{M}\left(n^{-r_{\gamma_{1}}}+n^{-r_{\gamma_{0}}}\right).
\end{align*}
Therefore, we conclude
\begin{align*}
 & \left|\mathbb{E}_{k}\left[\Delta_{k,6}p_{j}\right]\right|\leq8\zeta_{p}C_{M}\left(n^{-r_{\gamma_{1}}}+n^{-r_{\gamma_{0}}}\right).
\end{align*}

\end{proof}
\begin{lem}\label{lem:remainder.8}
Suppose Assumptions \ref{asm:unconfounded}-\ref{asm:margin} hold. Then, for all $n$ such that
\[
\left(4C_{M}C_{\tau}^{-1}\left(n^{-r_{\gamma_{1}}}+n^{-r_{\gamma_{0}}}\right)\right)^{\frac{1}{\alpha+2}}<t^{*},
\]
we have
\begin{align*}
\left|\mathbb{E}_{k}\left[S_{k,3}\right]\right|\leq\frac{\zeta_{p}^{2}K}{n} & c_{1}^{+}C_{M}\left\{ n^{-r_{\gamma_{1}}}+n^{-r_{\gamma_{0}}}+n^{-r_{\omega_{1}}}+n^{-r_{\omega_{0}}}\right\} \\
+\frac{\zeta_{p}^{2}K}{n} & c_{3}^{+}\left(n^{-r_{\gamma_{1}}}+n^{-r_{\gamma_{0}}}\right)^{\frac{\alpha}{\alpha+2}},
\end{align*}
for some finite constants $c_{1}^{+}$ and $c_{3}^{+}$ defined below
in the proof. 
\end{lem}

\begin{proof}
By the decomposition result for $\hat{\xi}^{k,+}-\xi^{+}$, and under
Assumptions \ref{asm:unconfounded}-\ref{asm:quality}, it is straightforward to see that
\begin{align*}
\mathbb{E}_{k}\left[\left(\hat{\xi}^{k,+}-\xi^{+}\right)^{2}\right] & \leq c_{1}^{+}C_{M}\left\{ n^{-r_{\gamma_{1}}}+n^{-r_{\gamma_{0}}}+n^{-r_{\omega_{1}}}+n^{-r_{\omega_{0}}}\right\} \\
 & +c_{2}^{+}\mathbb{E}_{k}\left[\left(\mathbf{1}\{\hat{\gamma}_{1}^{k}-\hat{\gamma}_{0}^{k}\geq0\}-\mathbf{1}\{\gamma_{1}-\gamma_{0}\geq0\}\right)^{2}\right],
\end{align*}
where, $c_{1}^{+}$ and $c_{2}^{+}$ are some constants that depend
on $\underline{\pi},C_{\tau},C_{e}$ and $\tilde{C}_{M}$. With  Assumption
\ref{asm:margin}, we have
\begin{align*}
 & \mathbb{E}_{k}\left[\left(\mathbf{1}\{\hat{\gamma}_{1}^{k}-\hat{\gamma}_{0}^{k}\geq0\}-\mathbf{1}\{\gamma_{1}-\gamma_{0}\geq0\}\right)^{2}\right]\\
= & \mathbb{E}_{k}\left[\mathbf{1}\left\{ \text{sgn}\left(\hat{\gamma}_{1}^{k}-\hat{\gamma}_{0}^{k}\right)\neq\text{sgn}\left(\tau\right)\right\} \right]\\
\leq & \mathbb{E}_{k}\left[\mathbf{1}\left\{ \text{sgn}\left(\hat{\gamma}_{1}^{k}-\hat{\gamma}_{0}^{k}\right)\neq\text{sgn}\left(\tau\right)\right\} \mathbf{1}\left\{ \left|\tau\right|\leq t\right\} \right]\\
+ & \mathbb{E}_{k}\left[\mathbf{1}\left\{ \text{sgn}\left(\hat{\gamma}_{1}^{k}-\hat{\gamma}_{0}^{k}\right)\neq\text{sgn}\left(\tau\right)\right\} \mathbf{1}\left\{ \left|\tau\right|>t\right\} \right]\\
\leq & P\left\{ \left|\tau(X)\right|\leq t\right\} +P\left\{ \left|\hat{\gamma}_{1}^{k}-\hat{\gamma}_{0}^{k}-\tau\right|>t\right\} \\
\leq & C_{\tau}t^{\alpha}+\frac{\mathbb{E}_{k}\left[\left(\hat{\gamma}_{1}^{k}-\hat{\gamma}_{0}^{k}-\tau\right)^{2}\right]}{t^{2}}\\
\leq & C_{\tau}t^{\alpha}+\frac{2C_{M}\left(n^{-r_{\gamma_{1}}}+n^{-r_{\gamma_{0}}}\right)}{t^{2}}.
\end{align*}
Optimizing over $t$ and choosing $n$ large enough so that 
\[
\left(\frac{4C_{M}\left(n^{-r_{\gamma_{1}}}+n^{-r_{\gamma_{0}}}\right)}{C_{\tau}}\right)^{\frac{1}{\alpha+2}}<t^{*}
\]
 yields 
\begin{align*}
 & \mathbb{E}_{k}\left[\left(\mathbf{1}\{\hat{\gamma}_{1}^{k}-\hat{\gamma}_{0}^{k}\geq0\}-\mathbf{1}\{\gamma_{1}-\gamma_{0}\geq0\}\right)^{2}\right]\\
\leq & c_{3}^{+}\left(n^{-r_{\gamma_{1}}}+n^{-r_{\gamma_{0}}}\right)^{\frac{\alpha}{\alpha+2}},
\end{align*}
where $c_{3}^{+}$ is a finite constant that depends on $C_{M}$,
$C_{\tau}$ and $\alpha$. The conclusion follows by combining the
above results with (\ref{pf:hat.B.1}). 
\end{proof}

\section{Minimax lower bound}\label{sec:lower.bound}

To gauge the optimality of the convergence rate derived in Theorem \ref{thm:main} in case of a nonparametric $\delta^*$, a minimax lower bound is needed. To proceed, we introduce the following standard class of smooth functions.  
\begin{defn}[smoothness]
Let $s=s_{0}+t$ for some $s_{0}\in\mathbb{N}_{+}$ and $0<t\leq1$,
and let $C>0$. A function $f:\mathbb{R}^{d}\rightarrow\mathbb{R}$
is $(s,C)$-smooth if for every $\alpha=(\alpha_{1},\ldots,\alpha_{d})$,
$\alpha_{i}\in\mathbb{N}_{+}$, $\sum_{i=1}^{d}\alpha_{i}=s_{0}$,
the partial derivative $\frac{\partial^{s_{0}}}{\partial x_{1}^{\alpha_{1}}\cdots\partial x_{d}^{\alpha_{d}}}$
exists and satisfies 
\[
\left|\frac{\partial^{s_{0}}f}{\partial x_{1}^{\alpha_{1}}\cdots\partial x_{d}^{\alpha_{d}}}(x)-\frac{\partial^{s_{0}}f}{\partial x_{1}^{\alpha_{1}}\cdots\partial x_{d}^{\alpha_{d}}}(z)\right|\leq C\left\Vert x-z\right\Vert ^{t},x,z\in\mathbb{R}^{d}.
\]
Denote by $\mathcal{F}^{(s,C)}$ the set of all $(s,C)$-smooth functions
$f:\mathbb{R}^{d}\rightarrow\mathbb{R}$. 
\end{defn}
We then consider the following class of distributions,  a subset of $\mathcal{P}$, in which the form of $\delta^*$ is simple and smooth and $W$ is uniformly distributed in $[0,1]^{d_W}$. 

\begin{defn}\label{def:DGP_lower_bound}
Let $X=\left(X_{1},W\right)$. Denote by $\mathcal{P}(s,C)\subseteq\mathcal{P}$
the class of distributions of $(Y(1),Y(0),D,X)$ such that:
\begin{itemize}
\item[(i)] $W$ is uniformly distributed in $[0,1]^{d_{W}}$; 
\item[(ii)] $X_{1}=\text{sgn}\left(m(W)+\varepsilon\right)$,\footnote{We define $\text{sgn}(0):=1$ as a convention.}
where $m:\mathbb{R}^{d_{W}}\rightarrow[0,1]$ is such that $m\in\mathcal{F}(s,C)$,
and $\varepsilon$ is uniformly distributed in $[-1,0]$ and independent
of $W$;
\item[(iii)] $\tau(x_{1},w)=x_{1}$ for all $x_{1}\in\left\{ -1,1\right\} $,
$w\in[0,1]^{d_{W}}$, where 
\begin{align*}
\tau(x_{1},w) & =\mathbb{E}[Y(1)\mid X_{1}=x_{1},W=w]-\mathbb{E}[Y(0)\mid X_{1}=x_{1},W=w];
\end{align*}
\item[(iv)] The joint distribution of $(Y(1),Y(0),D,X)$ satisfies Assumption \ref{asm:unconfounded}. 
\end{itemize}
In the $\mathcal{P}(s,C)$ class defined above, $\tau(x)=\tau(x_{1},w)=x_{1}\in\{-1,1\}$,
so that $\tau^{2}(x)=1$ for all $x$. It follows the optimal rule
($\alpha=2$) becomes 
\begin{align*}
\delta^{*}(W) & =\mathbb{E}\left[\mathbf{1}\{X_{1}\geq0\}\mid W\right] =Pr\left\{ X_{1}=1\mid W\right\} =m(W),
\end{align*}
a conditional expectation of a bounded outcome $\mathbf{1}\{X_{1}\geq0\}$. Thus, we can  study a minimax lower bound for $\delta^*$ by assessing a minimax lower bound for the conditional expectation function, for which we know relatively well.  Inspired by \citet[][Theorem 3.2 and Problem 3.3]{gyorfi2006distribution}, we derive the following
result. 
\end{defn}

\begin{thm}\label{thm:lower.bound}
For the class $\mathcal{P}(s,C)$, consider any rule
$\delta_{n}$ that depends on data $Z^{n}$. Then, we have
\[
\liminf_{n\rightarrow\infty}\frac{\sup_{P\in\mathcal{P}(s,c)}\mathbb{E}_{P_{n}}\left[L(\delta_{n},\tau)-L(\delta^{*},\tau)\right]}{C^{\frac{2d_{W}}{6s+3d_{W}}}n^{-\frac{2s}{2s+d_{W}}}}\geq C_{1}
\]
for some $C_{1}>0$ independent of $C$.
\end{thm}

Theorem \ref{thm:lower.bound} is established with the following useful lemma.

\begin{lem}\label{lem:Bayes_classification}
Let $u=(u_{1},\ldots, u_{N})$ be a vector of dimension $N$
such that $u_{j}\in[0,1]$ for all $j=1\ldots N$, and let $\mathrm{C}$
be a random variable taking values in $\left\{ -1,1\right\} $ with
equal probability. Write $\mathbf{Y}=(Y_{1},\ldots,Y_{N})^{\prime}$,
where 
\[
Y_{j}=\text{ \ensuremath{\frac{1}{2}}+\ensuremath{\frac{1}{2}}}\mathrm{C}u_{j}+\varepsilon_{j},j=1\ldots N,
\]
and $\left\{ \varepsilon_{j}\right\} _{j=1}^{N}$ are iid with a uniform
distribution in $[-1,0].$ Let $g^{*}:\mathbb{R}^{N}\rightarrow\left\{ -1,1\right\} $
be the Bayes  decision for $\mathrm{C}$ based on \textup{$\mathbf{Y}$}.
It follows 
\[
Pr\left\{ g^{*}(\mathbf{Y})\neq\mathrm{C}\right\} \geq 1-N\left\Vert u\right\Vert .
\]
\end{lem}

\section{Computational details}\label{sec:compute}
Our practical procedure for selecting $\lambda_{1,n}$ in \eqref{eq:practical.omega} is motivated by the following simple observation: $\omega_1$ also satisfies
\begin{equation}
\mathbb{E}\left[D\omega_{1}(X)Y\right]=\mathbb{E}\left[2\left(\gamma_{1}(X)-\gamma_{0}(X)\right)\gamma_{1}(X)\right].\label{eq:cv.criterion}
\end{equation}
Thus, we select $\lambda_{1,n}$ via cross-validation to minimize
an out-of-sample prediction error criterion that mimics (\ref{eq:cv.criterion}) detailed in the algorithm below. 

We now explain how to calculate $\hat{\xi}(Z_{i})$
given in our main proposal. For each $I_{k}$,
$k\in[K]$, calculation of $\hat{\gamma}_{1}^{k}$ and $\hat{\gamma}_{0}^{k}$
is straightforward and any machine learnt estimators can be applied.
We construct $\hat{\omega}_{1}^{k}$ and $\hat{\omega}_{0}^{k}$ for
$k\in[K]$ using the following $10$-fold cross validation procedure
(other number of folds can be considered completely analogously).
Let $\Lambda=\left\{ 0.1,0.2,0.3\ldots,4.9,5\right\} $ be a set of
penalty candidates. For each $k\in[K]$:
\begin{enumerate}
\item Split $I_{k}^{c}$ into approximately equal sized 10 folds, call them
$I_{k,1}^{c}$, $I_{k,2}^{c}$,..., $I_{k,10}^{c}$. For each fold
$I_{k,j}^{c}$, $j=1\ldots10$:
\begin{itemize}
\item[(a)] Estimate $\gamma_{1}$ and $\gamma_{0}$ using any estimator\footnote{It should be the same procedure that produced $\hat{\gamma}_{1}^{k}$
and $\hat{\gamma}_{0}^{k}$. For example, in our empirical application,
$\gamma_{1}$ and $\gamma_{0}$ are estimated via 10-fold cross-validated
lasso in glmnet package.} with data in $I_{k}^{c}$ but \textit{not} in $I_{k,j}^{c}$. Denote
by $\hat{\gamma}_{1}^{k,j}$ and $\hat{\gamma}_{0}^{k,j}$ the estimated
functional forms. Their predicted values for observations in $I_{k,j}^{c}$
are denoted as $\hat{\gamma}_{1}^{k,j}(X_{i}),\hat{\gamma}_{0}^{k,j}(X_{i}),i\in I_{k,j}^{c}.$
\item[(b)] For each $\lambda\in\Lambda$, calculate
\[
\hat{a}_{1}^{k,j}=\left(\hat{G}_{1}^{k,j}\hat{G}_{1}^{k,j}+\lambda\hat{G}_{1}^{k,j}\right)^{-}\hat{G}_{1}^{k,j}\hat{P}^{k,j},
\]
where 
\begin{align*}
\hat{G}_{1}^{k,j} & =\frac{\sum_{i\in I_{k}^{c},i\notin I_{k,j}^{c}}\left[D_{i}b(X_{i})b^{\prime}(X_{i})\right]}{\sum_{i\in I_{k}^{c},i\notin I_{k,j}^{c}}\mathbf{1}\left\{ i\in I_{k}^{c},i\notin I_{k,j}^{c}\right\} },\\
\hat{P}^{k,j} & =\frac{\sum_{i\in I_{k}^{c},i\notin I_{k,j}^{c}}\left[2\left(\hat{\gamma}_{1}^{k,j}(X_{i})-\hat{\gamma}_{0}^{k,j}(X_{i})\right)b(X_{i})\right]}{\sum_{i\in I_{k}^{c},i\notin I_{k,j}^{c}}\mathbf{1}\left\{ i\in I_{k}^{c},i\notin I_{k,j}^{c}\right\} },
\end{align*}
and
\[
\hat{a}_{0}^{k,j}=\left(\hat{G}_{0}^{k,j}\hat{G}_{0}^{k,j}+\lambda\hat{G}_{0}^{k,j}\right)^{-}\hat{G}_{1}^{k,j}\hat{P}^{k,j},
\]
where 
\begin{align*}
\hat{G}_{0}^{k,j} & =\frac{\sum_{i\in I_{k}^{c},i\notin I_{k,j}^{c}}\left[(1-D_{i})b(X_{i})b^{\prime}(X_{i})\right]}{\sum_{i\in I_{k}^{c},i\notin I_{k,j}^{c}}\mathbf{1}\left\{ i\in I_{k}^{c},i\notin I_{k,j}^{c}\right\} }.
\end{align*}
\item[(c)] Then, for each observation $i$ in $I_{k,j}^{c}$, calculate
\begin{equation*}
\hat{\omega}_{1}^{k,j}(\lambda,X_{i})  =\left[\hat{a}_{1}^{k,j}\right]^{\prime}b(X_{i}),\quad \hat{\omega}_{0}^{k,j}(\lambda,X_{i})  =\left[\hat{a}_{0}^{k,j}\right]^{\prime}b(X_{i}).
\end{equation*}
\item[(d)] For $\omega_{1}$, the cross-validation error in fold $I_{k,j}^{c}$
is 
\begin{align*}
CV_{1}(\lambda,j,k)= & \sum_{i\in I_{k,j}^{c}}\left[D_{i}\hat{\omega}_{1}^{k,j}(\lambda,X_{i})Y_{i}-2\left(\hat{\gamma}_{1}^{k,j}(X_{i})-\hat{\gamma}_{0}^{k,j}(X_{i})\right)\hat{\gamma}_{1}^{k,j}(X_{i})\right]^{2},
\end{align*}
and for $\omega_{0}$, the cross-validation error for fold $I_{k,j}^{c}$
is
\begin{align*}
CV_{0}(\lambda,j,k)= & \sum_{i\in I_{k,j}^{c}}\left[(1-D_{i})\hat{\omega}_{0}^{k,j}(\lambda,X_{i})Y_{i}-2\left(\hat{\gamma}_{1}^{k,j}(X_{i})-\hat{\gamma}_{0}^{k,j}(X_{i})\right)\hat{\gamma}_{0}^{k,j}(X_{i})\right]^{2}.
\end{align*}
\end{itemize}
\item For $\omega_{1}$, the total cross-validation error across all folds
$j=1\ldots 10$ is
$
TCV_{1}(\lambda,k)=\sum_{j=1}^{10}CV(\lambda,j,k)$, and analogously for $\omega_{0}$, 
$
TCV_{0}(\lambda,k)=\sum_{j=1}^{10}CV(\lambda,j,k)$.

\item Let $\hat{\lambda}_{k}^{1}$ solve $\min_{\lambda\in\Lambda}TCV_{1}(\lambda,k)$.
The fitted value of the cross validated estimator for $\omega_{1}$
(constructed using data from $I_{k}^{c}$) for each observation in
fold $I_{k}$ is then
\begin{align*}
\hat{\omega}_{1}(X_{i}) & =b^{\prime}(X_{i})\hat{a}_{1}^{k},\text{ for all }i\in I_{k},\\
\hat{a}_{1}^{k} & =\left(\hat{G}_{1}^{k}\hat{G}_{1}^{k}+\hat{\lambda}_{k}^{1}\hat{G}_{1}^{k}\right)^{-}\hat{G}_{1}^{k}\hat{P}^{k},\\
\hat{G}^{k} & =\frac{\sum_{i\in I_{k}^{c}}\left[D_{i}b(X_{i})b(X_{i}^{\prime})\right]}{\sum_{i\in I_{k}^{c}}\mathbf{1}\left\{ i\in I_{k}^{c}\right\} },\\
\hat{P}^{k} & =\frac{\sum_{i\in I_{k}^{c}}\left[2\left(\hat{\gamma}_{1}^{k}(X_{i})-\hat{\gamma}_{0}^{k}(X_{i})\right)b(X_{i})\right]}{\sum_{i\in I_{k}^{c}}\mathbf{1}\left\{ i\in I_{k}^{c}\right\} }.
\end{align*}
\item Let $\hat{\lambda}_{k}^{0}$ solve $\min_{\lambda\in\Lambda}TCV_{0}(\lambda,k)$.
The fitted value of the cross validated estimator for $\omega_{0}$
(constructed using data from $I_{k}^{c}$) for each observation in
fold $I_{k}$ is then
\begin{align*}
\hat{\omega}_{0}(X_{i}) & =b^{\prime}(X_{i})\hat{a}_{0}^{k},\text{ for all }i\in I_{k},\\
\hat{a}_{0}^{k} & =\left(\hat{G}_{0}^{k}\hat{G}_{0}^{k}+\hat{\lambda}_{k}^{0}\hat{G}_{0}^{k}\right)^{-}\hat{G}_{0}^{k}\hat{P}^{k},\\
\hat{G}_{0}^{k} & =\frac{\sum_{i\in I_{k}^{c}}\left[(1-D_{i})b(X_{i})b^{\prime}(X_{i})\right]}{\sum_{i\in I_{k}^{c}}\mathbf{1}\left\{ i\in I_{k}^{c}\right\} },\\
\hat{P}^{k} & =\frac{\sum_{i\in I_{k}^{c}}\left[2\left(\hat{\gamma}_{1}^{k}(X_{i})-\hat{\gamma}_{0}^{k}(X_{i})\right)b(X_{i})\right]}{\sum_{i\in I_{k}^{c}}\mathbf{1}\left\{ i\in I_{k}^{c}\right\} }.
\end{align*}
\item 
For each $i\in I_{k}$, the debiased weight is
\begin{align*}
\hat{\xi}(Z_i) & =\left[\hat{\gamma}_{1}^{k}(X_{i})-\hat{\gamma}_{0}^{k}(X_{i})\right]^{2}\\
 & +\left[D_{i}\hat{\omega}_{1}^{k}(X_{i})(Y_{i}-\hat{\gamma}_{1}^{k}(X_{i}))-\left(1-D_{i}\right)\hat{\omega}_{0}^{k}(X_{i})(Y_{i}-\hat{\gamma}_{0}^{k}(X_{i}))\right].
\end{align*}
\end{enumerate}

\section{Additional results for Appendix \ref{sec:lower.bound}}

\subsection{Proof of Lemma \ref{lem:Bayes_classification}}
Let $\mathbf{y}=(y_{1},\ldots,y_{N})^{\prime}$ be a realization of
$\mathbf{Y}$. For each $\mathbf{y}$ in the support of $\mathbf{Y}$,
the Bayes decision is
\[
g^{*}(\mathbf{y})=\begin{cases}
1 & Pr\{\mathrm{C}=1\mid\mathbf{y}\}\geq\frac{1}{2},\\
-1 & Pr\{\mathrm{C}=1\mid\mathbf{y}\}<\frac{1}{2},
\end{cases}
\]
and $Pr\left\{ g^{*}(\mathbf{Y})\neq\mathrm{C}\right\} =\mathbb{E}\left[\min\left\{ Pr\{\mathrm{C}=1\mid\mathbf{\mathbf{Y}}\},1-Pr\{\mathrm{C}=1\mid\mathbf{\mathbf{Y}}\}\right\} \right]$,
where $\mathbb{E}[\cdotp]$ is with respect to the marginal distribution
of $\mathbf{\mathbf{\mathbf{Y}}}$. Applying Bayes rule yields
\[
\begin{aligned} & Pr\{\mathrm{C}=1\mid\mathbf{y}\}
=  \frac{f\{\mathbf{y}\mid\mathrm{C}=1\}}{f\{\mathbf{y}\mid\mathrm{C}=1\}+f\{\mathbf{y}\mid\mathrm{C}=-1\}},
\end{aligned}
\]
where $f\{\mathbf{y}\mid\mathrm{C}=1\}$ is the pdf of $\mathbf{Y}$
given $\mathrm{C}=1$, and $f\{\mathbf{y}\mid\mathrm{C}=-1\}$ is
the pdf of $\mathbf{Y}$ given $\mathrm{C}=-1$. Algebra shows 
\begin{align*}
f\{\mathbf{y}\mid\mathrm{C}=1\}) & =\Pi_{j=1}^{N}f(y_{j}\mid\mathrm{C}=1)\\
 & =\begin{cases}
1, & \text{\ensuremath{y_{j}}}\in\left[\text{ \ensuremath{\frac{u_{j}-1}{2}}},\text{ \ensuremath{\frac{u_{j}+1}{2}}}\right],j=1\ldots N,\\
0, & \text{otherwise},
\end{cases}
\end{align*}
and
\begin{align*}
f\{\mathbf{y}\mid\mathrm{C}=-1\}) & =\Pi_{j=1}^{N}f(y_{j}\mid\mathrm{C}=-1)\\
 & =\begin{cases}
1, & \text{\ensuremath{y_{j}}}\in\left[\text{ \ensuremath{\frac{\mathrm{-}u_{j}-1}{2}}},\text{ \ensuremath{\frac{-u_{j}+1}{2}}}\right],j=1\ldots N,\\
0, & \text{otherwise}.
\end{cases}
\end{align*}
Therefore, 
\begin{align*}
 & \min\left\{ Pr\{\mathrm{C}=1\mid\mathbf{\mathbf{Y}}\},1-Pr\{\mathrm{C}=1\mid\mathbf{\mathbf{Y}}\}\right\} \\
= & \begin{cases}
\frac{1}{2}, & \text{\ensuremath{y_{j}}}\in\left[\text{ \ensuremath{\frac{u_{j}-1}{2}}},\frac{-u_{j}+1}{2}\right],j=1,\ldots N,\\
0, & \text{otherwise}.
\end{cases}
\end{align*}
It follows 
\begin{align*}
Pr\left\{ g^{*}(\mathbf{Y})\neq\mathrm{C}\right\}  & =\frac{1}{2}Pr\left\{ \text{\ensuremath{Y_{j}}}\in\left[\text{ \ensuremath{\frac{u_{j}-1}{2}}},\frac{-u_{j}+1}{2}\right],j=1,\ldots N\right\} \\
 & =\frac{1}{2}\prod_{i=1}^{N}Pr\left\{ \text{\ensuremath{Y_{j}}}\in\left[\text{ \ensuremath{\frac{u_{j}-1}{2}}},\frac{-u_{j}+1}{2}\right]\right\} 
\end{align*}
where for all $j=1\ldots N$,
\begin{align*}
 & Pr\left\{ \text{\ensuremath{Y_{j}}}\in\left[\text{ \ensuremath{\frac{u_{j}-1}{2}}},\frac{-u_{j}+1}{2}\right]\right\} \\
= & \frac{1}{2}Pr\left\{ \text{\ensuremath{Y_{j}}}\in\left[\text{ \ensuremath{\frac{u_{j}-1}{2}}},\frac{-u_{j}+1}{2}\right]\mid C=1\right\} \\
+ & \frac{1}{2}Pr\left\{ \text{\ensuremath{Y_{j}}}\in\left[\text{ \ensuremath{\frac{u_{j}-1}{2}}},\frac{-u_{j}+1}{2}\right]\mid C=-1\right\} \\
= & 1-u_{j}.
\end{align*}
We then conclude 
\begin{align*}
Pr\left\{ g^{*}(\mathbf{Y})\neq\mathrm{C}\right\}  & =\prod_{j=1}^{N}\left(1-u_{j}\right) \geq\left(1-\max_{j\in\left\{ 1,\ldots,N\right\} }u_{j}\right)^{N}\\
 & \geq1-N\max_{j\in\left\{ 1,\ldots,N\right\} }u_{j}\geq1-N\left\Vert u\right\Vert .
\end{align*}

\subsection{Proof of Theorem \ref{thm:lower.bound}}
\paragraph*{Step 1}
We construct (depending on $n$) a subclass of distributions of $(Y(1),Y(0),D,X)$
contained in $\mathcal{P}(s,C)$, as follows:
\begin{itemize}
\item[(i)] $X=(X_{1},W^{\prime})^{\prime}$, where $W$ is uniformly distributed
in $[0,1]^{d_{W}}$, and
\item[(ii)] $X_{1}=\text{sgn}\left(m(w)+\varepsilon\right)$ for all $w\in[0,1]^{d_{W}}$,
where $m:\mathbb{R}^{d_{W}}\rightarrow[0,1]$, $m\in\mathcal{F}^{\mathcal{C}_{n}}\subset\mathcal{F}(s,C)$
is defined shortly in step 2, and $\varepsilon\sim\text{uniform}[-1,0]$
is independent of $W$,
\item[(iii)] $\tau(x_{1},w)=x_{1}$ for all $x_{1}\in\left\{ -1,1\right\} $, $w\in[0,1]^{d_{W}}$,
\item[(iv)] The joint distribution of $(Y(1),Y(0),D,X)$ satisfies Assumption
\ref{asm:unconfounded}. Moreover, the functional form of $\pi(x)$ is independent of any
parameters in $\mathcal{F}^{\mathcal{C}_{n}}$. 
\end{itemize}
Denote by $\mathcal{P}^{\mathcal{C}_{n}}$ the class of distributions
above. 

\paragraph*{Step 2 }

We now construct $\mathcal{F}^{\mathcal{C}_{n}}$ appeared in step
1. Let $M_{n}=\left\lceil \left(C^{\frac{2}{3}}n\right)^{\frac{3}{(2s+d_{W})}}\right\rceil $.
Partition $[0,1]^{d_{W}}$ into $S_{n}:=\left\lceil n^{2}M_{n}^{\frac{d_{W}}{3}}\right\rceil $
cubes $\left\{ A_{n,j}\right\} _{j=1}^{S_{n}}$, each with side length
$\frac{1}{S_{n}}$ and with centers $a_{n,j}$, $j=1\ldots S_{n}$.
Choose a function $\bar{g}:\mathbb{R}^{d_{W}}\rightarrow\mathbb{R}$
such that the support of $\bar{g}$ is a subset of $\left[-\frac{1}{2},\frac{1}{2}\right]^{d_{W}}$,
$\bar{g}\in\mathcal{F}(s,2^{t-1})$, and $\bar{g}(w)\in[0,C^{-1}]$
for all $w$. Define $g:\mathbb{R}^{d_{W}}\rightarrow\mathbb{R}$
as $g(w)=C\cdotp\bar{g}(w)$, which possesses the following properties:
\begin{itemize}
\item[(a)] the support of $g$ is a subset of $\left[-\frac{1}{2},\frac{1}{2}\right]^{d_{W}}$;
\item[(b)] $\int g^{2}(w)dw=C^{2}\int\bar{g}^{2}(w)dw$ and $\int\bar{g}^{2}(w)dw>0$;
\item[(c)] $g\in\mathcal{F}^{(s,C\cdotp2^{t-1})}$, and $g(w)\in[0,1]$ for all
$w$.
\end{itemize}
Let $c_{n}=(c_{n,1},c_{n,2}\ldots c_{n,S_{n}})^{\prime}$ be a vector
of $+1$ or $-1$ components. Denote by $\mathcal{C}_{n}$ the set
of all such vectors. For $c_{n}\in\mathcal{C}_{n}$, consider the
following function:
\begin{align*}
m^{(c_{n})}(w) & =\frac{1}{2}+\frac{1}{2}\sum_{j=1}^{S_{n}}c_{n,j}g_{n,j}(w),
\end{align*}
where
$
g_{n,j}(w)=M_{n}^{-s}g(M_{n}(w-a_{n,j}))$.
As $g(w)\in[0,1]$ for all $w$ and $M_{n}\geq1$, it follows $0<M_{n}^{-s}\leq1$
and $m^{(c_{n})}(w)\in[0,1]$ for all $w$. Define $\mathcal{F}^{\mathcal{C}_{n}}:=\left\{ m^{(c_{n})}(w):\mathbb{R}^{d_{W}}\rightarrow[0,1]\mid c_{n}\in\mathcal{C}_{n}\right\} $.
We verify below that $\mathcal{F}^{\mathcal{C}_{n}}\subset\mathcal{F}^{(s,C)}$.  Pick each $m^{(c_{n})}\in\mathcal{F}^{\mathcal{C}_{n}}$. Consider
any $\alpha=(\alpha_{1},\ldots,\alpha_{d_{W}})$ such that $\alpha_{i}\in\mathbb{N}_{+}$,
$\sum_{j=1}^{d_{W}}\alpha_{j}=s_{0}$. One may verify that for all
$w,z\in[0,1]^{d_{W}}$, the following holds:
\[
\left|\frac{\partial^{s_{0}}m^{(c_{n})}}{\partial w_{1}^{\alpha_{1}}\cdots\partial w_{d_{W}}^{\alpha_{d_{W}}}}(w)-\frac{\partial^{s_{0}}m^{(c_{n})}}{\partial w_{1}^{\alpha_{1}}\cdots\partial w_{d_{W}}^{\alpha_{d_{W}}}}(z)\right|\leq C\left\Vert x-z\right\Vert ^{t},x,z\in[0,1]^{d_{W}}.
\]

\paragraph*{Step 3}

For an arbitrary rule $\delta_{n}$, we show that its excess risk
can be lowered bounded by classification problem. First note that,
for all $P\in\mathcal{P}(s,C)$, we have 
\begin{align*}
L(\delta_{n},\tau) & =\mathbb{E}\left[\tau^{2}(X)\left(\mathbf{1}\left\{ \tau(X)\geq0\right\} -\delta_{n}(W)\right)^{2}\right] =\mathbb{E}\left[\left(\mathbf{1}\left\{ X_{1}\geq0\right\} -\delta_{n}(W)\right)^{2}\right].
\end{align*}
As
a result, 
$\delta^{*}(w)=Pr\left\{ X_{1}=1\mid W=w\right\} =m(w)$, and
\[
L(\delta_{n},\tau)-L(\delta^{*},\tau)=\int\left(\delta_{n}(w)-m(w)\right)^{2}dF_{W}(w).
\]
Then, we have
\begin{align*}
 & \sup_{P\in\mathcal{P}(s,c)}\mathbb{E}_{P_{n}}\left[L(\delta_{n},\tau)-L(\delta^{*},\tau)\right]\\
\geq & \sup_{P\in\mathcal{P}^{\mathcal{C}_{n}}}\mathbb{E}_{P_{n}}\left[L(\delta_{n},\tau)-L(\delta^{*},\tau)\right]\\
= & \sup_{m^{(c_{n})}\in\mathcal{F}^{\mathcal{C}_{n}}}\mathbb{E}_{P_{n}}\left[\int\left(\delta_{n}(w)-m^{(c_{n})}(w)\right)^{2}dF_{W}(w)\right].
\end{align*}
Furthermore, we can write each $\delta_{n}$ as $\delta_{n}=\frac{1}{2}+\frac{1}{2}m_{n}$
for some $m_{n}\in[-1,1]$, and we can also write $m^{(c_{n})}(w)=\frac{1}{2}+\frac{1}{2}\tilde{m}^{(c_{n})}(w),$where
$\tilde{m}^{(c_{n})}(w)=\sum_{j=1}^{S_{n}}c_{n,j}g_{n,j}(w)\in[-1,1].$
Therefore,
\begin{align*}
 & \sup_{m^{(c_{n})}\in\mathcal{F}^{\mathcal{C}_{n}}}\mathbb{E}_{P_{n}}\left[\int\left(\delta_{n}(w)-m^{(c_{n})}(w)\right)^{2}dF_{W}(w)\right]\\
= & \frac{1}{4}\sup_{m^{(c_{n})}\in\mathcal{F}^{\mathcal{C}_{n}}}\mathbb{E}_{P_{n}}\left[\int\left(m_{n}(w)-\tilde{m}^{(c_{n})}(w)\right)^{2}dF_{W}(w)\right].
\end{align*}
Denote by $\hat{m}_{n}(x)=\sum_{j=1}^{S_{n}}\hat{c}_{n,j}g_{n,j}(w)$
the least squares projection of $m_{n}$ onto $\left\{ \tilde{m}^{(c_{n})}:c_{n}\in\mathcal{C}_{n}\right\} $,
which we note is an orthogonal system. Then,
\begin{align*}
 & \int\left(m_{n}(w)-\tilde{m}^{(c_{n})}(w)\right)^{2}dF_{W}(w)\\
\geq & \int\left(\hat{m}_{n}(w)-\tilde{m}^{(c_{n})}(w)\right)^{2}dF_{W}(w)\\
= & \sum_{j=1}^{S_{n}}\int_{A_{n,j}}\left(\hat{c}_{n,j}-c_{n,j}\right)^{2}g_{n,j}^{2}(w)dw\\
= & M_{n}^{-2s}M_{n}^{-d_{W}}\left(\sum_{j=1}^{S_{n}}\left(\hat{c}_{n,j}-c_{n,j}\right)^{2}\right)\int g^{2}(w)dw.
\end{align*}
Let $\tilde{c}_{n,j}=1$ if $\hat{c}_{n,j}\geq0$ and $-1$ otherwise.
As 
$\left|\hat{c}_{n,j}-c_{n,j}\right|\geq\frac{\left|\tilde{c}_{n,j}-c_{n,j}\right|}{2}$,
we have 
\[
\sum_{j=1}^{S_{n}}\left(\hat{c}_{n,j}-c_{n,j}\right)^{2}\geq\frac{1}{4}\sum_{j=1}^{S_{n}}\left(\tilde{c}_{n,j}-c_{n,j}\right)^{2}=\sum_{j=1}^{S_{n}}\mathbf{1}\left\{ \tilde{c}_{n,j}\neq c_{n,j}\right\} .
\]
In summary, we have 
\begin{align*}
 & \sup_{m^{(c_{n})}\in\mathcal{F}^{\mathcal{C}_{n}}}\mathbb{E}_{P_{n}}\left[\int\left(m_{n}(w)-\tilde{m}^{(c_{n})}(w)\right)^{2}dF_{W}(w)\right]\\
\geq & M_{n}^{-(2s+d_{W})}\int g^{2}(w)d(w)\sup_{m^{(c_{n})}\in\mathcal{F}^{\mathcal{C}_{n}}}\mathbb{E}_{P_{n}}\left[\sum_{j=1}^{S_{n}}\mathbf{1}\left\{ \tilde{c}_{n,j}\neq c_{n,j}\right\} \right]\\
\geq & M_{n}^{-(2s+d_{W})}S_{n}C^{2}\int\bar{g}^{2}(x)dx\sup_{m^{(c_{n})}\in\mathcal{F}^{\mathcal{C}_{n}}}\mathbb{E}_{P_{n}}\left[\frac{1}{S_{n}}\left(\sum_{j=1}^{S_{n}}\mathbf{1}\left\{ \tilde{c}_{n,j}\neq c_{n,j}\right\} \right)\right],
\end{align*}
where note 
\[
\liminf_{n\rightarrow\infty}\frac{M_{n}^{-(2s+d_{W})}S_{n}C^{2}}{C^{\frac{2d_{W}}{6s+3d_{W}}}n^{-\frac{2s}{(2s+d_{W})}}}>0,\int\bar{g}^{2}(x)dx>0.
\]
Thus, it suffices to show
\[
\liminf_{n\rightarrow\infty}\sup_{m^{(c_{n})}\in\mathcal{F}^{\mathcal{C}_{n}}}\left[\frac{1}{S_{n}}\left(\sum_{j=1}^{S_{n}}Pr\left\{ \tilde{c}_{n,j}\neq c_{n,j}\right\} \right)\right]>0
\]
and does not depend on $C$. To this end, let $C_{n,1},\ldots C_{n,M_{n}^{d}}$
be a sequence of iid random variables independent of data $Z^{n}=\left\{ Y_{i},D_{i},X_{1i},W_{i}\right\} _{i}^{n}$,
and satisfy $Pr\{C_{n,1}=1\}=Pr\{C_{n,1}=-1\}=1/2$. Let $C_{n}:=(C_{n,1},\ldots,C_{n,S_{n}})$.
Then,
\begin{align*}
\sup_{m^{(c_{n})}\in\mathcal{F}^{\mathcal{C}_{n}}}\left[\frac{1}{S_{n}}\left(\sum_{j=1}^{S_{n}}Pr\left\{ \tilde{c}_{n,j}\neq c_{n,j}\right\} \right)\right]\geq\frac{1}{S_{n}}\left(\sum_{j=1}^{S_{n}}Pr\left\{ \tilde{c}_{n,j}\neq C_{n,j}\right\} \right),
\end{align*}
where $\tilde{c}_{n,j}$ can be viewed as a binary decision using
data $Z^{n}$ on $C_{n,j}$. For each $j=1,\ldots,S_{n}$:
\[
Pr\left\{ \tilde{c}_{n,j}\neq C_{n,j}\right\} =\mathbb{E}_{P_{n}}\left[Pr\left\{ \tilde{c}_{n,j}\neq C_{n,j}\mid Y_{i},D_{i},W_{i},i=1,\ldots,n\right\} \right].
\]
Denote by $\left\{ \left\{ Y_{i_{1}},D_{i_{1}},W_{i_{1}}\right\} ,\ldots,\left\{ Y_{i_{l}},D_{i_{l}},W_{i_{l}}\right\} \right\} $
those units such that $W_{i}\in A_{n,j}$. Note 
$X_{1i_{q}}=\text{sgn}\left(V_{1i_{q}}\right)$,
where 
$
\text{ \ensuremath{V_{1i_{q}}}=\ensuremath{\frac{1}{2}}+\ensuremath{\frac{1}{2}C_{n,j}}}g_{n,j}(W_{i_{q}})+\varepsilon_{i_{q}}$
for all $q=1\ldots l$, and $\left\{ \varepsilon_{i_{q}}\right\} _{q=1\ldots l}$
are iid and uniformly distributed in $[-1,0]$. In addition, $\left((X_{1i,i=1\ldots n}\setminus(X_{1i_{1}},\ldots,X_{1i_{l}})\right)$
depends only on $C\setminus\left\{ C_{n,j}\right\} $ and on $W_{i}$
for those $i\notin\left\{ i_{1},\ldots i_{l}\right\} $. Therefore,
the following must hold:
\begin{align*}
 & Pr\left\{ \tilde{c}_{n,j}\neq C_{n,j}\mid Y_{i},D_{i},W_{i},i=1,\ldots,n\right\} \\
\geq & Pr\left\{ C_{n,j}^{B}\neq C_{n,j}\mid Y_{i},D_{i},W_{i},i=1,\ldots,n\right\} ,
\end{align*}
where $C_{n,j}^{B}$ is the conditional Bayes decision for $C_{n,j}$
that only uses data $\left\{ W_{i_{1}}\ldots W_{i_{l}}\right\} $
and $\left\{ V_{i_{1}},\ldots,V_{i_{l}}\right\} $. 

\paragraph*{Step 4}

We derive the Bayes risk of the conditional Bayes rule $C_{n,j}^{B}$
defined in Step 3. Applying Lemma \ref{lem:Bayes_classification}, we have 
\begin{align*}
 & Pr\left\{ C_{n,j}^{B}\neq C_{n,j}\mid Y_{i},D_{i},W_{i},i=1,\ldots,n\right\} \\
\geq & 1-l\sqrt{\sum_{q=1}^{l}g_{n,j}^{2}(W_{i_{q}})}
\geq  1-n\sqrt{\sum_{i=1}^{n}g_{n,j}^{2}(W_{i})}.
\end{align*}
 Then, law of iterated expectations and Jensen's inequality together
yields
\begin{align*}
Pr\left\{ C_{n,j}^{B}\neq C_{n,j}\right\}  & =\mathbb{E}\left\{ Pr\left\{ C_{n,j}^{B}\neq C_{n,j}\mid Y_{i},D_{i},W_{i},i=1,\ldots,n\right\} \right\} \\
 & \geq1-n\sqrt{n\mathbb{E}\left[g_{n,j}^{2}(W)\right]}\\
 & =1-\sqrt{n^{3}M_{n}^{-\left(2s+d_{W}\right)}\int g^{2}(x)dx}\\
 & \geq1-\sqrt{\int\bar{g}^{2}(x)dx}>0.
\end{align*}

\end{document}